\PassOptionsToPackage{unicode}{hyperref}
\PassOptionsToPackage{hyphens}{url}
\PassOptionsToPackage{dvipsnames,svgnames,x11names}{xcolor}
\documentclass[12pt]{article}

\usepackage{amsmath,amssymb}
\usepackage{iftex}
\usepackage{bm}
\usepackage{tikz}
\usepackage[ruled,vlined]{algorithm2e}
\usepackage{booktabs}   
\usepackage{longtable}  
\usepackage{array}      
\usepackage{multirow}   
\usepackage{colortbl}   
\usepackage{xcolor}     
\usepackage{caption}    
\usepackage{upquote}
\usepackage{natbib}
\usepackage{amsthm}
\usepackage{mathrsfs}
\usepackage{amsmath,amssymb}
\usepackage{iftex}

\IfFileExists{upquote.sty}{\usepackage{upquote}}{}
\IfFileExists{microtype.sty}{
  \usepackage[]{microtype}
  \UseMicrotypeSet[protrusion]{basicmath} 
}{}
\makeatletter
\@ifundefined{KOMAClassName}{
  \IfFileExists{parskip.sty}{%
    \usepackage{parskip}
  }{
    \setlength{\parindent}{0pt}
    \setlength{\parskip}{6pt plus 2pt minus 1pt}}
}{
  \KOMAoptions{parskip=half}}
\makeatother
\usepackage{xcolor}
\makeatletter
\ifx\paragraph\undefined\else
  \let\oldparagraph\paragraph
  \renewcommand{\paragraph}{
    \@ifstar
      \xxxParagraphStar
      \xxxParagraphNoStar
  }
  \newcommand{\xxxParagraphStar}[1]{\oldparagraph*{#1}\mbox{}}
  \newcommand{\xxxParagraphNoStar}[1]{\oldparagraph{#1}\mbox{}}
\fi
\ifx\subparagraph\undefined\else
  \let\oldsubparagraph\subparagraph
  \renewcommand{\subparagraph}{
    \@ifstar
      \xxxSubParagraphStar
      \xxxSubParagraphNoStar
  }
  \newcommand{\xxxSubParagraphStar}[1]{\oldsubparagraph*{#1}\mbox{}}
  \newcommand{\xxxSubParagraphNoStar}[1]{\oldsubparagraph{#1}\mbox{}}
\fi
\makeatother

\usepackage{longtable,booktabs,array}
\usepackage{calc} 
\usepackage{etoolbox}
\makeatletter
\patchcmd\longtable{\par}{\if@noskipsec\mbox{}\fi\par}{}{}
\makeatother
\IfFileExists{footnotehyper.sty}{\usepackage{footnotehyper}}{\usepackage{footnote}}
\makesavenoteenv{longtable}
\usepackage{graphicx}
\makeatletter
\def\maxwidth{\ifdim\Gin@nat@width>\linewidth\linewidth\else\Gin@nat@width\fi}
\def\maxheight{\ifdim\Gin@nat@height>\textheight\textheight\else\Gin@nat@height\fi}
\makeatother
\setkeys{Gin}{width=\maxwidth,height=\maxheight,keepaspectratio}
\makeatletter
\def\fps@figure{htbp}
\makeatother

\makeatletter
\@ifpackageloaded{caption}{}{\usepackage{caption}}
\AtBeginDocument{%
\ifdefined\contentsname
  \renewcommand*\contentsname{Table of contents}
\else
  \newcommand\contentsname{Table of contents}
\fi
\ifdefined\listfigurename
  \renewcommand*\listfigurename{List of Figures}
\else
  \newcommand\listfigurename{List of Figures}
\fi
\ifdefined\listtablename
  \renewcommand*\listtablename{List of Tables}
\else
  \newcommand\listtablename{List of Tables}
\fi
\ifdefined\figurename
  \renewcommand*\figurename{Figure}
\else
  \newcommand\figurename{Figure}
\fi
\ifdefined\tablename
  \renewcommand*\tablename{Table}
\else
  \newcommand\tablename{Table}
\fi
}
\@ifpackageloaded{float}{}{\usepackage{float}}
\floatstyle{ruled}
\@ifundefined{c@chapter}{\newfloat{codelisting}{h}{lop}}{\newfloat{codelisting}{h}{lop}[chapter]}
\floatname{codelisting}{Listing}

\makeatother
\makeatletter
\@ifpackageloaded{caption}{}{\usepackage{caption}}
\@ifpackageloaded{subcaption}{}{\usepackage{subcaption}}
\makeatother
\usepackage{bookmark}

\IfFileExists{xurl.sty}{\usepackage{xurl}}{} 
\hypersetup{
  pdftitle={Title},
  pdfauthor={Author 1; Author 2},
  pdfkeywords={3 to 6 keywords, that do not appear in the title},
  colorlinks=true,
  linkcolor={blue},
  filecolor={Maroon},
  citecolor={Blue},
  urlcolor={Blue},
  pdfcreator={LaTeX via pandoc}}

\newcommand{\anon}{1}

\newtheorem{proposition}{Proposition}
\newtheorem{theorem}{Theorem}
\newcommand{\iid}{\overset{iid}{\sim}}
\newcommand{\R}{\mathbb{R}}
\newcommand{\bd}{d}
\newcommand{\KL}{\operatorname{KL}}

\newcommand{\DenX}[1]{\mathscr{D}(#1)}

\begin{document}

\def\spacingset#1{\renewcommand{\baselinestretch}%
{#1}\small\normalsize} \spacingset{1}


\if1\anon
{
  \title{\bf NP-LEAP: Nonparametric Latent Exchangeability Prior for Model-Lean Borrowing from Historical Data}
  %
  \author{Ethan M. Alt, \hspace{.2cm}\\
    Quantitative Sciences Innovation, GSK
    \\ \\
    Miheer Dewaskar, \hspace{.2cm} \\ Department of Mathematics and Statistics, 
    University of New Mexico 
    \\ \\
    Jacob M. Maronge, Yuelin Lu, and Matthew A. Psioda, \hspace{.2cm} \\
    Quantitative Sciences Innovation, GSK
  }
  \maketitle
} \fi

\if0\anon
{
  \bigskip
  \bigskip
  \bigskip
  \begin{center}
    {\LARGE\bf NP-LEAP: Nonparametric Latent Exchangeability Prior for Model-Lean Borrowing from Historical Data}
\end{center}
  \medskip
} \fi

\bigskip
\begin{abstract}
Bayesian dynamic borrowing (BDB) methods leverage historical data to reduce treatment effect uncertainty, yet existing approaches rely on parametric outcome models susceptible to misspecification. We propose the nonparametric latent exchangeability prior (NP-LEAP), an outcome-agnostic, assumption-lean framework to borrow information from historical data. The NP-LEAP performs individual-level exchangeability assessment, inducing Bayesian model averaging over all possible partitions of the historical data into exchangeable and nonexchangeable subsets. Although applicable to a variety of data types with choice of appropriate kernel, the NP-LEAP is particularly well-suited for studies with time-to-event outcomes, where parametric BDB is potentially triply misspecified — imposing a parametric baseline hazard, the proportional hazards structure, and blanket exchangeability. We establish posterior consistency under mild regularity conditions. Simulation studies demonstrate favorable operating characteristics relative to parametric borrowing methods and nonborrowing semiparametric frequentist methods. We illustrate the method by augmenting the control arm in a randomized trial of patients with non-small cell lung cancer.
\end{abstract}

\noindent%
{\textit Keywords:} Bayesian nonparametrics; Dynamic borrowing; Historical data; Exchangeability; Density estimation
\vfill

\newpage
\spacingset{1.8} 


\section{Introduction}
\label{sec:intro}

Bayesian dynamic borrowing (BDB) methods leverage historical data to reduce
treatment effect uncertainty in current studies, yet existing approaches almost
universally rely on parametric outcome models.
This reliance is problematic across all outcome types — continuous, binary, longitudinal, and time-to-event — because borrowing under a misspecified model can yield biased inference even when populations are genuinely exchangeable. The ideal solution is therefore outcome-agnostic: a borrowing framework that
accommodates any outcome type with minimal modeling assumptions. This paper introduces the \textit{nonparametric latent exchangeability prior}
(NP-LEAP), which achieves this by operating directly in the space of outcome
densities via Bayesian nonparametric modeling, with the choice of kernel
determining applicability to continuous, discrete, or censored outcomes.

Borrowing information in a misspecified model is perhaps most acute for time-to-event outcomes, where parametric BDB is
potentially \textit{triply} misspecified relative to standard semiparametric
practice: it imposes a parametric baseline hazard, a proportional hazards assumption, and assumes blanket exchangeability of all historical
observations, when in reality only a subset may be exchangeable with the current data. Each assumption is unverifiable in practice, and violation of any one induces systematic bias.

Most BDB methods are motivated by group-level exchangeability arguments.
Given data for $J$ groups $\{y_{ij} : j = 1, \ldots, J,\ i = 1, \ldots, n_j\}
\subseteq \R^d$, a standard model assumes $y_{ij} \sim F(\cdot | \bm{\theta}_j)$
with $\bm{\theta}_j \sim G(\cdot | \bm{\eta})$, where $F(\cdot|\bm{\theta}_j)$
is a parametric distribution for group $j$ and $G(\cdot|\bm{\eta})$ is a
hyperprior governing the degree of shrinkage (which may be degenerate).
This family includes power priors \citep{ibrahim2000power}, commensurate priors
\citep{hobbs2012commensurate}, and robust mixture priors
\citep{schmidli2014robust}.
Three limitations are shared across these approaches.
First, when $J$ is small---as in the common clinical trial setting of one
current study $D$ and one historical study $D_0$---the hyperprior for
$\bm{\eta}$ must be strongly informative to reduce uncertainty, limiting
objectivity.
Second, these methods apply blanket discounting: all individuals in $D_0$ are
treated as equally relevant or irrelevant, with no individual-level assessment.
Third, parametric assumptions are retained throughout, and borrowing properties
under model misspecification have received limited attention.

\cite{alt2024leap} address the first two limitations through the latent
exchangeability prior (LEAP), which introduces individual-level latent
indicators $\epsilon_{0j} \in \{0,1\}$ to classify each historical observation
as exchangeable with the current data or not.
Formally, the LEAP assumes $\{y_i\}_{i=1}^n \overset{i.i.d.}{\sim}
F_1(\cdot|\bm{\theta}_1)$ and $\{y_{0j}\}_{j=1}^{n_0} \overset{i.i.d.}{\sim}
\sum_{k=1}^K \omega_k F_k(\cdot|\bm{\theta}_k)$ with prior $\pi(\bm{\theta}_1, \ldots, \bm{\theta}_K, \bm{\omega})$, inducing Bayesian model
averaging over all partitions of $D_0$ into exchangeable and non-exchangeable
subsets, with partition weights determined by marginal likelihoods.
Because inference is driven by sample size rather than number of datasets,
prior elicitation is more objective than in group-level approaches.
However, the LEAP retains parametric assumptions for the sampling model and
requires pre-specification of the number of mixture components $K$.
The NP-LEAP removes both constraints.

The NP-LEAP specifies a Dirichlet process mixture model (DPMM) for the exchangeable and nonexchangeable density components, with individual-level latent Bernoulli indicators classifying each historical observation as exchangeable with the current data or not (Section~\ref{ss:np-leap}). This yields Bayesian model averaging over all $2^{n_0}$ partitions of the historical data, with partition weights driven by marginal likelihoods rather than subjective tuning (Proposition~\ref{prop:leap_bma}). When covariates are available, the DPMM can be replaced by a dependent Dirichlet process, as in the time-to-event application of Section~\ref{sec:analysis}.

We establish posterior consistency of the NP-LEAP under mild regularity
conditions (Theorem~\ref{thm:consistency}).
Crucially, while we illustrate the framework with time-to-event outcomes---where
the tension between parametric borrowing and semiparametric practice is
sharpest---the NP-LEAP applies equally to continuous, discrete, or longitudinal
outcomes through appropriate specification of the kernels $f(\cdot|\bm{\theta})$
and $g(\cdot|\bm{\lambda})$.

\subsection{Related Work}
\label{ss:related}

The normalized power prior \citep[NPP;][]{ibrahim2000power,duan2006evaluating},
robust meta-analytic predictive prior \citep[RMAPP;][]{schmidli2014robust}, and
commensurate prior \citep[CP;][]{hobbs2012commensurate} are canonical BDB
approaches, each sharing the parametric modeling limitations described above.
Recent work has strengthened these priors through hyperparameter selection
\citep{shen2023optimal,demartino2025eliciting,egidi2022avoiding} and propensity
score adjustment
\citep{wang2019propensity,wang2022propensity,Psioda2025IPWDynamicBorrowing},
but parametric assumptions are retained throughout.

In time-to-event settings specifically, several frequentist approaches borrow
under the PH model \citep{liu2014estimating,huang2016efficient,li2023accommodating},
inheriting its structural constraints; since the Cox estimator converges to a
censoring-distribution-dependent quantity under misspecification
\citep{lin1989robust}, borrowing under misspecified PH models additionally
requires that censoring distributions be identical across studies.
\cite{gao2025doubly} propose a doubly robust procedure for the restricted mean
survival time, but this is estimand-specific; \citet{kwiatkowski2024case} offer
a Bayesian piecewise constant hazards approach but require pre-specification of
the time axis partition.

Regulatory agencies have been open to augmenting trial data with external
controls \citep{fda2023externally}, and a recent FDA draft guidance on
Bayesian methodology in clinical trials \citep{fda2026bayesian} explicitly
endorses dynamic discounting over static approaches on the basis of superior
operating characteristics under prior-data conflict, while identifying the
primary challenge of dynamic methods as the need to specify additional
parameters governing the similarity measure and rate at which borrowing
declines.
The NP-LEAP directly addresses this challenge: its nonparametric prior adapts
automatically to the degree of exchangeability without requiring
pre-specification of mixture components or a discounting rate, and its
individual-level exchangeability indicators provide precisely the robustness
to prior-data conflict that the guidance prioritizes.

\subsection{Contributions}
Our approach offers three key advantages over existing methods. First, nonparametric modeling yields assumption-lean inference, estimating the outcome density flexibly rather than restricting it to a parametric family. Second, modeling the non-exchangeable historical data component via a DPMM removes the need to prespecify the number of mixture components, as the Dirichlet process adapts automatically. Third, by modeling the full outcome density, NP-LEAP is estimand-agnostic: median survival, RMST, time-averaged hazard ratios, and milestone probabilities are all obtainable within a single unified model.

The remainder of the paper is organized as follows. Section~\ref{sec:methodology} introduces the NP-LEAP, covering background on the LEAP and DPMMs, the main model, its Bayesian model averaging interpretation (Proposition~\ref{prop:leap_bma}), its clustering scheme (Proposition~\ref{prop:cluster}), a posterior consistency theorem (Theorem~\ref{thm:consistency}) with proof sketch, and the dependent Dirichlet process (DDP; \citet{maceachern2000dependent}) when controlling for covariates. Section~\ref{sec:sim} presents a simulation study motivated by the real data application. Section~\ref{sec:analysis} applies NP-LEAP to trial data, including propensity score matching, survival analysis, functional summaries, and a tipping point analysis. Section~\ref{sec:discussion} concludes.

\section{Data and Scientific Questions}
\label{sec:data}
We consider two clinical studies of patients with non-small cell lung cancer (NSCLC), the most prevalent form of lung cancer and a central focus of ongoing immunotherapy drug development.

A Phase 2 (i.e., the current data) randomized controlled trial evaluated chemotherapy combined with an investigational immunotherapy agent versus chemotherapy combined with placebo. A total of $104$ patients were randomized in a 2:1 ratio ($70$ treatment, $34$ control). The primary estimand is the difference in overall survival probabilities at 12 months. 

An earlier Phase 3 trial (i.e., the historical data) enrolled 343 patients to a similar chemotherapy control regimen. The substantially larger historical control arm offers the possibility of augmenting the current control arm through Bayesian dynamic borrowing, potentially reducing uncertainty about the OS treatment effect estimate — a common and practically important challenge in Phase 2 oncology development, where control arms are routinely undersized.

This application poses three scientific questions that the analysis in Section~\ref{sec:analysis} answers.
\begin{enumerate}
    \item Does the time-specific survival difference between the active and control arms, $\Delta(t) = S_1(t) - S_0(t)$, meaningfully differ at clinically relevant milestone time points?
    \item Does borrowing from the historical control arm increase the precision of treatment effect estimates sufficiently to alter the inferential conclusion?
    \item How sensitive is any conclusion to the degree of assumed exchangeability between the two control populations?
\end{enumerate}

Two features of the data make this a particularly demanding test case for any borrowing method. First, as shown in Section~\ref{sec:ps_matching}, propensity-score-matched Kaplan–Meier curves suggest potential non-exchangeability between the concurrent and historical control groups at event times under one year. Second, late separation of the treatment and control survival functions in the current trial renders the proportional hazards assumption implausible. Since the Cox estimator converges to a censoring-distribution-dependent quantity under misspecification \citep{lin1989robust}, applying either standard frequentist analysis or existing parametric Bayesian borrowing methods here requires assumptions that the data contradict.

\section{Methodology}
\label{sec:methodology}

\subsection{Latent exchangeability prior}
The LEAP is a class of priors that performs individualized discounting of the historical data. It is predicated on the assumption that the historical data follow a finite mixture model, and the current data follow one component of the mixture (without loss of generality, the first component). Mathematically, we can express this assumption as
\begin{align}
    &y_i \sim F_1(\cdot | \bm{\theta}_1), \ \ i = 1, \ldots, n, 
    \quad y_{0j} \sim \sum_{k=1}^K \omega_k F_k(\cdot | \bm{\theta}_k), \ \ j = 1, \ldots, n_0,
    \label{eq:leap_assumption}
\end{align}
where $\{ F_1(\cdot | \bm{\theta}_k) \}_{k=1}^K,$ are parametric distributions functions, $ \bm{\omega} = (\omega_1,  \ldots, \omega_K)$ are component weights with $\omega_k > 0$ and $\sum_{k=1}^K \omega_k = 1$, and an appropriate prior $\pi(\bm{\theta}_1, \ldots, \bm{\theta}_K, \bm{\omega})$. Mathematically, analysis under the LEAP is equivalent to model averaging across all possible partitions of the historical data into exchangeable and nonexchangeable subsets (see also Proposition \ref{prop:leap_bma}).

Sections~\ref{ss:np-leap-abstract}--\ref{sec:cluster} develop methodology that removes the need to specify either
$K$ or a parametric family. Section~\ref{sec:prelim_dpmm} first provides the necessary background.

\subsection{Dirichlet process mixture models}
\label{sec:prelim_dpmm}

The Dirichlet process mixture model (DPMM) for data $\bm{y} = \{y_i\}_{i=1}^n$
is represented by
\begin{align}
    y_i \mid \bm{\theta}'_i \sim f(\cdot \mid \bm{\theta}'_i),
    \qquad
    \bm{\theta}'_i \mid \mu \sim \mu,
    \qquad
    \mu \mid \alpha, H_0 \sim \mathrm{DP}(\alpha, H_0),
    \label{eq:dpm_hier}
\end{align}
where $f(\cdot \mid \bm{\theta})$ is a kernel depending on $\bm{\theta} \in \Theta$
and $\mu \sim \mathrm{DP}(\alpha, H_0)$ \citep{ferguson1973bayesian} is a random
probability measure on $\Theta$ with concentration parameter $\alpha > 0$ and
base measure $H_0$.
The stick-breaking construction \citep{sethuraman1994constructive} gives
$\mu = \sum_{k=1}^{\infty} w_k\, \delta(\bm{\theta}_k)$ with
$\{\bm{\theta}_k\}_{k=1}^{\infty} \iid H_0$ and $\bm{w} \sim \mathrm{GEM}(\alpha)$
\citep{ewens1990population}, so the marginal density of $y_i$ given $\mu$ is
the infinite kernel mixture
$f_{\mu}(y) = \sum_{k=1}^{\infty} w_k f(y \mid \bm{\theta}_k)$;
mixing a continuous kernel with the almost surely discrete DP yields a flexible
continuous density.
As $\alpha \downarrow 0$, $\bm{w} \overset{p}{\to} (1, 0, 0, \ldots)'$ and
$f_{\mu}$ converges to $f(\cdot \mid \bm{\theta}_1)$ with $\bm{\theta}_1 \sim H_0$,
casting parametric Bayesian inference as a limiting case of the DPMM.

Via the Chinese Restaurant Process (CRP) representation
\citep{gershman2012tutorial}, the DPMM induces a clustering structure on
$\bm{y}$ through latent indicators $\bm{z} = \{z_i\}_{i=1}^n$. For observation $i$, the prior clustering probabilities are given by
$$
    \pi(z_i = k | \bm{z}_{1:(i-1)}) \propto 
    \begin{cases}
        N_{k}^{[1:(i-1)]} & \text{ for an existing cluster } k = 1, \ldots, K_i \\
        \alpha & \text{ for a new cluster } k = K_i + 1
    \end{cases}
    ,
$$
where $\bm{z}_{1:(i-1)}$ are the cluster memberships for the first $i-1$ observations and $N_{k}^{[1:(i-1)]}$ is the number of observations already assigned to cluster
$k$.
This ``rich get richer'' property means larger clusters preferentially attract
new members.
In Section~\ref{sec:cluster}, we derive the analogous scheme under the NP-LEAP,
wherein historical observations may either join exchangeable clusters shared
with the current data or be assigned to separate nonexchangeable clusters.

\subsection{A nonparametric extension of the LEAP}
\label{ss:np-leap-abstract}

We first state the model and its BMA interpretation for general densities $f, g$. Proposition~\ref{prop:leap_bma} holds for any prior on the density space and does not require a parametric form, so the result is more broadly applicable than the NP-LEAP itself. Section~\ref{ss:np-leap} then instantiates this framework by placing DPMM priors on $f$ and $g$. Let the current data $\bm{y} = \{y_i\}_{i=1}^n$ have density $f$ and suppose that the historical data $\bm{y_{0}} = \{y_{0j}\}_{j=1}^{n_0}$ is a two-part mixture $\gamma f + (1-\gamma) g$ consisting of the current data density $f$ and another  density $g$. Based on latent indicators $\bm{\epsilon_0} = \{\epsilon_{0j}\}_{j=1}^n \iid \textrm{Ber}(\gamma)$, we may equivalently express this model as
\begin{align}
    \{y_i\}_{i=1}^n | f, g \iid f
    \quad \text{ and for }  j = 1, \ldots, n_0: \quad
    y_{0j} | \epsilon_{0j}, f, g \sim f^{\epsilon_{0j}} g^{1 - \epsilon_{0j}}
    \label{eq:generalized_leap_assumption}
\end{align}
Here $\gamma$ is the prior probability that an historical data individual $j$ will be exchangeable with those in the current data (i.e.~$y_{0j} \sim f$). In the following, $\DenX{\R^d}$ refers to the space of densities on $\R^d$. 
\begin{proposition}
    \label{prop:leap_bma}
    Suppose that $\Pi$ is a jointly independent prior distribution on densities $f, g \in \DenX{\R^d}$ and exchangeability proportion $\gamma \in [0,1]$ (i.e.~$\Pi(f \in A, g \in B, \gamma \in C) = \Pi(f \in A) \times \Pi(g \in B) \times \Pi(\gamma \in C)$ for any measurable $A, B \subseteq \DenX{\R^d}$ and $C \subseteq [0,1]$). Then the marginal posterior distribution of $P(f \in  \cdot |\bm{y}, \bm{y}_0)$ for the model \eqref{eq:generalized_leap_assumption} satisfies
    $$
        P( f \in A | \bm{y}, \bm{y}_0) = \sum_{\bm{b} \in \{0,1\}^{n_0}} P(\bm{\epsilon}_0 = \bm{b} | \bm{y}, \bm{y}_0) \cdot P(f \in A | \bm{y}, \bm{y}_0, \bm{\epsilon}_0 = \bm{b})
    $$
    for every measurable subset $A \subseteq \DenX{\R^d}$, where given $\bm{b} = (b_1, \ldots, b_{n_0}) \in \{0,1\}^{n_0}$
    $$
    P( f \in A | \bm{y}, \bm{y}_0, \bm{\epsilon}_0 = \bm{b}) \propto \int_{A} \prod_{i=1}^n f(y_i) \prod_{j=1}^{n_0} f(y_{0j})^{b_j} \Pi(\bd f)
    $$
    denotes the posterior of $f$ based on pooling the current and exchangeable historical observations when the event $\bm{\epsilon_0} = \bm{b}$ is known, and  denoting $|\bm{b}| \doteq \sum_{j=1}^{n_0} b_j$, 
    the posterior probability of the event $\bm{\epsilon_0} = \bm{b}$ is given by
    $$
    \begin{aligned}
    P(\bm{\epsilon}_0 = \bm{b} | \bm{y}, \bm{y}_0) \propto \int \prod_{i=1}^n f(y_i) \prod_{j=1}^{n_0} f(y_{0j})^{b_j} \Pi(\bd f) \times \int \prod_{j=1}^{n_0} g(y_{0j})^{1-b_j}\Pi(\bd g)  \times \int   \gamma^{|\bm{b}|} (1-\gamma)^{n_0-|\bm{b}|} \Pi(\bd \gamma)
    \end{aligned}
    $$
    for a proportionality constant such that $\sum_{\bm{b} \in \{0,1\}^{n_0}} P(\bm{\epsilon}_0 = \bm{b} | \bm{y}, \bm{y}_0) = 1$. 
\end{proposition}
The proof of Proposition \ref{prop:leap_bma} is in Section A of the Supplementary Materials. Proposition~\ref{prop:leap_bma} extends the result of \cite{alt2024leap} to a nonparametric setting. Crucially, whereas the original LEAP employs a $K$-component parametric mixture with weights $\{\omega_k\}_{k=1}^K$, the generalization in \eqref{eq:generalized_leap_assumption} requires only a two-component mixture for the historical data. In this paper, we specify $f$ and $g$ as DPMMs, placing DP priors on their respective mixing measures $\mu$ and $\nu$.

\subsection{Nonparametric latent exchangeability prior}
\label{ss:np-leap}

We now introduce the nonparametric latent exchangeability prior
(NP-LEAP) by taking the priors of $f$ and $g$ in Proposition \ref{prop:leap_bma} to be Dirichlet process mixture models (Section \ref{sec:prelim_dpmm}). Namely, starting from two parametric families of kernels $\{f(y|\bm{\theta}): \bm{\theta} \in \Theta\}$ and $\{g(y|\bm{\lambda}): \bm{\lambda} \in \Lambda\}$ on $\R^d$ (e.g.~isotropic location-scale Gaussians) and independent random distributions  $\mu \sim \mathrm{DP}(\alpha, P_0)$ and $\nu \sim \mathrm{DP}(\eta, Q_0)$ on $\Theta$ and $\Lambda$, respectively, we consider model \eqref{eq:generalized_leap_assumption} with $f=f_{\mu}$ and $g=g_{\nu}$ following the mixture densities 
$f_\mu(y) \doteq \int f(y \mid \bm{\theta})\,\mu(\bd \bm{\theta})$ and
$g_\nu(y) \doteq \int g(y \mid \bm{\lambda})\,\nu(\bd \bm{\lambda})$. 

The NP-LEAP model extending \eqref{eq:generalized_leap_assumption} is equivalently expressed as:
\begin{align}
    y_i \mid \mu & \sim f_\mu,
    \quad i = 1, \ldots, n,
    \qquad 
    y_{0j} \mid \mu, \nu, \gamma
        \sim \gamma\, f_\mu + (1-\gamma)\, g_\nu,
    \quad j = 1, \ldots, n_0.
    \label{eq:npleap_hier}
\end{align}
Using the stick-breaking representation $\mu = \sum_{k=1}^{\infty} w_k \delta(\bm{\theta}_k)$ and $\nu = \sum_{l=1}^{\infty} \rho_l \delta(\bm{\lambda}_l)$ from Section \ref{sec:prelim_dpmm}, model  \eqref{eq:npleap_hier} corresponds to the  mixture likelihood
\begin{align}
    &p(\bm{y}, \bm{y}_0| \{\bm{\theta}_k\}_{k=1}^\infty, \{\bm{\lambda}_l\}_{l=1}^\infty, \bm{w}, \bm{\rho}, \gamma) \notag\\
    &= \left[ \prod_{i=1}^n \left\{
        \sum_{k=1}^{\infty} w_{k} f(y_i \mid \bm{\theta}_k)
       \right\}
    \right]
    \left[
        \prod_{j=1}^{n_0} \left\{
            \gamma \sum_{k=1}^{\infty} w_{k} f(y_{0j} \mid \bm{\theta}_k)
            + (1 - \gamma)
              \sum_{l=1}^{\infty} \rho_{l} g(y_{0j} \mid \bm{\lambda}_{l})
        \right\}
    \right],
    \label{eq:npleap_stick}
\end{align}
with priors $\{\bm{\theta}_k\}_{k=1}^\infty \iid P_0$,
$\{\bm{\lambda}_{l}\}_{l=1}^\infty \iid Q_0$,
$\bm{w} \sim \text{GEM}(\alpha)$, and $\bm{\rho} \sim \text{GEM}(\eta)$.

The representation in \eqref{eq:npleap_hier} and \eqref{eq:npleap_stick} illustrates two
substantial improvements over the LEAP. First, the infinite mixture
$f_{\mu}(y) = \sum_{k=1}^{\infty} w_k f(y \mid \bm{\theta}_k)$ allows for the current data density to be flexibly modeled. Second,
historical data are modeled as a two-part mixture of infinite mixtures $f_{\mu}$ and $g_{\nu}$,
adding flexibility in modeling the historical data and obviating having
to choose the number of components in the mixture model for the
traditional LEAP.

\subsection{Clustering scheme of the NP-LEAP}
\label{sec:cluster}
Proposition~\ref{prop:cluster} characterizes the prior and full conditional posterior for the cluster indicators $\bm{z}$
and $\bm{z}_0$ under the NP-LEAP; the proof is in Section~B of the Supplement. We assume $\gamma \sim \mathrm{Beta}(c_1, c_0)$ for ease of exposition.

The starting point is the introduction of latent cluster indicators
$\bm{z} = (z_1, \ldots, z_n)$ for the current data and
$\bm{z}_0 = (z_{01}, \ldots, z_{0n_0})$ for the historical data,
together with $\bm{\epsilon}_0 = (\epsilon_{01}, \ldots, \epsilon_{0n_0})$,
$\epsilon_{0j} \iid \mathrm{Ber}(\gamma)$, and $\gamma \sim \mathrm{Beta}(c_1, c_0)$.
Each $z_i \in \{1, 2, \ldots\}$ indexes an exchangeable cluster for current
observation $i$.  For historical observation $j$, $z_{0j} \in \{1, 2, \ldots\}$
indexes an \emph{exchangeable} cluster (parameter $\bm{\theta}_{z_{0j}}$)
when $\epsilon_{0j} = 1$, and a \emph{nonexchangeable} cluster (parameter
$\bm{\lambda}_{z_{0j}}$) when $\epsilon_{0j} = 0$.  Under this convention,
model~\eqref{eq:npleap_hier} is equivalently described by
\begin{align}
    y_i \mid z_i &\sim f\!\left(\cdot \mid \bm{\theta}_{z_i}\right),
    \qquad 
    y_{0j} \mid z_{0j},\, \epsilon_{0j} \sim
    f\!\left(\cdot \mid \bm{\theta}_{z_{0j}}\right)^{\epsilon_{0j}}
    g\!\left(\cdot \mid \bm{\lambda}_{z_{0j}}\right)^{1 - \epsilon_{0j}}
    \label{eq:np-leap-cluster}
\end{align}
for $i = 1, \ldots, n$ and $j = 1, \ldots, n_0$ with $\{\bm{\theta}_k\}_{k=1}^\infty \iid P_0$ and
$\{\bm{\lambda}_l\}_{l=1}^\infty \iid Q_0$ as in~\eqref{eq:npleap_stick}. 

Proposition~\ref{prop:cluster} describes the prior $\Pi$ and full conditional
posterior $P$ for the cluster indicators $\bm{z}$ and $\bm{z_0}$ under the NP-LEAP.
\begin{proposition}
    \label{prop:cluster}
    Consider the hierarchical model in \eqref{eq:np-leap-cluster}, with
    $z_i$ and $z_{0j}$ denoting cluster assignments,
    $\epsilon_{0j}\in\{0,1\}$ the exchangeability indicator,
    $k$ indexing exchangeable clusters, and $l$ nonexchangeable clusters.
    Define sequential counts
    \[
        n_k^{[1:(i-1)]} = \sum_{q=1}^{i-1} I(z_q=k), \quad
        N_{0k}^{[1:(j-1)]} = \sum_{q=1}^{j-1} I(\epsilon_{0q}=1,\,z_{0q}=k), \quad
        M_{0l}^{[1:(j-1)]} = \sum_{q=1}^{j-1} I(\epsilon_{0q}=0,\,z_{0q}=l),
    \]
    with class totals $N_0^{[1:(j-1)]}$ and $M_0^{[1:(j-1)]}$,
    let $n_k = \sum_{i=1}^n I(z_i=k)$, and denote leave-one-out
    versions by the superscript $(-i)$ or $(-j)$.
    Throughout, the full conditional for each individual follows the
    same structural form as its prior predictive rule, with sequential
    counts replaced by leave-one-out counts and each weight multiplied
    by the appropriate likelihood factor.

    \begin{enumerate}
        \item \textbf{Current data.}
        The prior assigns individual $i$ to cluster $k$ with weight
        proportional to $n_k^{[1:(i-1)]}$ (existing) or $\alpha$ (new).
        The full conditional is
        \[
            P(z_i = k \mid \bm{y}, \bm{y}_0, \bm{\Omega}_{-i}) \propto
            \begin{cases}
                \bigl(n_k^{(-i)} + N_{0k}\bigr)\,f(y_i \mid \bm{\theta}_k)
                    & k \in \{1,\ldots,K_i\}, \\[4pt]
                \alpha\displaystyle\int f(y_i \mid \bm{\theta})\,P_0(d\bm{\theta})
                    & k = K_i + 1,
            \end{cases}
        \]
        where $N_{0k} = \sum_j I(\epsilon_{0j}=1,\,z_{0j}=k)$.

        \item \textbf{Historical data (exchangeable, $\epsilon_{0j}=1$).}
        Let $\bm{\Omega}_{0j}$ collect all preceding assignments. The prior is
        \[
            \Pi(z_{0j}=k,\;\epsilon_{0j}=1\mid\bm{\Omega}_{0j}) =
            \frac{N_0^{[1:(j-1)]}+c_1}{j-1+c_1+c_0}
            \times
            \begin{cases}
                \dfrac{n_k + N_{0k}^{[1:(j-1)]}}{n+N_0^{[1:(j-1)]}+\alpha}
                    & \text{existing }k,\\[6pt]
                \dfrac{\alpha}{n+N_0^{[1:(j-1)]}+\alpha}
                    & \text{new }k.
            \end{cases}
        \]
        The full conditional weights are proportional to
        $\bigl(N_0^{(-j)}+c_1\bigr)\bigl(n_k+N_{0k}^{(-j)}\bigr)
        f(y_{0j}\mid\bm{\theta}_k)$ for an existing cluster $k \in \{1,\ldots,K_j\}$
        and $\bigl(N_0^{(-j)}+c_1\bigr)\alpha
        \int f(y_{0j}\mid\bm{\theta})\,P_0(d\bm{\theta})$ for a new one,
        with the common factor $(n_0-1+c_1+c_0)^{-1}$.

        \item \textbf{Historical data (nonexchangeable, $\epsilon_{0j}=0$).}
        The prior is
        \[
            \Pi(z_{0j}=l,\;\epsilon_{0j}=0\mid\bm{\Omega}_{0j}) =
            \frac{M_0^{[1:(j-1)]}+c_0}{j-1+c_1+c_0}
            \times
            \begin{cases}
                \dfrac{M_{0l}^{[1:(j-1)]}}{M_0^{[1:(j-1)]}+\eta}
                    & \text{existing }l,\\[6pt]
                \dfrac{\eta}{M_0^{[1:(j-1)]}+\eta}
                    & \text{new }l.
            \end{cases}
        \]
        The full conditional weights are proportional to
        $\bigl(M_0^{(-j)}+c_0\bigr)M_{0l}^{(-j)}\,g(y_{0j}\mid\bm{\lambda}_l)$
        for an existing cluster $l \in \{1,\ldots,L_j\}$ and
        $\bigl(M_0^{(-j)}+c_0\bigr)\eta
        \int g(y_{0j}\mid\bm{\lambda})\,Q_0(d\bm{\lambda})$ for a new one,
        with the common factor $(n_0-1+c_1+c_0)^{-1}$.
    \end{enumerate}
\end{proposition}
Part (1) is precisely the ``rich-get-richer'' clustering scheme under a
standard DPMM, augmented so that exchangeable historical individuals in
cluster $k$ also increase its attractiveness to new current data
individuals. Part (2) shows that the prior probability of a historical
individual joining an exchangeable cluster is proportional to the
\emph{combined} count of current and historical individuals already
assigned there---``the rich get even richer''---with the class-level
weight governed by a Dirichlet--multinomial predictive rule over
$c_1$ and $c_0$. Part (3) is the nonexchangeable analogue: an
unexchangeable individual is assigned to nonexchangeable clusters via a
separate ``rich get richer'' rule that is entirely decoupled from the
current data. In both Parts (2) and (3), the full conditional simply
replaces sequential counts with their leave-one-out counterparts and
multiplies each weight by the appropriate kernel evaluated at
$y_{0j}$; a historical individual is therefore more likely to be
classified as exchangeable when its outcome aligns with clusters shared
with the current data, and will gravitate toward a nonexchangeable
cluster otherwise. Parts (1)--(3) together yield a collapsed Gibbs sampler (see \citet{neal2000markov}). Note that Part~(1) reduces exactly to the standard DPMM full conditional
from Section~\ref{sec:prelim_dpmm} when $N_{0k} = 0$ for all $k$
(i.e., no historical observation has yet been assigned to any exchangeable
cluster), confirming that the NP-LEAP nests the standard DPMM as a
limiting case.

\subsection{Posterior consistency}

A natural concern when incorporating historical data under the NP-LEAP is asymptotic bias in estimating $f_0$, particularly when $n_0$ is large. We show this bias is absent under weak conditions.

Suppose that the current and historical data, now considered as random variables, are drawn independently from densities $f_0$ and $\gamma_0 f_0 + (1-\gamma_0)g_0$, respectively, based on a  non-exchangeable component density $g_0$ and a non-trivial exchangeable proportion $\gamma_0 \in (0,1)$. Recall that $f_{\mu}(y) = \sum_{k=1}^n w_k f(y|\bm{\theta}_k)$ denotes the current data density under the NP-LEAP model \eqref{eq:npleap_hier}. The following result provides mild  conditions under which the NP-LEAP posterior error $\|f_{\mu} - f_0\|_1 \doteq \int |f_{\mu}(x) - f_0(x)|dx$ in estimating the true current data density $f_0$ vanishes asymptotically.  

\begin{theorem}  Suppose that  $Y_1, Y_2, \ldots \iid f_0$ and $Y_{01}, Y_{02}, \ldots\iid \gamma_0 f_0 + (1-\gamma_0)g_0$ for some $\gamma_0 \in (0,1)$. Given $\beta \in [0, \infty)$, let us consider an asymptotic regime  where we take $\bm{Y}_n \doteq \{Y_1, \ldots, Y_n\}$ and $\bm{Y}_{n_0} \doteq \{Y_{01}, \ldots, Y_{0n_0}\}$ to be the current and historical data, respectively, satisfying $n_0/n \to \beta$ as $n \to \infty$. Suppose further that
\begin{enumerate}
    \item $f_0, g_0$ are densities on $\R^d$ that satisfy the standard regularity Assumption 6 from Section F of the Supplement (requiring  continuity and mild tail conditions),
    \item $f(\cdot|\bm{\theta})$ and $g(\cdot|\bm{\lambda})$ are isotropic Gaussian kernels on $\R^d$ with  location and scale parameters $\bm{\theta}, \bm{\lambda} \in \R^d \times (0,\infty)$, and base measures   
    $P_0 = \nu_1 \times \xi_1$ and $Q_0 = \nu_2 \times \xi_2$ on $\R^d \times(0, \infty)$ have a product form and satisfy Assumption 5 in Section F of Supplement (requiring that $\nu_1, \nu_2$ have exponentially decaying tails and $\xi_1,\xi_2$ have a finite moment and exponentially decaying lower tails), and
    \item $\gamma$ is independently assigned a prior on $[0,1]$ that contains $\gamma_0$ in its support.
\end{enumerate}
Then for each $\epsilon > 0$, the NP-LEAP posterior $P(\cdot|\bm{Y}_n, \bm{Y}_{n_0})$ for the model \eqref{eq:npleap_stick} satisfies
$$
\lim_{n \to \infty} P(\|f_0 - f_{\mu}\|_1 > \epsilon|\bm{Y}_n, \bm{Y}_{n_0}) = 0 \qquad \text{ almost surely.}
$$
\label{thm:consistency}
\end{theorem}

\begin{proof}[Proof sketch.] The proof proceeds by verifying two standard conditions for posterior consistency of non-i.i.d.\ observations \citep[Theorem 6.41]{ghosal2017fundamentals}. The first is a Kullback--Leibler (KL) support condition: for every $\epsilon > 0$, the NP-LEAP prior assigns positive mass to the set of parameters $(f_{\mu}, g_{\nu}, \gamma)$ for which the weighted KL divergence $\frac{1}{1+\beta} \KL(f_0 | f_{\mu}) + \frac{\beta}{1+\beta}\KL(\gamma_0 f_0 + (1-\gamma_0)g_0 \,|\, \gamma f_{\mu} + (1-\gamma)g_{\nu})$ is less than $\epsilon$. This holds by Lemma 3 of the Supplement because the DPMM priors $f_\mu$ and $g_\nu$ assign positive mass to arbitrary Kullback--Leibler neighborhoods of $f_0$ and $g_0$ under mild conditions on the base measures, and the prior on $\gamma$ has $\gamma_0$ in its support. The second is a sieve condition: the prior places exponentially small mass outside sets that can be covered by exponentially many Hellinger balls. This is established using Lemma 5 of the Supplement by constructing sieves for the DPMM with Gaussian kernel using a finite truncation argument, bounding the $L^1$ metric entropy of the sieve, and showing the prior mass of the complement is exponentially small (Section F of the Supplement).
\end{proof}

The full proof is in Section E of the Supplement (Theorem 1 and Corollary 1.1); Gaussian kernel conditions are in Section F. Theorem 1 thus guarantees consistency even when historical data are only partially exchangeable, provided $n_0 = O(n)$.
The technical derivations in the supplement ultimately formalize the insight that the NP-LEAP posterior will concentrate around maximizer of the likelihood, or equivalently around solutions $(\mu^*, \nu^*, \gamma^*)$ of the minimization problem. The technical derivations in the supplement ultimately formalize the insight that the NP-LEAP posterior will concentrate around maximizer of the likelihood, or equivalently around solutions $(\mu^*, \nu^*, \gamma^*)$ of the minimization problem
\begin{equation}
\min_{(\mu, \nu, \gamma)} \frac{1}{1+\beta}\KL(f_0|f_{\mu}) + \frac{\beta}{1+\beta}\KL(\gamma_0 f_0 + (1-\gamma_0)g_0|\gamma f_{\mu} + (1-\gamma) g_{\nu} ),
\label{eq:kl_min}
\end{equation}
where $\KL$ denotes the Kullback-Leibler divergence between densities. 

\paragraph*{Remark} While a solution $(\mu^*, \nu^*, \gamma^*)$ of \eqref{eq:kl_min} will satisfy  $f_0  = f_{\mu^*}$, it does not imply that $\gamma^* = \gamma_0$ or $g_{\nu^*} = g_0$ (see Remark 2 in Section E of the Supplement). That is why Theorem \ref{thm:consistency} only claims consistency for $f_0$, and not $\gamma_0$ or $g_0$. This is an inherent non-identifiability of the NP-LEAP model: given $f_0$, there exist multiple pairs $(\gamma, g_{\nu})$ that yield identical likelihoods, so the data cannot distinguish among them. However, for the primary inferential goal of estimating $f_0$ and downstream functionals, this non-identifiability is inconsequential.

\subsection{ANOVA DDP Model for Time-to-Event Data}
In our simulations and data analyses we employ the ANOVA dependent Dirichlet process (DDP) model \citep{deiorio2002anova} for the
time-to-event setting. Let $y_i = \min\{\log t_i, \log c_i\}$ denote
the log observed time and $a_i \in \{0,1\}$ the treatment indicator.
We adopt a log-normal kernel
$f(y \mid a;\, \bm{\theta}_k)
 = \phi(y \mid \beta_{k0} + \beta_{k1}\,a,\, \tau_k^{-1})$,
where $\bm{\theta}_k = (\bm{\beta}_k', \tau_k)'$ contains
arm-specific intercepts and a precision parameter. The base measure
$P_0$ is semi-conjugate:
$\bm{\beta}_k \sim N(\bm{\mu}_0, \bm{\Sigma}_0)$ and
$\tau_k \sim \text{Gamma}(\delta_\alpha, \kappa_\alpha)$.
We elicit $\bm{\mu}_0$ as the MLE of a log-normal accelerated failure
time (AFT) model and $\bm{\Sigma}_0$ as the inverse observed
information matrix scaled to correspond to 2 effective observations,
a common weakly informative data-driven choice
\citep{roy2018bayesian}. For the nonexchangeable mixing measure $\nu$
with base measure $Q_0$, we use a diffuse prior
($\bm{\mu}_{Ql} = \bm{0}$,
 $\bm{\Sigma}_{Ql} = \text{diag}\{10^2, 10^2\}$,
 $\tau_{Ql} \sim \text{Gamma}(0.5, 0.5)$)
to avoid artificially favoring the exchangeable assignment. Full hyperparameter elicitation details and the complete Gibbs sampler are
provided in Sections~C and~D of the Supplement.

\section{Simulation study}
\label{sec:sim}
We present a simulation study with three primary objectives: (i) to assess the robustness of NP-LEAP under varying degrees and magnitudes of non-exchangeability between historical and current data; (ii) to compare NP-LEAP against parametric Bayesian borrowing methods and demonstrate that model misspecification in the parametric borrowing component can cause greater harm than not borrowing; and (iii) to benchmark NP-LEAP against standard frequentist non-borrowing approaches, including the Cox PH model, to quantify the efficiency gains from nonparametric borrowing.

\subsection{Data generation process}
To keep simulations realistic, we base the data generating process on the motivating trial. We fit a spline-based non-proportional hazards to the pooled current and historical data with time-varying effects for treatment and study indicator, controlling for age, sex, and BMI. Current data sets of $n = 104$ individuals are generated using covariate bootstrap sampling, 2:1 permuted block randomization, and event times drawn from the fitted spline model. Dropout is calibrated so that approximately 5\% of individuals drop out after one year; accrual follows a $\text{Beta}(1,2)$ distribution. We generate 9 scenarios crossing exchangeability probability $\gamma \in \{0.0, 0.5, 1.0\}$ with outcome drift, i.e., we define $\beta_{\text{unexch}} = \beta_{\text{hist}} + \delta$ for $\delta \in \{-0.5, 0, 0.5\}$. The cumulative hazard function for the generated data is given by
$$
H(t | a, \bm{x}, \epsilon) = \exp\left\{ 
      \sum_{j=0}^4 \alpha_j z^{(k)}_j(\log t) 
    + (1- \epsilon) [ \beta_{\text{unexch}} + \beta_1 z^{(k)}_j(\log t)]
    + \sum_{j=0}^1 \psi_j a z^{(k)}_j(\log t)
    + \bm{x}'\bm{\xi}
\right\}
$$
where $H(t | a, \bm{x}, \epsilon)$ is the conditional cumulative hazard for arm $a$ of study $s$, $z_j^{(k)}(q)$ is the $j^{th}$ basis function for a spline of degree $k$ based on continuous variable $q$, and $\bm{x}$ includes sex, BMI, and weight as baseline covariates. Parameters used to generate the data were based on maximum likelihood and are given in Table~\ref{tab:spline_mle}.
{\spacingset{1}
\begin{table}[ht]
\centering
\begin{tabular}{cccccccccc} 
\toprule
$\hat{\alpha}_0$ & $\hat{\alpha}_1$ & $\hat{\alpha}_2$ & $\hat{\alpha}_3$ & $\hat{\beta}_{\text{hist}}$ & $\hat{\beta}_1$ & $\hat{\psi}_0$ & $\hat{\psi}_1$ & $\hat{\xi}_{\text{sex}}$ & $\hat{\xi}_{\text{BMI}}$ \\ 
\midrule
1.00 & -0.21 & 0.28 & -0.01 & -0.26 & 0.19 & 0.32 & 0.16 & 0.65 & 0.02 \\
\bottomrule
\end{tabular}
\caption{Maximum likelihood estimates from fitting a proportional hazards model using splines with $k = 3$ knots; the $\alpha_j$'s are time effects, the $\beta_j$'s are historical study-specific effects, the $\psi_j$'s are treatment effects, and the $\delta$'s are baseline covariate effects.}
\label{tab:spline_mle}
\end{table}
}

To mimic regulatory practice (fixed historical data), we generate 100{,}000 candidate historical data sets per scenario and retain the one whose KM estimator most closely aligns with the data generating survival function. Full details are in Section D of the Supplement.

\subsection{Comparator methods and hyperparameter specification}
To compare the relative performance of the NP-LEAP, we study the operating characteristics of multiple related approaches. For all the approaches, we first conduct 1:1 propensity score matching to control for differences in baseline covariates between the two studies, although we retain all current data individuals. Thus, all methods incorporating historical control data are based on the matched subset.

All DPMM concentration parameters are assigned Gamma$(8, 2)$ to encourage a fair degree of flexibility. For the NP-LEAP, we elicit $\pi(\gamma) = 0.9 \cdot f_{\beta}(\gamma | 1, 1) + 0.1 \cdot I(\gamma = 1)$, where $f_{\beta}(\cdot | a, b)$ denotes the $\text{Beta}(a, b)$ density; this prior is vague but encourages pooling if supported by the data. We omit a point mass at $\gamma = 0$ because the NP-LEAP already discounts to near-zero under extreme non-exchangeability, whereas it can struggle to pool under full exchangeability without additional mass at $\gamma = 1$.

For comparator methods, we considered the Kaplan-Meier estimator, semiparametric Cox model \citep{cox1972regression}, the Bayesian nonparametric (BNP) approach without borrowing, and Bayesian parametric log-normal AFT models under both vague and data borrowing priors%
. One of the comparators is an asymptotic approximation to the normalized power prior that we refer to as the normalized asymptotic power prior (NAPP), for which we used a normal approximation to the historical data log-normal AFT likelihood based on the MLE and information matrix. The hyperprior for the discounting parameter for the NAPP was specified as $\pi(a_0) = 0.9 \cdot I(0 \le a_0 \le 1) + 0.1 \cdot f_{\beta}(a_0 | 100, 1)$. A similar prior was elicited for the exchangeability probability for the parametric LEAP. 

\subsection{Simulation study results}
\label{sec:sim_results}
We take the estimand of interest to be the time-specific difference in survival
probabilities, i.e., $\Delta(t) = S_1(t) - S_0(t)$ for
$t \in \{ 3, 6, \ldots, 24 \} \equiv \mathcal{T}$ where $t$ is months since
randomization.
Let $\hat{\Delta}_b(t)$ and $(L_b(t), U_b(t))$ denote the point estimate
(posterior mean for Bayesian approaches; MLE for frequentist) and 95\% interval
estimate (credible or confidence interval) for replicate $b = 1, \ldots,
B = 10{,}000$.
We compute bias, MSE, 95\% coverage, and interval score
\citep{Gneiting01032007}:
$\text{Bias}(t) = B^{-1} \sum_{b=1}^B (\hat{\Delta}_b(t) - \Delta(t))$%
, $\text{MSE}(t) = B^{-1} \sum_{b=1}^B (\hat{\Delta}_b(t) - \Delta(t))^2$%
, $\text{Coverage}(t) = B^{-1} \sum_{b=1}^B I( L_b(t) \le \Delta(t) \le U_b(t) )$%
, $\text{IS}(t) = (U_b(t) - L_b(t)) + \frac{2}{0.05}
\min\{ |\Delta(t) - L_b(t)|, |\Delta(t) - U_b(t)| \}
\cdot I(\Delta(t) \not \in [L_b(t), U_b(t)])$,
where IS is the interval width plus a penalty proportional to the distance to
the nearest bound when the interval misses the truth.

Figure~\ref{fig:sim_results_errorbar} shows the OCs of the various methods,
where the shapes and bars denote, respectively, the mean and range of the OC
across all times in $\mathcal{T}$.
\begin{figure}
    \centering
    \includegraphics[width=0.99\linewidth,keepaspectratio]{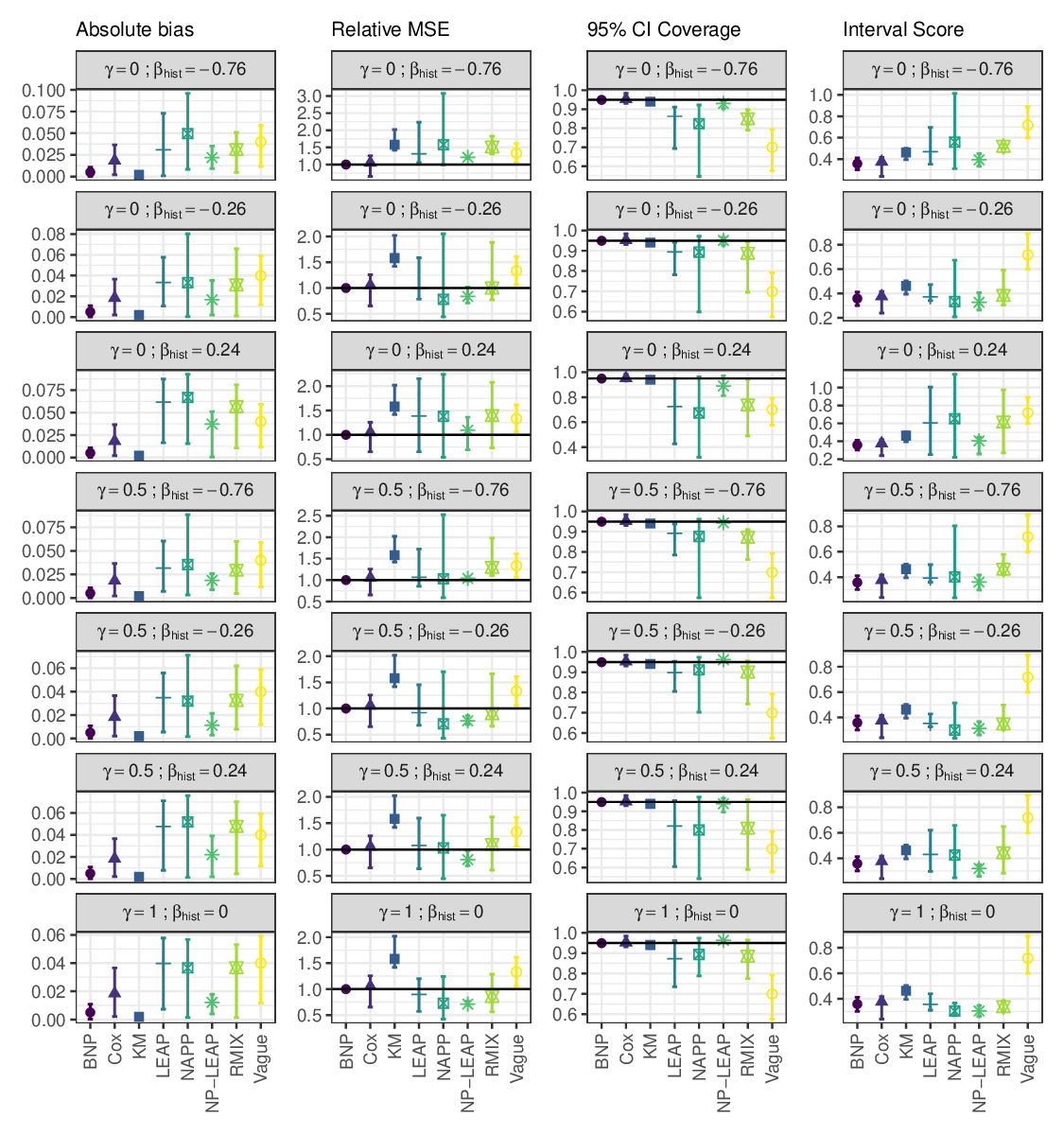}
    \caption{Operating characteristics for each method. BNP = Bayesian
    nonparametric. Symbols and error bars denote the mean and range of each
    operating characteristic across each time point $t \in \{3, 6, \ldots,
    24\}$; NP-LEAP = nonparametric LEAP, Cox = frequentist Cox proportional
    hazards model; KM = Kaplan Meier estimator; LEAP = AFT model with LEAP;
    NAPP = asymptotic approximation to normalized power prior; RMIX = AFT model
    with robust mixture prior; Vague = AFT model with vague prior}
    \label{fig:sim_results_errorbar}
\end{figure}

\textbf{Non-borrowing benchmark.}
Among methods that do not borrow historical data, the Kaplan--Meier (KM)
estimator is approximately unbiased but yields MSE at least 40\% higher than
the Bayesian nonparametric DPMM (BNP).
The frequentist Cox PH model introduces bias through its constant hazard ratio
constraint but reduces variance through implicit smoothing---particularly at
early event times---so that it can outperform BNP in MSE at individual time
points despite being misspecified.
Averaged across all time points, however, BNP achieves lower MSE than KM with
superior coverage and interval scores.

\textbf{Model misspecification in parametric borrowing is more damaging than not
borrowing.}
Across most scenarios, parametric borrowing methods (LEAP, NAPP, RMIX) exhibit
substantially elevated bias and MSE when exchangeability is partial or absent,
particularly for large outcome drift ($\delta = \pm 0.25$).
In several scenarios their MSE exceeds that of BNP (no borrowing), confirming
that borrowing under a misspecified parametric model can be actively harmful.
NP-LEAP largely avoids this pathology: its bias is comparable to the Cox model
across scenarios, and its MSE is generally as good as or better than BNP except
in the most extreme non-exchangeability scenario ($\gamma = 0$,
$\delta = -0.25$).

\textbf{NP-LEAP achieves meaningful efficiency gains under partial and full
exchangeability.}
NP-LEAP performs particularly well under scenarios with partial or full
exchangeability, where MSE is lower than BNP, coverage is approximately 95\%,
and the interval score is the lowest among all methods.
The asymmetry in performance between $\delta = -0.25$ and $\delta = +0.25$ at
fixed $\gamma$ suggests that the direction as well as the magnitude of outcome
drift affects borrowing, possibly reflecting differences in relative follow-up
and trial duration between current and historical data sets.

\textbf{Parametric misspecification is costly even without borrowing.}
The Vague AFT prior---a parametric model that borrows no historical
data---consistently underperforms the nonparametric BNP approach in
interval score and MSE across nearly all scenarios.
This confirms that the gains from NP-LEAP reflect two compounding
benefits: robustness to model misspecification and efficient use of
historical information; even the first benefit alone justifies the
nonparametric approach relative to a vague parametric specification.

In summary, NP-LEAP dominates parametric borrowing approaches across operating characteristics and outperforms nonborrowing semiparametric and parametric frequentist methods in many scenarios.
Given the realistic data-generating process---time-varying treatment and non-exchangeability effects, plausible drift magnitudes up to
$\delta = \pm 0.5$, and sample sizes matching the motivating
study---these results support the application in Section~\ref{sec:analysis}.

\section{Data analysis}
\label{sec:analysis}

\subsection{Propensity score matching and covariate balance}
\label{sec:ps_matching}
We use propensity score (PS) matching \citep{rosenbaum1983central} to assess
covariate balance between the current and historical study populations.
Table~\ref{tab:summstats} shows baseline covariates for the current, historical,
and 1:1 PS-matched historical subset.
Before matching, there is imbalance in age, sex assigned at birth, and height;
PS matching substantially reduces these differences.
Nevertheless, the Kaplan--Meier curves in Figure~\ref{fig:km_curve} suggest
potential non-exchangeability between the concurrent and historical control
groups at event times less than one year.
{\spacingset{1}
\begin{table}[t]
\fontsize{12.0pt}{14.4pt}\selectfont
\begin{tabular*}{\linewidth}{@{\extracolsep{\fill}}lccc}
\toprule
\textbf{Characteristic} & \textbf{Current} & \textbf{Historical} & \textbf{Historical (matched)} \\
 & N = 104\textsuperscript{\textit{1}} & N = 343\textsuperscript{\textit{1}} & N = 104\textsuperscript{\textit{1}} \\
\midrule\addlinespace[2.5pt]
Age & 63.8 (9.6) & 63.0 (8.2) & 63.2 (9.8) \\ 
Sex &  &  &  \\ 
\ \ \ \ Female & 25 (24\%) & 60 (17\%) & 25 (24\%) \\ 
\ \ \ \ Male & 79 (76\%) & 283 (83\%) & 79 (76\%) \\ 
BMI & 24.9 (5.2) & 24.9 (4.8) & 25.0 (4.8) \\ 
Weight & 73.2 (16.3) & 70.9 (14.6) & 73.6 (15.7) \\ 
Height & 171.2 (7.7) & 168.8 (8.3) & 171.4 (8.4) \\ 
\bottomrule
\end{tabular*}
\caption{Summary of baseline covariates between the current and historical data sets (before and after propensity score matching).}
\begin{minipage}{\linewidth}
\textsuperscript{\textit{1}}Mean (SD); n (\%)\\
\end{minipage}
\label{tab:summstats}
\end{table}
}

\begin{figure}
    \centering
    \includegraphics[width=0.8\linewidth,keepaspectratio]{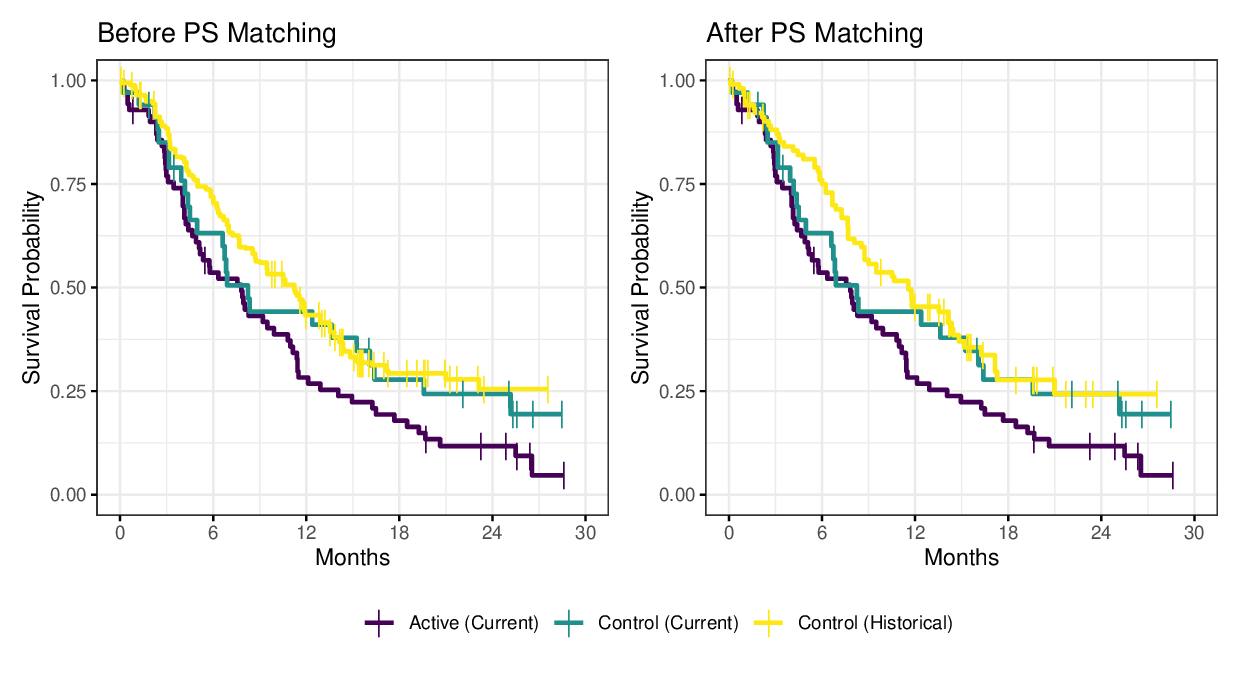}
    \caption{Kaplan-Meier curves for each arm of the current data and the
    historical control arm (before and after PS matching).}
    \label{fig:km_curve}
\end{figure}


\subsection{Survival function analyses}
Figure~\ref{fig:analysis_sfun} shows the posterior mean and 95\% credible
interval (CI) for the arm-specific survival functions and their difference,
where BNP and NP-LEAP respectively denote the ANOVA-DDP without and with
borrowing.
As expected given that we augment only the control arm, the posterior mean and
CI for the treatment arm are largely unchanged under NP-LEAP.

\begin{figure}
    \centering
    \includegraphics[width=0.8\linewidth,keepaspectratio]{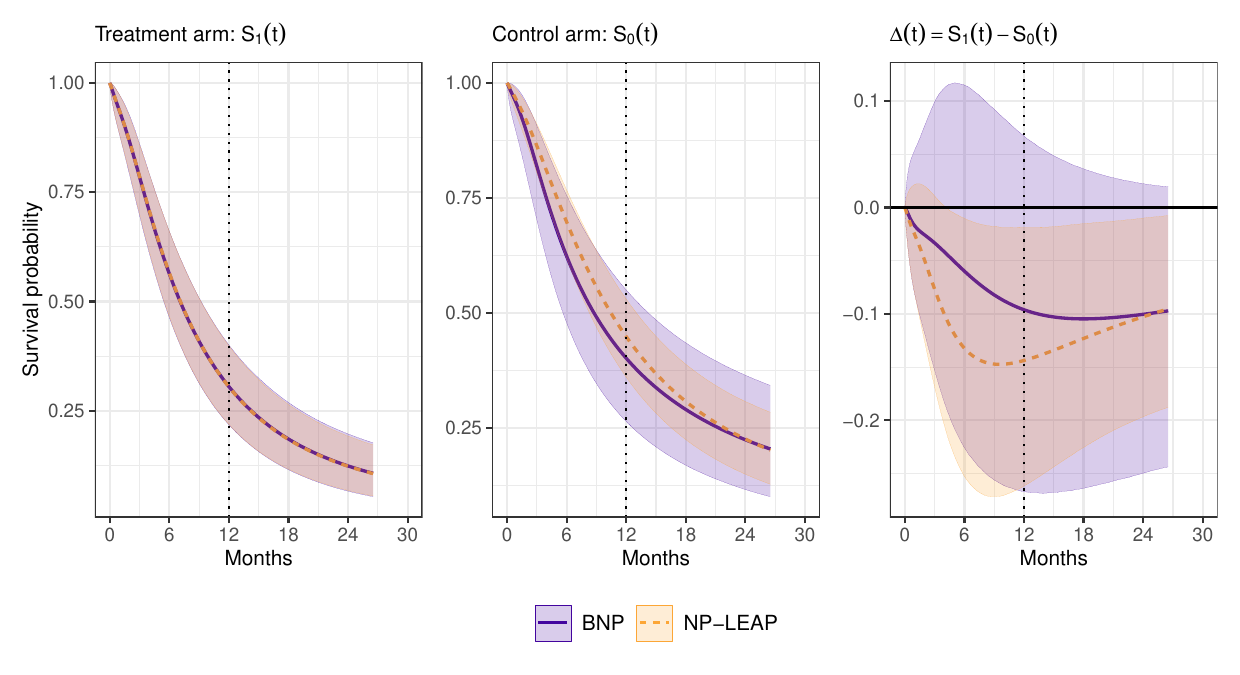}
    \caption{Posterior mean and 95\% credible interval for the arm-specific
    survival functions and their difference. BNP and NP-LEAP respectively denote
    the ANOVA-DDP model without and with borrowing.}
    \label{fig:analysis_sfun}
\end{figure}

For the control arm, the posterior mean survival probabilities shift upward
under NP-LEAP and the 95\% CI narrows noticeably, illustrating the precision
gains from borrowing.
The substantial overlap between BNP and NP-LEAP in the control arm estimates,
together with the large posterior mass on $\gamma$ near 1, suggests the
borrowing is not unduly influential.

For the survival difference $\Delta(t)$, the 95\% CI under NP-LEAP is
substantially narrower and almost completely contained within the BNP
interval, confirming a precision gain without material change to point
estimates---consistent with the simulation results in
Section~\ref{sec:sim_results}.
Notably, the BNP 95\% CI for $\Delta(t)$ straddles zero throughout
virtually the entire follow-up period, meaning that without borrowing
the analysis yields no statistically meaningful conclusion about the
treatment effect at any time point.
Under NP-LEAP, the upper bound of the 95\% CI falls below zero from
approximately 4.2 months onward, tipping the analysis from inconclusive
to actionable and illustrating how principled borrowing can resolve
inferential uncertainty that the current data alone cannot.

\textbf{Efficacy analysis for the primary endpoint.}
The key endpoint was the difference in 12-month survival probabilities,
$\Delta(12) = S_1(12) - S_0(12)$.
Under BNP, $E[\Delta(12) \mid D] = -0.10$ with a 95\% CI of $(-0.27, 0.07)$.
Under NP-LEAP, $E[\Delta(12) \mid D, D_0] = -0.14$ with a 95\% CI of
$(-0.26, -0.02)$, indicating poorer OS at 12 months for the
investigational arm.
The upper bound of the NP-LEAP 95\% CI remains below zero for all time points
beyond 4.2 months, so this conclusion is robust to the choice of milestone
survival time.

\textbf{Tipping point analysis.}
When borrowing information for analyses submitted to regulators, it has become common to perform a tipping point analysis to understand the sensitivity of the study conclusions to the amount of borrowing \citep{Best2021Bayesian}. Using the NP-LEAP, it is quite easy to perform this type of analysis by fixing a value for the proportion of exchangeability ($\gamma$) over a range of values. We perform such a tipping point analysis for $\gamma \in \{ 0.0, 0.1, 0.2, \ldots, 1.0 \}$. The results are presented in Figure~\ref{fig:tipping_point}, where we show the posterior mean and 95\% credible interval for the difference in survival probabilities at 12 months. The 95\% credible interval excludes $0$ at $\gamma \in \{0.9, 1.0\}$.
\begin{figure}
    \centering
    \includegraphics[width=0.7\linewidth,keepaspectratio]{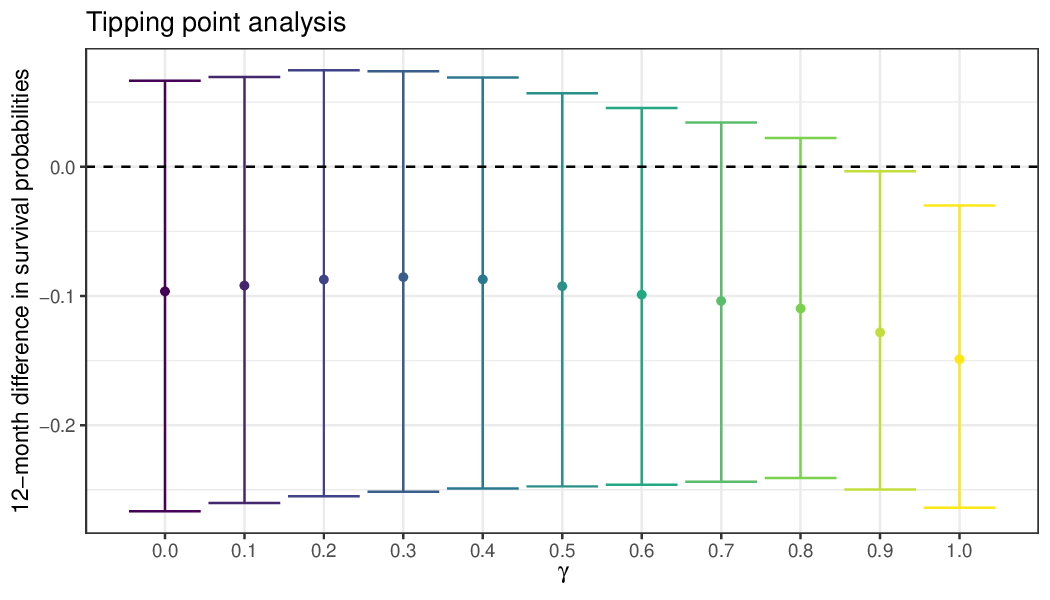}
    \caption{Tipping point analysis for difference in 12-month survival probabilities. The points depict the posterior mean and the bars depict the 95\% credible interval.}
    \label{fig:tipping_point}
\end{figure}

Interestingly, the lower tail of the CI is quite robust to the degree of borrowing while the upper tail shifts more dramatically particularly when $\gamma > 0.5$. We conjecture that this is because, as shown in Figure~\ref{fig:km_curve}, there is some separation between the historical control data in a neighborhood less than 12 months while congruence is observed in a neighborhood greater than 12 months. Indeed, changing the estimand to 24 months (data not shown) shows that the lower bound also changes.

\textbf{Efficacy analysis for secondary endpoints.} 
A significant advantage of the Bayesian nonparametric approach is the ability to conduct inference on the survival functions as well as functionals of the survival functions under a single, cohesive estimation framework. Figure~\ref{fig:postdens} compares the posterior density under the traditional DPMM with the NP-LEAP for median survival time and restricted mean survival time (RMST) evaluated at the last observed event time ($\approx 26.5$ months).

Using the posterior mean as the point estimate, the median survival time in the investigative treatment group is approximately $1.9$ months less than the control group [CI: $(-7.32, 2.16)$] according to the BNP approach without borrowing and around $3.4$ months [95\% CI: $(-6.36, -0.36)$] using the proposed NP-LEAP approach. Note that the 95\% CI under NP-LEAP is completely contained in that under the BNP approach without borrowing. Thus, while borrowing resulted in a substantial reduction in uncertainty compared to without, the overlap in CIs provides some confidence that the borrowing is not overly influential. The difference in RMSTs yielded similar conclusions.

\begin{figure}
    \centering
    \includegraphics[width=0.9\linewidth]{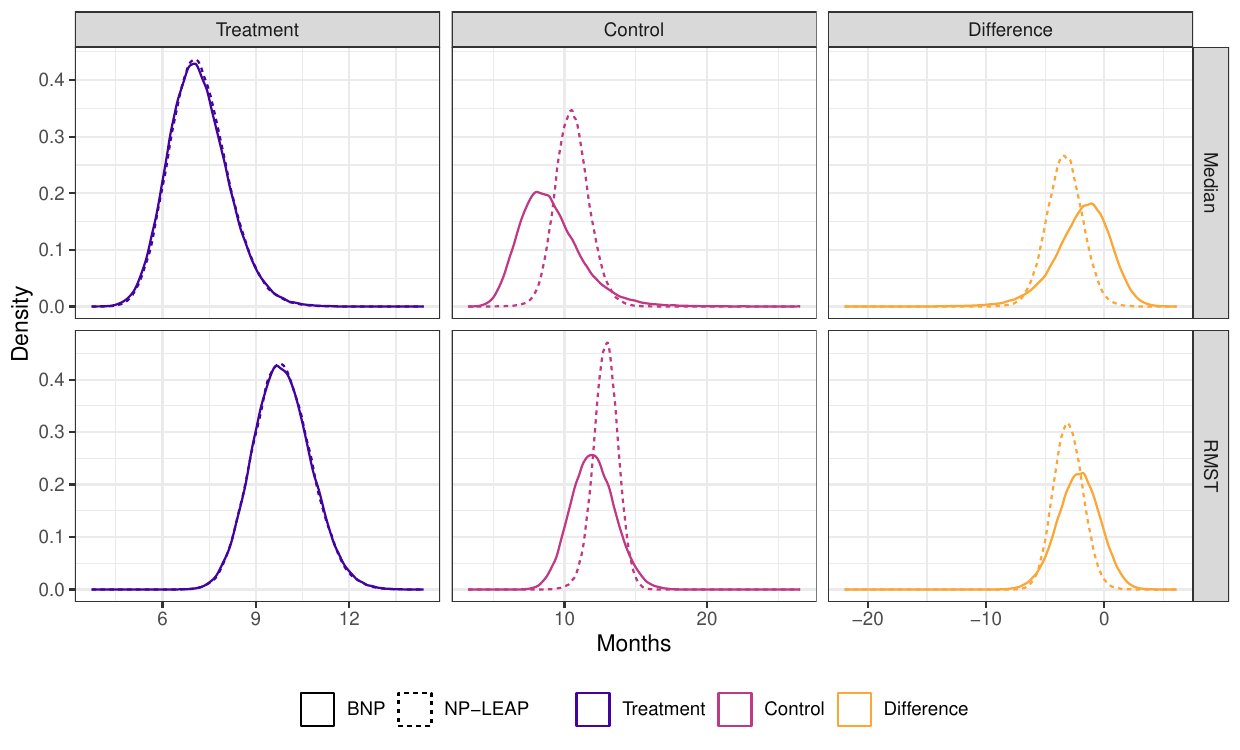}
    \caption{Posterior density for arm-specific RMST and medians along with their differences.}
    \label{fig:postdens}
\end{figure}

\section{Discussion}
\label{sec:discussion}

We introduced the NP-LEAP, a Bayesian nonparametric framework for dynamic
borrowing that extends the latent exchangeability prior beyond parametric
model classes.
By allowing the borrowing mechanism to operate on flexible nonparametric
outcome models, NP-LEAP removes reliance on restrictive distributional
assumptions while preserving individualized, observation-level discounting
across data sources.
Proposition~\ref{prop:leap_bma} establishes that the NP-LEAP implicitly
performs Bayesian model averaging over all $2^{n_0}$ partitions of the
historical data into exchangeable and nonexchangeable subsets, with
partition weights updated by the data.
Theorem~\ref{thm:consistency} provides a formal guarantee that this
borrowing does not induce asymptotic bias in estimating the current data
density, provided the historical sample size grows at the same rate as
the current sample size.
Through simulation and a real data application in NSCLC, we demonstrated
that NP-LEAP achieves meaningful efficiency gains under partial or full
exchangeability while remaining robust to prior-data conflict in a way
that parametric borrowing methods are not.

A practical advantage of the NP-LEAP is computational tractability.
The collapsed Gibbs sampler (Sections~C and~D of the Supplement) is
implemented in custom C\texttt{++} code and analyzes a single data set
in approximately five seconds, making the method feasible for the large
simulation studies required by regulatory agencies to assess operating
characteristics, as well as for routine use in data analysis.

Although NP-LEAP applies broadly across outcome types through appropriate
kernel specification, two settings warrant caution.
First, for binary outcomes, the use of a DPMM with a Bernoulli kernel
introduces potential identifiability concerns: a finite mixture of
infinite Bernoulli mixtures may not be well-identified from data, as
distinct component configurations can yield similar marginal likelihoods.
In such settings, practitioners should consider whether a parametric
or semiparametric approach is more appropriate, or whether the historical
sample size is sufficient to support nonparametric identification.
Second, the clustering mechanism underlying the NP-LEAP requires enough
observations to form meaningful clusters; with very small current sample
sizes, the Dirichlet process will tend to produce few clusters and the
nonparametric flexibility may offer little advantage over a parametric
specification, while potentially increasing variance.
We recommend that practitioners assess the posterior number of clusters
as a diagnostic for whether the nonparametric component is being
meaningfully utilized.

Several directions merit further investigation.
First, characterizing frequentist operating characteristics of the
NP-LEAP---such as type I error control and the asymptotic distribution
of posterior summaries---is important for regulatory acceptance and
remains an open theoretical problem.
Second, Theorem~\ref{thm:consistency} establishes posterior consistency
but does not characterize the rate at which the posterior contracts
around the truth; deriving contraction rates under the NP-LEAP would
strengthen the theoretical foundation and clarify how the ratio
$n_0 / n$ affects estimation precision.
Finally, the NP-LEAP framework can be expanded to allow for controlling for covariates in three different ways: through the probability of being exchangeable; through the kernel; or through the cluster weights. It is of practical and theoretical interest to know what is the best way to handle covariates given this flexibility.

The NP-LEAP demonstrates that dynamic borrowing need not
be confined to parametric outcome models.
By integrating nonparametric modeling with individualized discounting, the framework provides a flexible and principled foundation for information sharing across studies, with potential applications beyond the motivating example considered here.

\section{Data Availability Statement}\label{data-availability-statement}
Data not available due to legal and commercial restrictions.

\phantomsection\label{supplementary-material}
\bigskip

\begin{center}

{\large\bf SUPPLEMENTARY MATERIAL}

\end{center}

\begin{description}
\item[Appendix:]
Detailed derivations and supplementary figures are presented in the Appendix (pdf).
\end{description}



\bibliography{03_refs.bib}

\end{document}


\def\spacingset#1{\renewcommand{\baselinestretch}%
{#1}\small\normalsize} \spacingset{1}


\if1\anon
{
  \title{\bf Supplementary Materials for: NP-LEAP: Nonparametric Latent Exchangeability Prior for Model-Lean Borrowing from Historical Data}
  \author{Ethan M. Alt, \hspace{.2cm}\\
    Statistics and Data Science Innovation Hub, GSK
    \\ \\
    Miheer Dewaskar, \hspace{.2cm} \\ Department of Mathematics and Statistics, 
    University of New Mexico 
    \\ \\
    Jacob M. Maronge, Yuelin Lu, and Matthew A. Psioda, \hspace{.2cm} \\
    Statistics and Data Science Innovation Hub, GSK
  }
  \maketitle
} \fi

\if0\anon
{
  \bigskip
  \bigskip
  \bigskip
  \begin{center}
    {\LARGE\bf Supplementary Materials for: NP-LEAP: Nonparametric Latent Exchangeability Prior for Model-Lean Borrowing from Historical Data}
\end{center}
  \medskip
} \fi

\bigskip

\noindent%
{\it Keywords:} bayesian nonparametrics; bayesian dynamic borrowing; historical data
\vfill

\newpage
\spacingset{1.8} 

\section{Proof of BMA interpretation of NP LEAP}

We suppose that the current data $\bm{y} = \{y_i\}_{i=1}^n \subseteq \R^d$ is drawn from a density $f$ and that the historical data $\bm{y_0} = \{y_{0j}\}_{j=1}^{n_0} \subseteq \R^d$ is drawn from a two-part mixture density $\gamma f + (1-\gamma) g$ composed of $f$ and another density $g$ on $\R^d$.
Based on independent latent indicators $\bm{\epsilon_0} = \{\epsilon_{0j}\}_{j=1}^n \iid \textrm{Ber}(\gamma)$, we may equivalently express this model as
\begin{align}
    \{y_i\}_{i=1}^n | f, g \iid f
    \quad \text{ and for }  j = 1, \ldots, n_0: \quad
    y_{0j} | \epsilon_{0j}, f, g \sim \begin{cases}
      f & \text{ if } \epsilon_{0j} = 1 \\
      g & \text{ if } \epsilon_{0j} = 0
    \end{cases}
    \label{eq:generalized_leap_assumption}
\end{align}
Here $\gamma = \Pi(\epsilon_{0j} = 1|\gamma)$ is the prior probability that an historical data individual $j$ will be exchangeable with those in the current data (i.e.~$y_{0j} \sim f$). 

We assume an independent but otherwise arbitrary prior $\Pi(\bd f, \bd g, \bd \gamma) = \Pi(\bd f) \Pi(\bd g) \Pi(\bd \gamma)$ on densities $f, g$ and exchangeable probability $\gamma \in [0,1]$.
Then for some proportionality constant that only depends on $\bm{y}$ and $\bm{y}_0$, the joint posterior distribution is seen to factor as
\begin{equation}
\begin{aligned}
    &P(\bd f, \bd g, \bd \gamma, \bm{\epsilon}_0 | \bm{y}, \bm{y}_0)\\
    &\propto \Pi(\bd f) \Pi(\bd g) \Pi(\bd \gamma) \gamma^{|\bm{\epsilon}_0|}(1-\gamma)^{n_0-|\bm{\epsilon}_0|} 
     \prod_{i=1}^n f(y_i) 
    \prod_{j=1}^{n_0} \left\{ f(y_{0j})^{\epsilon_{0j}} g(y_{0j})^{1 - \epsilon_{0j}} \right\}
     \\
    &= \left[ \Pi(\bd f) \prod_{i=1}^n f(y_i) \prod_{j : \epsilon_{0j} = 1} f(y_{0j}) \right] \times \left[ \Pi(\bd g) \prod_{j : \epsilon_{0j} = 0} g(y_{0j}) \right]  \times \left[ \gamma^{|\bm{\epsilon}_0|}(1-\gamma)^{n_0-|\bm{\epsilon}_0|}  \Pi(\bd \gamma) \right]\\
    &=  P(\bd f | \bm{y}, \bm{y}_{0, \text{exch}}, \bm{\epsilon}_0) \times P(\bd g | \bm{y}, \bm{y}_{0, \text{nonexch}}, \bm{\epsilon}_0) \times p( \bm{y}, \bm{y}_{0}|\bm{\epsilon_0}) \times P(\bd \gamma | \bm{\epsilon_0}) \times \Pi(\bm{\epsilon_0})
\end{aligned} 
\label{eq:posterior_factorization}
\end{equation}
where $|\bm{\epsilon}_0| \doteq \sum_{j=1}^{n_0} \epsilon_{0j}$ and
\begin{itemize}
    \item $\Pi(\bm{\epsilon_0}) = \int \gamma^{|\bm{\epsilon}_0|}(1-\gamma)^{n_0 - |\bm{\epsilon}_0|} \Pi(\bd \gamma)$ is the marginal probability of observing the partition $\bm{\epsilon_0} \in \{0,1\}^{n_0}$ and $P(\bd \gamma | \bm{\epsilon_0}) = \frac{1}{\Pi(\bm{\epsilon_0})}\gamma^{|\bm{\epsilon}_0|}(1-\gamma)^{n_0-|\bm{\epsilon}_0|}  \Pi(\bd \gamma)$ is the conditional probability of $\gamma$ having observed such a partition.
    \item $p( \bm{y}, \bm{y}_{0}|\bm{\epsilon_0}) = \int \int \prod_{i=1}^n f(y_i) \prod_{j=1}^{n_0} \left\{ f(y_{0j})^{\epsilon_{0j}} g(y_{0j})^{1 - \epsilon_{0j}} \right\} \Pi(\bd f) \Pi(\bd g)$ is the marginal likelihood of observing the data under the partition $\bm{\epsilon}_0$. Due to independence, this is seen to split into the product of two marginal likelihoods:
    $$
    \begin{aligned}
    p( \bm{y}, \bm{y}_{0}|\bm{\epsilon_0}) &= \underbrace{\int \Pi(\bd f) \left[ \prod_{i=1}^n f(y_i) \prod_{j : \epsilon_{0j} = 1} f(y_{0j}) \right]}_{p(\bm{y}, \bm{y}_{0, \text{exch}}|\bm{\epsilon_0})} \times \underbrace{\int \Pi(\bd g) \prod_{j : \epsilon_{0j} = 0} g(y_{0j}) dg}_{p(\bm{y}_{0, \text{nonexch}}|\bm{\epsilon_0})}
    \end{aligned}
    $$
    where $\bm{y}_{0, \text{exch}} = \{ y_{0j} : \epsilon_{0j} = 1 \}$ and $\bm{y}_{0, \text{nonexch}} = \{ y_{0j} : \epsilon_{0j} = 0 \}$ are the subsets of historical observations that are exchangeable and non-exchangeable, respectively.
    \item Finally note that the posterior distributions of $f$ and $g$ given $\bm{\epsilon}_0$ and the data $\bm{y}, \bm{y_0}$ have the form
    $$
    P(\bd f | \bm{y}, \bm{y}_{0, \text{exch}}, \bm{\epsilon}_0) = \frac{\Pi(\bd f) \prod_{i=1}^n f(y_i) \prod_{j : \epsilon_{0j} = 1} f(y_{0j})}{p(\bm{y}, \bm{y}_{0, \text{exch}}|\bm{\epsilon_0})}
    $$
    and
    $$
    P(\bd g | \bm{y}, \bm{y}_{0, \text{nonexch}}, \bm{\epsilon}_0) = \frac{\Pi(\bd g) \prod_{j : \epsilon_{0j} = 0} g(y_{0j})}{p(\bm{y}_{0, \text{nonexch}}|\bm{\epsilon_0})}
    $$
     respectively.
\end{itemize}

Since $P(\bm{\epsilon_0}|\bm{y}, \bm{y}_0) \propto p(\bm{y}, \bm{y}_0|\bm{\epsilon_0}) \Pi(\bm{\epsilon}_0)$ for a proportionality constant that only depends on $\bm{y}$ and $\bm{y}_0$, we see that the left-hand side of \eqref{eq:posterior_factorization} is proportional (and thus equal) to the probability distribution:
$$
P(\bd f, \bd g, \bd \gamma, \bm{\epsilon}_0 | \bm{y}, \bm{y}_0) = P(\bd f | \bm{y}, \bm{y}_{0, \text{exch}}, \bm{\epsilon}_0) \times P(\bd g | \bm{y}, \bm{y}_{0, \text{nonexch}}, \bm{\epsilon}_0) \times P(\bd \gamma | \bm{\epsilon_0}) \times P(\bm{\epsilon_0}|\bm{y}, \bm{y}_0).
$$

Marginalization with respect to $g$ and $\gamma$ for every fixed $f$ and $\bm{\epsilon}_0$ yields
$$
    P(\bd f, \bm{\epsilon}_0 | \bm{y}, \bm{y}_0) = P(\bd f | \bm{y}, \bm{y}_{0, \text{exch}}, \bm{\epsilon}_0) \times P(\bm{\epsilon_0}|\bm{y}, \bm{y}_0)
$$
and thus
$$
    P(\bd f| \bm{y}, \bm{y}_0) = \sum_{\bm{\epsilon}_0 \in \{0, 1\}^{n_0}}
    P(\bd f | \bm{y}, \bm{y}_{0, \text{exch}}, \bm{\epsilon}_0) \times P(\bm{\epsilon}_0 | \bm{y}, \bm{y}_0).
$$

Thus, it follows that the nonparametric extension of the LEAP is equivalent to Bayesian model averaging over all possible partitions of the historical data into exchangeable and non-exchangeable groups.

\section{Proof of clustering scheme}
Let $\epsilon_{1i} = 1$ since all current data individuals are
guaranteed to be exchangeable. Throughout this proof, $F$ and $G$
appear as subscript labels on counts such as $N_{0k}$ (exchangeable historical, cluster $k$) and
$M_{0l}$ (nonexchangeable historical, cluster $l$), where $k$ indexes
exchangeable clusters and $l$ indexes nonexchangeable clusters.

The strategy is to consider the limiting case of a two-part mixture of Bayesian clustering models, i.e., consider the mixture model $F(\cdot | \bm{\theta}) = \sum_{k=1}^{K} w_{k} F(\cdot | \bm{\theta}_k)$ where $\bm{w} = (w_{1}, \ldots, w_{K})' \sim \text{Dirichlet}(1/K, \ldots, 1/K)$ with prior $\pi(\bm{\theta}) = \prod_{k=1}^{K} \pi(\bm{\theta}_k)$. Let $G(\cdot | \bm{\lambda})$, $\{ (\rho_l, \bm{\lambda}_{l}')', l = 1, \ldots, L \}$, and relevant priors be analogously defined.

As $\epsilon_{di}$ denotes membership to the exchangeable distribution $F$, let $z_{di} \in \{1,\dots,K_{\epsilon_{di}}\}$ denote to which component individual $i$ from data set $d$ belongs. Let $\bm{\Omega} = (\gamma, \bm{\theta}', \bm{\lambda}')$ and let $\bm{\delta} = (\alpha, \eta,\bm{\epsilon}_1, \bm{\epsilon}_0, \bm{z}_1, \bm{z}_0)$ the class and cluster-within-class membership variables. The joint prior of the two-part mixture of mixtures model is given by
\begin{align}
    \pi(\bm{\Omega}, \bm{\delta})
    &=
    \frac{
        \gamma^{N_{0} + c_1 - 1}(1 - \gamma)^{M_{0} + c_0 - 1}
    }{
        B(c_1, c_0)
    }
    \notag \\ &\quad \times
    \frac{\Gamma(\alpha)}{\Gamma(\alpha / K)^{K}}
    \prod_{k=1}^{K}
        w_{k}^{\,\tilde{N}_k + \alpha / K - 1}
        \,\pi(\bm{\theta}_{k})
    \notag \\ &\quad \times
    \frac{\Gamma(\eta)}{\Gamma(\eta / L)^{L}}
    \prod_{l=1}^{L}
        \rho_{l}^{\,M_{0l} + \eta / L - 1}
        \,\pi(\bm{\lambda}_{l}),
    \label{eq:fmm_jointprior}
\end{align}
where, $\tilde{N}_k = \sum_{d=0}^1 \sum_{i=1}^{n_d}
  I(\epsilon_{di} = 1, z_{di} = k)$
is the total number of individuals assigned to cluster $k$ of
the exchangeable class for $k = 1, \ldots, K$,
$M_{0l} = \sum_{j=1}^{n_0} I(\epsilon_{0j} = 0, z_{0j} = l)$
for $l = 1, \ldots, L$.

Integrating out $\bm{\Omega}$ from \eqref{eq:fmm_jointprior}, we obtain the marginal prior on class memberships and cluster parameters, which is given by
\begin{align}
    \pi(\bm{\delta}) &= \frac{ B( N_{0} + c_1, M_{0} + c_0) }{ B(c_1, c_0) }
    \notag \\ &\quad\times
    \frac{\Gamma(\alpha)}{ \Gamma(\tilde{N} + \alpha) }
    \prod_{k=1}^{K} \frac{\Gamma(\tilde{N}_k + \alpha / K)}{ \Gamma(\alpha / K) }
    \notag \\ &\quad\times
    \frac{\Gamma(\eta)}{ \Gamma(M_0 + \eta) }
    \prod_{l=1}^{L} \frac{\Gamma(M_{0l} + \eta / L)}{ \Gamma(\eta / L) },
    \label{eq:class_joint_prior}
\end{align}

\subsection{Prior clustering scheme for the current data}

Let $N_{0k}^{[1:(i-1)]}$ denote the number of exchangeable historical
individuals assigned to cluster $k$ among the first $i-1$ historical
individuals, so that $N_0 = \sum_{k=1}^K N_{0k}$ is the total number
of exchangeable historical individuals. Let $n_k = \sum_{i=1}^n I(z_i = k)$
denote the number of current-data individuals in cluster $k$, with
$\sum_{k=1}^K n_k = n$. Since all current data individuals are
exchangeable by definition, no current individual belongs to a
nonexchangeable cluster. In the sequel, the superscript $[1:(i-1)]$
denotes quantities based on the first $i-1$ historical individuals and
$(-i)$ denotes all individuals except individual $i$.

The prior in \eqref{eq:class_joint_prior} gives the joint prior for class memberships and cluster-within-class memberships. We now focus on the current data clustering. The prior probability of assigning individual $i$ to cluster $k$ is given by
\begin{align}
    \pi(z_{i} = k | \bm{z}_{1:(i-1)}) = \frac{\pi(z_i = k, \bm{z}_{1:(i-1)})}{\pi(\bm{z}_{1:(i-1)})}
    = \frac{ \frac{ \Gamma(n_k^{[1:(i-1)]} + 1 + \alpha / K)} { \Gamma(i+\alpha) } }{ \frac{ \Gamma(n_k^{[1:(i-1)]} + \alpha / K)} { \Gamma(i-1+\alpha) } }
    =
    \frac{n_k^{[1:(i-1)]} + \alpha / K }{ i - 1 + \alpha }.
    \label{eq:prior_cluster_current}
\end{align}
The expression in \eqref{eq:prior_cluster_current} is precisely the same as that under a Dirichlet clustering model using only the current data. Hence, taking limits as $K \to \infty$ produces precisely the same DPMM prior clustering model for the current data, which is the so-called ``rich get richer'' scheme since the probability of being assigned to an existing cluster is proportional to the number of members in that cluster.

\subsection{Prior clustering scheme for the historical data}
Now, we wish to express the prior clustering and class assignment
probabilities for the historical data given the current data clusters.

Conditional on assignments for the current data, the prior probability
of clustering individual $i$ from the historical data is given by
\begin{align}
    \pi(\epsilon_{0i} = \ell,\, z_{0i} = k \mid
        \bm{z},\, \bm{z}_0^{[1:(i-1)]},\, \bm{\epsilon}_0^{[1:(i-1)]})
    =
\begin{cases}
    \dfrac{N_{0}^{[1:(i-1)]} + c_1}{ i - 1 + c_1 + c_0 }
    \cdot
    \dfrac{n_k + N_{0k}^{[1:(i-1)]} + \alpha / K}{
           n + N_{0}^{[1:(i-1)]} + \alpha}
    & \text{if } \ell = 1,
    \\[10pt]
    \dfrac{M_{0}^{[1:(i-1)]} + c_0}{ i - 1 + c_1 + c_0 }
    \cdot
    \dfrac{M_{0l}^{[1:(i-1)]} + \eta / L}{
           M_{0}^{[1:(i-1)]} + \eta}
    & \text{if } \ell = 0.
\end{cases}
    \label{eq:prior_cluster_hist1}
\end{align}

Taking limits of \eqref{eq:prior_cluster_hist1} as $K, L \to \infty$
yields
\begin{align}
    \pi(\epsilon_{0i} = 1,\, z_{0i} = k \mid \cdots)
    &=
    \frac{N_{0}^{[1:(i-1)]} + c_1}{i - 1 + c_1 + c_0}
    \cdot
    \begin{cases}
        \dfrac{n_k + N_{0k}^{[1:(i-1)]}}{n + N_{0}^{[1:(i-1)]} + \alpha}
        & \text{existing cluster } k \\[8pt]
        \dfrac{\alpha}{n + N_{0}^{[1:(i-1)]} + \alpha}
        & \text{new cluster}
    \end{cases}
    \label{eq:prior_cluster_hist_limit_exch}
    \\[10pt]
    \pi(\epsilon_{0i} = 0,\, z_{0i} = l \mid \cdots)
    &=
    \frac{M_{0}^{[1:(i-1)]} + c_0}{i - 1 + c_1 + c_0}
    \cdot
    \begin{cases}
        \dfrac{M_{0l}^{[1:(i-1)]}}{M_{0}^{[1:(i-1)]} + \eta}
        & \text{existing cluster } l \\[8pt]
        \dfrac{\eta}{M_{0}^{[1:(i-1)]} + \eta}
        & \text{new cluster}
    \end{cases}
    \label{eq:prior_cluster_hist_limit_nonexch}
\end{align}
where, for brevity, the ellipsis ($\cdots$) denotes all cluster and
exchangeability assignments thus far.

\subsection{Full conditional posterior clustering scheme}
\subsubsection{Current data}
Assume the current occupied clusters after removing individual $i$ have parameters $(\bm{\theta}_{1}, \ldots, \bm{\theta}_{i})$. Let the base measure for $\bm{\theta}$ be given by $P_0$. For an existing cluster $k$, applying the derived prior above with the likelihood contribution for individual $i$ yields.
$$
    p(z_i = k | \cdots) \propto 
    \begin{cases}
        \left( n_k^{(-i)} + N_{0k} \right) f(y_i | \bm{\theta}_k) & \text{ for an existing cluster } k \in \{1, \ldots, K_i\}
        \\
        \alpha \int f(y_i | \bm{\theta}) dP_0(\bm{\theta}) & \text{ for a new cluster } k = K_i + 1
    \end{cases}
$$
where $n_k^{(-i)}$ denotes the number of individuals from the current data assigned to cluster $k$ of the exchangeable density after removing individual $i$ and $N_{0k}$ denotes the total number of individuals from the historical data assigned to that cluster.

\subsubsection{Historical data}
For historical individual $j$, let $K_j$ and $L_j$ denote the
number of occupied exchangeable and nonexchangeable clusters after
removing individual $j$ from the historical data, and let $N_{0}^{(-j)}$,
$N_{0k}^{(-j)}$ denote the corresponding class and
cluster-within-class counts with $M_{0}^{(-j)}$ and $M_{0l}^{(-j)}$ similarly defined.

\paragraph{Case $\epsilon_{0j} = 1$ (exchangeable).}
\begin{align}
    p\!\left(\epsilon_{0j}=1,\, z_{0j}=k \mid \cdots\right) \propto
    \left(N_{0}^{(-j)} + c_1\right) \times
    \begin{cases}
        \dfrac{n_k + N_{0k}^{(-j)}}{n + N_{0}^{(-j)} + \alpha}
        \, f(y_{0j} \mid \bm{\theta}_k)
        & k \in \{1,\ldots,K_j \}
        \\[8pt]
        \dfrac{\alpha}{n + N_{0}^{(-j)} + \alpha}
        \displaystyle\int f(y_{0j} \mid \bm{\theta})\, dP_0(\bm{\theta})
        & k = K_j + 1
    \end{cases}
    \label{eq:fullcond_hist_exch}
\end{align}

\paragraph{Case $\epsilon_{0j} = 0$ (nonexchangeable).}
\begin{align}
    p\!\left(\epsilon_{0j}=0,\, z_{0j}=l \mid \cdots\right) \propto
    \left(M_{0}^{(-j)} + c_0\right) \times
    \begin{cases}
        \dfrac{M_{0l}^{(-j)}}{M_{0}^{(-j)} + \eta}
        \, g(y_{0j} \mid \bm{\lambda}_l)
        & l \in \{1,\ldots,L_j\}
        \\[8pt]
        \dfrac{\eta}{M_{0}^{(-j)} + \eta}
        \displaystyle\int g(y_{0j} \mid \bm{\lambda})\, dQ_0(\bm{\lambda})
        & l = L_j + 1
    \end{cases}
\end{align}

\section{MCMC Algorithm Derivation}
\label{sec:mcmc_alg}

The joint posterior distribution is given by
\begin{align}
    p(\bm{\Omega} | \bm{y}, \bm{y}_0) \propto & \pi(\bm{\Omega}) 
    \left\{ \prod_{i=1}^{n} \prod_{k=1}^{K}\left( w_k f(y_i | \bm{\theta}_k) \right)^{z_{ik}}
    \right\}
    \notag \\ & \times
    \left\{ \prod_{j=1}^{n_0}  \left[ \gamma \prod_{k=1}^{K} \left( w_k f(y_{0j} | \bm{\theta}_k) \right)^{z_{0jk}} \right]^{\epsilon_{0j}}
    \left[ (1 - \gamma) \prod_{l=1}^{L} \left( \rho_l g(y_{0j} | \bm{\lambda}_l) \right)^{z_{0jl}} \right]^{1 - \epsilon_{0j}}
    \right\}
    .
    \label{eq:npleap_joint}
\end{align}
where
\begin{itemize}
    \item $\bm{\Omega}$ contains all model parameters, cluster assignments, and class assignments
    \item $w_k = \nu_k \prod_{l < k} (1 - \nu_l)$ where $\nu_k \sim \text{Beta}(1, \alpha)$, $k = 1, \ldots, K$
    \item $\rho_{l} = \psi_l \prod_{m < l} (1 - \psi_l)$, $\psi_l \sim \text{Beta}(1, \eta)$
    \item $z_{ik} \in \{1, \ldots, K\}$ denotes to which cluster individual $i$ in the current data belongs
    %
    \item $\epsilon_{0j} \in \{0, 1\}$ denotes the class of historical individual $j$, with $\epsilon_{0j} = 1$ indicating exchangeability with the current data and $\epsilon_{0j} = 0$ indicating nonexchangeability
    \item $z_{0jk}$ denotes to which cluster individual $j$ in the historical data belongs (the number of such clusters depends on the value of $\epsilon_{0j}$)
    \item $z_{ik} = I(z_i = k)$, $\epsilon_{0jl} = I(\epsilon_{0j} = l)$, $z_{0jk} = I(z_{0j} = k)$.
\end{itemize}
\begin{align}
&\pi(\bm{\Omega}) = \left[ p_0 \cdot \frac{\gamma^{c_1-1}(1-\gamma)^{c_0-1}}{B(c_1,c_0)} + (1-p_0)\delta_1(\gamma) \right] \\
& \times \left[ \prod_{k=1}^{K-1} \alpha(1 - v_k)^{\alpha-1} \right]
\left[ \prod_{k=1}^{K} \pi(\bm{\theta}_k) \right]
\times \left[ \prod_{l=1}^{L-1} \eta(1 - \psi_l)^{\eta-1} \right]
\left[ \prod_{l=1}^{L} \pi(\bm{\lambda}_l) \right]
\end{align}

Define:
\begin{itemize}
    %
    \item Let $N_{0} = \sum_{j=1}^{n_0} \epsilon_{0j}$ denote the number of exchangeable historical data individuals and $M_{0} = \sum_{j=1}^{n_0} (1 - \epsilon_{0j})$ the number of nonexchangeable historical data individuals.
    \item $\tilde{N}_k = n_k + N_{0k}$ the \emph{total} number of individuals assigned to cluster $k$ of class $F$ (across current and historical data individuals)
    \item $\tilde{N} = \sum_{k=1}^{K} \tilde{N}_k = n + N_0$ the total
    number of exchangeable individuals (current and historical combined)
    \item $D^{\text{exch}}_k = \{ y_i : z_k = k \} \cup \{ y_{0j} : \epsilon_{0j} = 1, z_{0j} = k \}$ the data for current data and exchangeable historical data individuals assigned to cluster $k$ of class $F$
    \item $D^{\text{unexch}}_l = \{ y_{0j} : \epsilon_{0j} = 0, z_{0j} = l \}$ the data for unexchangeable historical data individuals assigned to cluster $l$
\item $L(\bm{\theta}_k \mid f, D^{\mathrm{exch}}_k)
  = \prod_{i:\,z_i=k} f(y_i \mid \bm{\theta}_k)
  \times \prod_{j:\,\epsilon_{0j}=1,\,z_{0j}=k} f(y_{0j} \mid \bm{\theta}_k)$
  the likelihood for exchangeable cluster $k$.

\item $L(\bm{\lambda}_l \mid g, D^{\mathrm{unexch}}_l)
  = \prod_{j:\,\epsilon_{0j}=0,\,z_{0j}=l} g(y_{0j} \mid \bm{\lambda}_l)$
  the likelihood for nonexchangeable cluster $l$.
\end{itemize}
Note that we can write the posterior in \eqref{eq:npleap_joint} as
\begin{align}
    p(\bm{\Omega} | \bm{y}, \bm{y}_0) \propto \pi(\bm{\Omega})
        \gamma^{N_{0}} (1 - \gamma)^{M_{0}} \prod_{k=1}^{K}
        w_k^{\tilde{N}_k}
        \left\{ L(\bm{\theta}_k | f, D^{\text{exch}}_k) \right\}
        \prod_{l=1}^{L} \left\{ \rho_l^{M_{0l}} L(\bm{\lambda}_l | g, D^{\text{unexch}}_l) \right\}
\end{align}

\subsection{Full conditional of $\gamma$}
We have
\begin{align}
    p(\gamma | \cdots) &= \frac{%
  p_0 \cdot \frac{\gamma^{c_1+N_0-1}\,(1-\gamma)^{c_0+M_0-1}}{B(c_1,c_0)}
  + (1-p_0)\,\delta_1(\gamma)\,I(M_0=0)
}{%
  \int_0^1\!\left[
    p_0 \cdot \frac{\gamma_*^{c_1+N_0-1}\,(1-\gamma_*)^{c_0+M_0-1}}{B(c_1,c_0)}
    + (1-p_0)\,I(\gamma_*=1)\,I(M_0=0)
  \right]d\gamma_*
}
    , \notag \\
    &= \frac{ p_0 \cdot \frac{ \gamma^{c_1 + N_{0} - 1} \times (1 - \gamma)^{c_0 + M_{0} - 1} }{B(c_1, c_0)} + (1 - p_0) \delta_1(\gamma) I(M_{0} = 0) }{ p_0 \frac{ B(c_1 + N_{0}, c_0 + M_{0}) }{ B(c_1, c_0) } + ( 1- p_0)I(M_{0} = 0) }
    , \notag \\
    &=
    \begin{cases}
       \frac{ p_0 \cdot \frac{ \gamma^{c_1 + N_{0} - 1} \times (1 - \gamma)^{c_0 + M_{0} - 1} }{B(c_1, c_0)} + (1 - p_0) \delta_1(\gamma) }{ p_0 \frac{ B(c_1 + N_{0}, c_0 + M_{0}) }{ B(c_1, c_0) } + (1- p_0) }
       & \text{ if } M_{0} = 0
       \\
              \frac{ p_0 \cdot \frac{ \gamma^{c_1 + N_{0} - 1} \times (1 - \gamma)^{c_0 + M_{0} - 1} }{B(c_1, c_0)}  }{ p_0 \frac{ B(c_1 + N_{0}, c_0 + M_{0}) }{ B(c_1, c_0) } }
       & \text{ if } M_{0} > 0
    \end{cases}
    , \notag \\
    &= \begin{cases}
        \tilde{p}_0 f_{\beta}(c_1 + N_{0}, c_0 + M_{0}) + (1 - \tilde{p}_0) \delta_1(\gamma)
        & \text{ if } M_{0} = 0
        \\
        f_{\beta}(c_1 + N_{0}, c_0 + M_{0})
        &\text{ if } M_{0} > 0
    \end{cases}
    , \notag \\
    &= \tilde{p}_0^{I(M_{0} = 0)} f_{\beta}(\gamma | c_1 + N_{0}, c_0 + M_{0}) + \left[ (1 - \tilde{p}_0) \cdot \delta_1(\gamma) \right]^{ I(M_{0} = 0) }
\end{align}
where $\tilde{p}_0 = \dfrac{ p_0 \cdot \frac{ B(c_1 + N_{0}, c_0 + M_{0}) }{ B(c_1, c_0) } }{ p_0 \cdot \frac{ B(c_1 + N_{0}, c_0 + M_{0}) }{ B(c_1, c_0) } + 1 - p_0 }$.

\subsection{Full conditional of DP concentration parameters}
We have
$p(\alpha, \eta | \cdots) \propto p(\alpha | \cdots) p(\eta | \cdots)$
where
\begin{align}
    p(\alpha | \cdots)
    &\propto \pi(\alpha) \prod_{k=1}^{K-1} \pi(v_k | \alpha)
    \notag \\
    &\propto \alpha^{a_{\alpha} - 1} e^{-b_\alpha \alpha}
        \prod_{k=1}^{K-1} \alpha (1 - v_k)^{\alpha - 1}
    \notag \\
    &= \alpha^{a_{\alpha} + (K-1) - 1}
        \exp\!\left\{
            -\alpha \left[ b_\alpha - \sum_{k=1}^{K-1} \log(1 - v_k) \right]
        \right\}
    \notag \\
    &\propto f_{\Gamma}\!\left(
        \alpha \,\Big|\,
        a_{\alpha} + K - 1,\;
        b_\alpha - \sum_{k=1}^{K-1} \log(1 - v_k)
    \right).
\end{align}
By a direct analogy, we have
\begin{align}
    p(\eta | \cdots)
    \propto f_{\Gamma}\!\left(
        \eta \,\Big|\,
        a_{\eta} + L - 1,\;
        b_\eta - \sum_{l=1}^{L-1} \log(1 - \psi_l)
    \right).
\end{align}

\subsection{Full conditional of stick-breaking variables}
Recall that $w_k = v_k \prod_{m < k} (1 - v_m)$ where
$v_k \overset{\mathrm{iid}}{\sim} \mathrm{Beta}(1, \alpha)$ for
$k = 1, \ldots, K-1$ (with $v_K = 1$ by convention), and
$\rho_l = \psi_l \prod_{m < l} (1 - \psi_m)$ where
$\psi_l \overset{\mathrm{iid}}{\sim} \mathrm{Beta}(1, \eta)$
for $l = 1, \ldots, L-1$ (with $\psi_L = 1$ by convention).
Since the stick-breaking variables for the two components enter the
posterior independently, we have
\begin{align}
    p(\bm{v}, \bm{\psi} | \cdots)
    \propto p(\bm{v} | \cdots)\, p(\bm{\psi} | \cdots).
\end{align}

\paragraph{Derivation for $\bm{v}$.}
The relevant factors in the posterior are the
$\mathrm{Beta}(1, \alpha)$ prior on each $v_k$ and the
likelihood contribution through $\prod_{k=1}^{K} w_k^{\tilde{N}_k}$.
Writing $w_k = v_k \prod_{m < k}(1-v_m)$ explicitly:
\begin{align}
    \prod_{k=1}^{K} w_k^{\tilde{N}_k}
    &= v_1^{\tilde{N}_1}
       \prod_{k=2}^{K-1}
         \left( v_k \prod_{m=1}^{k-1}(1-v_m) \right)^{\!\tilde{N}_k}
       \times
       \prod_{m=1}^{K-1}(1-v_m)^{\tilde{N}_K}
    \notag \\
    &= \prod_{k=1}^{K-1}
         v_k^{\tilde{N}_k}
         (1-v_k)^{\,\sum_{m=k+1}^{K} \tilde{N}_{m}}.
\end{align}
The second equality follows by collecting, for each $k \in
\{1,\ldots,K-1\}$, the power of $v_k$ (which is $\tilde{N}_k$)
and the power of $(1-v_k)$ (which accumulates one factor of
$\tilde{N}_{m}$ for every $m > k$).
Multiplying by the prior kernel $(1-v_k)^{\alpha - 1}$ gives
\begin{align}
    p(\bm{v} | \cdots)
    &\propto
    \prod_{k=1}^{K-1}
      v_k^{\tilde{N}_k}
      (1-v_k)^{\,\alpha + \sum_{m=k+1}^{K} \tilde{N}_{m} - 1},
\end{align}
which factors into independent Beta kernels.  Hence
\begin{align}
    v_{k} \mid \cdots
    \sim
    \mathrm{Beta}\!\left(
        \tilde{N}_k + 1,\;
        \alpha + \sum_{m=k+1}^{K} \tilde{N}_{m}
    \right),
    \quad k = 1, \ldots, K-1.
\end{align}

\paragraph{Derivation for $\bm{\psi}$.}
An entirely analogous argument applies to the nonexchangeable
component.  The relevant likelihood contribution is
$\prod_{l=1}^{L}\rho_l^{M_{0l}}$, and the prior kernel for each
$\psi_l$ is $(1-\psi_l)^{\eta - 1}$.  Expanding
$\rho_l = \psi_l \prod_{m<l}(1-\psi_m)$ and collecting powers
exactly as above gives
\begin{align}
    p(\bm{\psi} | \cdots)
    &\propto
    \prod_{l=1}^{L-1}
      \psi_l^{M_{0l}}
      (1-\psi_l)^{\,\eta + \sum_{m=l+1}^{L} M_{0m} - 1},
\end{align}
so that
\begin{align}
    \psi_{l} \mid \cdots
    \sim
    \mathrm{Beta}\!\left(
        M_{0l} + 1,\;
        \eta + \sum_{m=l+1}^{L} M_{0m}
    \right),
    \quad l = 1, \ldots, L-1.
\end{align}

\subsection{Full conditional of density parameters $(\theta, \lambda)$}
We have
\begin{align}
    p(\bm{\theta}, \bm{\lambda} | \dots) &\propto \pi(\bm{\theta}, \bm{\lambda} | \bm{\Omega} \setminus \{ \bm{\theta}, \bm{\lambda} \})
    \left[ \prod_{k=1}^{K} L(\bm{\theta}_k | D^{\text{exch}}_k) \right]
    \left[ \prod_{l=1}^{L} L(\bm{\lambda}_l | D^{\text{unexch}}_l) \right]
\end{align}
typically, $\pi(\bm{\theta}, \bm{\lambda} | \bm{\Omega} \setminus \{ \bm{\theta}, \bm{\lambda} \}) = \left[ \prod_{k=1}^{K} \pi(\bm{\theta}_k) \right] \left[ \prod_{l=1}^{L} \pi(\bm{\lambda}_l) \right]$ in which case we can write
\begin{align}
    p(\bm{\theta}, \bm{\lambda} | \dots) &\propto 
    \left[ \prod_{k=1}^{K} p(\bm{\theta}_k | D^{\text{exch}}_k) \right] \left[ \prod_{l=1}^{L} p(\bm{\lambda}_l | D^{\text{unexch}}_l) \right],
\end{align}
where $p(\bm{\theta}_k | D^{\text{exch}}_k)$ is the posterior density for $\bm{\theta}_k$ based on prior $\pi(\bm{\theta}_k)$, data $D^{\text{exch}}_k$, and kernel $f$ and $p(\bm{\lambda}_l | D^{\text{unexch}}_l)$ is the posterior for $\bm{\lambda}_l$ based on prior $\pi(\bm{\theta}_l)$, kernel $g$, and data $D^{\text{unexch}}_l$

\subsection{Full conditional of $z_i$}
We have
\begin{align}
    p(z_i = k | \cdots) \propto w_{k} f(y_i | \bm{\theta}_k) \propto \frac{ w_{k} f(y_i | \bm{\theta}_k) }{\sum_{k'=1}^{K} w_{k'} f(y_i | \bm{\theta}_{k'}) }.
\end{align}

\subsection{Full conditional of $\epsilon_{0j}$ and $z_{0j}$}
We have
\begin{align}
    p(\epsilon_{0j} = \ell, z_{0j} = k | \cdots) \propto 
     \begin{cases}
       \dfrac{\gamma w_{k} f(y_{0j} | \bm{\theta}_{k})}{
       \gamma \sum_{k'=1}^{K} w_{k'} f(y_{0j} | \bm{\theta}_{k'}) +
       (1 - \gamma) \sum_{l'=1}^{L} \rho_{l'} g(y_{0j} | \bm{\lambda}_{l'}) }
       & \ell = 1,\; k = 1, \ldots, K
        \\[10pt]
        \dfrac{(1 - \gamma) \rho_{l} g(y_{0j} | \bm{\lambda}_{l})}{
        \gamma \sum_{k'=1}^{K} w_{k'} f(y_{0j} | \bm{\theta}_{k'}) +
        (1 - \gamma) \sum_{l'=1}^{L} \rho_{l'} g(y_{0j} | \bm{\lambda}_{l'}) }
        & \ell = 0,\; l = 1, \ldots, L
    \end{cases}
    \label{eq:fullcond_cat}
\end{align}

\subsection{MCMC Algorithm}
\begin{enumerate}
    \item Sample $\gamma \sim \tilde{p}_0 \cdot \text{Beta}(c_1 + N_{0}, c_0 + M_{0}) + (1 - \tilde{p}_0) \cdot I(\gamma = 1)$ where 
    $
        \tilde{p}_0 = \left[\frac{ p_0 \cdot \frac{ B(c_1 + N_{0}, c_0 + M_{0}) }{ B(c_1, c_0) } }{ p_0 \cdot \frac{ B(c_1 + N_{0}, c_0 + M_{0}) }{ B(c_1, c_0) } + 1 - p_0 } \right]^{I(M_{0} = 0)}
    .
    $
    \item Sample $\alpha \sim \text{Gamma}\left( a_{\alpha} + K - 1,  
    b_\alpha - \sum_{k=1}^{K - 1} \log(1 - v_k)
    \right)$.
    \item Sample $\eta \sim \text{Gamma}\left( a_{\eta} + L - 1,  
    b_\eta - \sum_{l=1}^{L - 1} \log(1 - \psi_l)
    \right)$.
    \item Sample $v_k \sim \text{Beta} \left( \tilde{N}_k + 1, \alpha + \sum_{m = k + 1}^{K} \tilde{N}_{m} \right) $, \ \ $k = 1, \ldots, K-1$, where $\tilde{N}_k$ denotes the total number of individuals assigned to exchangeable cluster $k$. Compute $w_k = v_k \prod_{m < k} (1 - v_m)$
    \item Sample $\psi_l \sim \text{Beta} \left( M_{0l} + 1, \eta + \sum_{m = l + 1}^{L} M_{0m} \right) $, \ \ $l = 1, \ldots, L-1$, where $M_{0l}$ denotes the total number of historical data individuals assigned to $G$-cluster $l$. Compute $\rho_l = \psi_l \prod_{m < l} (1 - \psi_m)$
    \item Sample $\bm{\theta}_{k} \sim p(\bm{\theta}_{k} | D^{\text{exch}}_k)$ for $k = 1, \ldots, K$ and $\bm{\lambda}_l \sim p(\bm{\lambda}_l | D^{\text{unexch}}_l)$ for $l = 1, \ldots, L$.
    \item For $i = 1, \ldots, n$, sample $z_i$ from a categorical distribution with 
    $$
    \Pr(z_i = k) = \frac{w_{k} f(y_i | \bm{\theta}_k)}{\sum_{k'=1}^{K} w_{k'} f(y_i | \bm{\theta}_{k'}) }
    , \ \ k = 1, \ldots, K
    $$
    \item For $j = 1, \ldots, n_0$, sample $(\epsilon_{0j}, z_{0j})$ via the categorical distribution given in \eqref{eq:fullcond_cat}.
\end{enumerate}

\subsection{Hyperparameter elicitation for the ANOVA DDP model}
\label{sec:hyperpar}
For the exchangeable component base measure, we elicit $\bm{\mu}_0$ to be the maximum likelihood estimator of a log-normal AFT model and $\bm{\Sigma}_0$ to be the inverse of the observed information matrix discounted to 2 individuals, i.e., $\bm{\Sigma}_0 = [\mathcal{J}^{-1}]_{\beta\beta} \times \frac{n}{2}$, where $[\mathcal{J}^{-1}]_{\beta\beta}$ is the inverse observed information matrix for the regression coefficients. For the precision base measure, we elicit $\delta_{\alpha}, \kappa_{\alpha}$ to minimize the KL divergence between the log-normal approximation to $\tau$ based on the MLE with the discounted information matrix and a Gamma distribution; the Gamma base measure maintains semi-conjugacy and eases computational speed \citep{roy2018bayesian}. For the nonexchangeable component base measure, we use $\bm{\mu}_{Ql} = \bm{0}$ and $\bm{\Sigma}_{Ql} = \text{diag}\{10^2, 10^2\}$, with $\tau_{Ql} \sim \text{Gamma}(0.5, 0.5)$. Using a diffuse prior for $G$ avoids introducing artificial preference for assigning new clusters to the exchangeable group when the historical data are in fact exchangeable with the current data.

\section{Data generation process}
\subsection{Proportional hazards model}
To mimic the real data sets, we fit a spline-based proportional hazards model to the pooled current and historical data with time-varying effects for study ID (historical or current) and treatment arm controlling for baseline age, sex, and BMI. This results in a conditional cumulative hazard function given by
\begin{align}
H(t | s, a, \bm{x}) = \exp\left\{ 
      \sum_{j=0}^4 \alpha_j z^{(k)}_j(\log t) 
    + I(s = 0) [ \beta_{\text{hist}} + \beta_1 z^{(k)}_j(\log t)]
    + \sum_{j=0}^1 \gamma_j a z^{(k)}_j(\log t)
    + \bm{x}'\bm{\delta}
\right\}
,
\label{eq:spline_cond_hazard}
\end{align}
where $H(t | a, \bm{x}, s)$ is the conditional cumulative hazard for arm $a$ of study $s$, $z_j^{(k)}(q)$ is the $j^{th}$ basis function for a spline of degree $k$ based on continuous variable $q$, and $\bm{x}$ includes sex, BMI, and weight as baseline covariates. The MLE as fit by the \texttt{flexsurv} package in R \citep{jackson2016flexsurv} is given in Table~\ref{tab:spline_mle}. As will be seen below, the maximum likelihood analysis will serve as the basis for generating the event times for the simulation. The population survival functions based on the data generation process are presented in Figure S1 in the Supplementary Materials.

\begin{table}[ht]
\centering
\begin{tabular}{cccccccccc} 
\toprule
$\hat{\alpha}_0$ & $\hat{\alpha}_1$ & $\hat{\alpha}_2$ & $\hat{\alpha}_3$ & $\hat{\beta}_{\text{hist}}$ & $\hat{\beta}_1$ & $\hat{\gamma}_0$ & $\hat{\gamma}_1$ & $\hat{\delta}_{\text{sex}}$ & $\hat{\delta}_{\text{BMI}}$ \\ 
\midrule
1.00 & -0.21 & 0.28 & -0.01 & -0.26 & 0.19 & 0.32 & 0.16 & 0.65 & 0.02 \\
\bottomrule
\end{tabular}
\caption{Maximum likelihood estimates from fitting a proportional hazards model using splines with $k = 3$ knots; the $\alpha_j$'s are time effects, the $\beta_j$'s are historical study-specific effects, the $\gamma_j$'s are treatment effects, and the $\delta$'s are baseline covariate effects.}
\label{tab:spline_mle}
\end{table}

\subsection{Current data generation}
Baseline covariates $\bm{X}_i$ are generated using the  nonparametric boostrap distribution (i.e., sampling with replacement from the current data set). Treatment assignment $A_i$ was generated via permuted block randomization. Given the block $\bm{b} = (1, 1, 0)'$, for each 3 individuals, we (i) permute the elements of $\bm{b}$ to obtain $\tilde{\bm{b}}$ and (ii) assign the sequential 3 individuals according to $\tilde{\bm{b}}$. Based on the current data sample size of $n = 104$, this guarantees that the number of treated individuals is 69 or 70 and is more commonly used than simple randomization in practice.

Overall survival times $T_i$ were generated based on the survival model in \eqref{eq:spline_cond_hazard} using the MLEs in Table~\ref{tab:spline_mle} as the true values. Dropout times $C_i$ were generated from an exponential distribution where the rate is arm-specific, i.e., $C_i | A_i = a \sim \text{Exponential}(\zeta_a)$ where $\zeta_0 = 0.0797$ and $\zeta_1 \approx 0.0878$. The values were determined so that approximately 5\% of individuals would drop out of the study after 1 year, which was obtained by solving the optimization problem $\eta_{a} = \text{arg min}_{\Omega} \left\{ \left| \int_0^1 S(c | a) \Omega e^{-c\Omega} dc - 0.05 \right| \right\}$.

Finally, we assume that accural times are generated using a $\text{Beta}(1, 2)$ distribution so that all individuals are accrued by the end of the first year. The analysis triggers when the last patient in the risk set has been followed up for at least 2 years.

\subsection{Historical data generation}
For historical control data, we set $A_{0i} = 0$. We then draw a binary exchangeability indicator, i.e., for $\gamma \in \{0, 0.5, 1\}$ (respectively corresponding to no, partial, and full exchangeability), we generate $Z_{0i} \sim \text{Bernoulli}(\gamma)$. If the individual is exchangeable ($Z_{0i} = 1)$, their variables are generated the same way as the current data. 

If an individual is not exchangeable (i.e., if $Z_{0i} = 0)$, we generate covariates from the nonparametric bootstrap distribution using the real historical data set. Event times are drawn based on the conditional hazard function $H(\cdot | a = 0, s = 0, \bm{X}_{0i})$ in \eqref{eq:spline_cond_hazard} with $\beta_0 \in \hat{\beta}_0 + \delta$ for $\delta \in \{-0.5, 0, 0.5\}$ resulting in $\beta_0 \in \{ -0.26, -0.76, 0.24 \}$.

For unexchangeable individuals, we assumed censoring times were generated via $C_{0i} | Z_{0i}, \beta_0 \sim \text{Exponential}(\zeta_{\beta_0})$, where $\zeta_{-0.76} = 0.064$, $\zeta_{-0.26} = 0.720$, and $\zeta_{0.24} = 0.867$. Similar to the current data generation, the historical study ends when the last patient in the risk set has been followed up for 2 years.

To mimic a real data borrowing setting with regulatory constraints, we generate a single historical data set for each simulation scenario. Specifically, we repeat the data generation process {100,000} times, and select the historical data set whose KM estimator most closely aligns with the data generation process.

\section{Posterior consistency proof}
\label{sec:intro}

We establish the asymptotic consistency of the NP-LEAP posterior under standard regularity assumptions on the prior. For convenience, here we use a different notation than that of the main manuscript. The proof of the main result can be found in Section \ref{sec:proof-consistency}. The regularity assumptions on the priors from  Section \ref{sec:intro} are verified for the Dirichlet process location-scale mixtures of Gaussians in Section \ref{sec:dp-mixture-of-Gaussians}.

Let $\cX = \R^d$ and use $\DenX = \{h: \cX \to [0,\infty) : \int h(x) dx = 1 \}$ to denote the collection of all densities on $\cX$ with respect to the Lebesgue measure. Given an $h \in \DenX$, let $P_h(\cdot) = \int_{\cdot} h(x) dx$ denote the probability measure on $\cX$ induced by $h$. By identifying densities on $\DenX$ that are (Lebesgue) almost surely equal, we will regard $(\DenX, d_H)$ as a metric space, where $d_H(f,g) = \sqrt{\int(\sqrt{f}(x) - \sqrt{g}(x))^2 dx}$ is the Hellinger metric on $\DenX$. 

The \emph{LEAP model} with parameter $\theta = (f,g,p) \in \Theta \doteq \DenX \times \DenX \times [0,1]$ for the current and historical data $Z_{1:N} = (X_{1:n},Y_{1:m})$ is defined as
\begin{equation*}
\begin{aligned}
X_1, \ldots, X_n &\iid P_f\\
Y_1, \ldots, Y_m &\iid P_{p f + (1-p )g}.
\end{aligned}
\end{equation*}
In words, LEAP posits that the population that generates the historical observations $Y_{1:m}$ contains a proportion $p \in [0,1]$ of the population $f$ that generates the current observations $X_{1:n}$.

The likelihood for the LEAP model is given by:
$$
L(Z_{1:N}|\theta) = \prod_{i=1}^n f(X_i) \prod_{j=1}^m (p f(Y_j) + (1-p) g(Y_j)) \notag.
$$
 
Based on this likelihood and \cite[Section 6.7.1]{ghosal2017fundamentals}, we define a loss $\dLH$ on $\Theta$ for any $\theta_1 = (f_1, g_1, p_1), \theta_2 = (f_2,g_2,p_2) \in \Theta$ and $\Omega \in [0,1]:$
\begin{equation}
\label{def:LH-dist}
\dLH(\theta_1, \theta_2) \doteq \sqrt{\Omega d_H^2(f_1, f_2) + (1-\Omega) d_H^2\left(p_1 f_1 + (1-p_1) g_1, p_2 f_2 + (1-p_2) g_2\right)}.
\end{equation}

The following standard result is proved in Section \ref{sec:proof-aux-lemmas}.
\begin{lemma} For any $\Omega \in [0,1]$, $\dLH: \Theta \times \Theta \to [0,\sqrt{2}]$ is a pseudo-metric on $\Theta$, i.e. $\dLH$ is symmetric, satisfies the triangle inequality, and $\dLH(\theta, \theta) = 0$, but $\dLH(\theta, \theta') = 0$ does not mean that $\theta = \theta'$.
    \label{lem:pseudo-metric}
\end{lemma}

We make the assumption that data $Z_{1:N}$ is generated independently from an unknown $\theta_0 = (f_0, g_0, p_0) \in \Theta$ with $p_0 \in (0,1)$, and $\Omega \in [0,1]$ is the limiting fraction of current observations among all  observations.

\begin{assumption}
\label{a:data}
For some $\theta_0 = (f_0, g_0, p_0) \in \Theta$ with $p_0 \in (0,1)$, observations $Z_N = (X_{1:n}, Y_{1:m})$ are generated as $X_1, \ldots, X_n \iid f_0$ and $Y_1, \ldots, Y_m \iid p_0 f_0 + (1-p_0) g_0$, with $\eta_N \doteq n/N$ converging to $\Omega \in [0,1]$ as $N \to \infty$.
\end{assumption}

\begin{assumption}
    \label{a:leap-prior}
    The LEAP prior $\Pi = \Pi_1 \otimes \Pi_2 \otimes \Pi_3$ on $\Theta$ has a product form, i.e. in a prior draw  $(f,g,p) \sim \Pi$, the parameters $f$, $g$, and $p$ are independent of each other. 
\end{assumption}

While the component priors $\Pi_1, \Pi_2$ can be arbitrary non-parametric distributions on $\DenX$, in order to ensure consistency, we make the standard assumption \cite[Section~6.4]{ghosal2017fundamentals} that these priors assign positive mass to arbitrary neighborhoods of $f_0$ and $g_0$ (Assumption \ref{a:kl-supports} below), and that they give exponentially small mass outside a subset of densities that can be covered by exponentially many balls of any fixed size $\epsilon$ in the $L^1$ (or equivalently total variation) distance (Assumption \ref{a:prior-sieve} below).  We also assume that $p_0 \in (0,1)$ lies in the support of $\Pi_3$.

To state these assumptions precisely, consider first the \emph{Kullback Leibler} (KL) divergence between two densities $f, g \in \DenX$
$$
\KL(f|g) \doteq \begin{cases} 
        \int f(x) \log \frac{f(x)}{g(x)} dx &\text{if } f \ll g \\
        +\infty & \text{otherwise}
\end{cases}
$$
where we use the convention $0\log \frac{0}{0} = 0$ and say $f \ll g$ if $f(x) = 0$ whenever $g(x) = 0$. We say that $h_0 \in \DenX$ lies in the KL support of a distribution $\widetilde{\Pi}$ on $\DenX$ (denoted by $h_0 \in \KL(\widetilde{\Pi})$) if 
$$
\widetilde{\Pi}(\{h \in \DenX : \KL(h_0|h) < \epsilon\}) > 0 \qquad \forall \epsilon > 0.
$$

\begin{assumption}
    \label{a:kl-supports}
    $f_0 \in \KL(\Pi_1)$  and $g_0 \in \KL(\Pi_2)$ are in the Kullback-Leibler supports of $\Pi_1$ and $\Pi_2$. Further assume that  $\Pi_3([p_0-r,p_0 + r]) > 0$ for each $r > 0$.
\end{assumption}

Next, consider the $L^1$ norm $\|f-g\|_1 \doteq \int |f(x) - g(x)| dx$ on $\DenX$. The \emph{covering number} $\Nmc(\epsilon, A, \|\cdot\|_1)$ of a subset $A \subseteq \DenX$ is defined as the size $K$ of the smallest set  $\{f_1, \ldots, f_K\} \subseteq A$ such that for each $f \in A$ there is a $k \in \{1, \ldots,K\}$ such that $\|f-f_k\|_1 \leq \epsilon$. The next assumption states that for any $\epsilon > 0$ and sufficiently large $N \in \nat$, the priors $\Pi_1, \Pi_2$ place exponentially small masses outside regions $A_N, B_N \subseteq \DenX$, which can be covered with exponentially many $L^1$ balls of radius $\epsilon > 0$.

\begin{assumption}
    \label{a:prior-sieve}
    For each $\epsilon > 0$, there are constants $N_0, C > 0$ such that for each $N \in \nat$ with $N \geq N_0$, there are sets $A_N,B_N \subseteq \DenX$ such that
    \begin{enumerate}
        \item  $\log \Nmc(\epsilon, A_N, \|\cdot\|_1) \leq 3N\epsilon^2$ and $\log \Nmc(\epsilon, B_N, \|\cdot\|_1) \leq 3N\epsilon^2$, and 
        \item $\Pi_1(\DenX \setminus A_N) \leq e^{-NC}$ and $\Pi_2(\DenX \setminus B_N) \leq e^{-NC}$.
        \end{enumerate}
\end{assumption}

\begin{remark} Many standard non-parametric choices of prior distributions $\Pi_1$ and $\Pi_2$ will satisfy Assumptions \ref{a:kl-supports} and \ref{a:prior-sieve} for suitably regular $f_0, g_0 \in \DenX$ \cite[Chapter 7]{ghosal2017fundamentals}. For example, in Section \ref{sec:dp-mixture-of-Gaussians}, we show that a choice of Dirichlet process location-scale mixture of Gaussians for $\Pi_1$ and $\Pi_2$ will satisfy Assumptions \ref{a:kl-supports} and \ref{a:prior-sieve}, when $f_0$ and $g_0$ are sufficiently regular densities on $\R^d$, and the base measures satisfy the mild conditions in Assumption \ref{a:dp-base-measure}.
\end{remark}

Let $\Pi(\cdot|Z_{1:N})$ denote the LEAP posterior on $\Theta$, i.e.~given a measurable $T \subseteq \Theta$:
$$
\Pi(T|Z_{1:N}) \doteq \frac{\int_{T} L(\theta|Z_{1:N}) \Pi(d\theta)}{\int L(\theta|Z_{1:N}) \Pi(d\theta)}
$$

\begin{theorem}
\label{thm:consistency}
Suppose Assumptions \ref{a:data}, \ref{a:leap-prior}, \ref{a:kl-supports}, and \ref{a:prior-sieve} are satisfied. Then for every $\epsilon > 0$, almost surely, the LEAP posterior $\Pi(\cdot|Z_{1:N})$ satisfies
$$
\lim_{N \to \infty} \Pi(\{ \theta \in \Theta : \dLH(\theta, \theta_0) > \epsilon\} |Z_{1:N}) = 0.
$$
\end{theorem}

Note that this does not establish posterior consistency at $\theta_0$ since $\dLH$ is only a pseudo-metric on $\Theta$. Indeed, $\dLH(\theta, \theta_0) = 0$ for some $\Omega \in (0,1)$ only means that  $f = f_0$ and $(g,p)$ are such that:
\begin{equation}
p f_0 + (1-p) g = p_0 f_0 + (1-p_0) g_0. \label{eq:non-identifiability}
\end{equation}
For instance, $\dLH((f_0, h_0,0), (f_0, g_0, p_0)) = 0$ where $h_0 = p_0 f_0 + (1-p_0) g_0$. 

\begin{remark}
   Establishing consistency at $\theta_0$ is tricky because LEAP model suffers from some inherent non-identifiability. Given $\theta_0 = (f_0, h_0, p_0)$, consider the set $T_{{\theta_0}} \subseteq \Theta$ defined as
$$
T_{\theta_0} \doteq \{ \theta \in \Theta: \dLH(\theta, \theta_0) = 0\} = \{(f_0, g, p) : (g,p) \text{ satisfies } \eqref{eq:non-identifiability}\}.
$$
Even as $N \to \infty$, the posterior cannot discriminate (beyond the preference indicated by the prior) between two elements $\theta_1, \theta_2$ within $T_{\theta_0}$ because $L(Z_{1:N}|\theta) = L(Z_{1:N}|\theta_0)$ for each $\theta \in T_{\theta_0}$. Thus $\Pi(\cdot|Z_{1:N})$ when restricted to $T_{\theta_0}$ is simply proportional to the prior $\Pi(\cdot)$ restricted to this set. However, this does not prove that $\Pi(\cdot|Z_{1:N})$ is \emph{inconsistent} at $\theta_0$ since this exceptional set $T_{\theta_0}$ will have zero prior mass.
\end{remark}


We can however show that LEAP successfully estimates the density $f_0$ of current observations as soon as the limiting proportion of current to historical observations is non-zero, i.e. $\Omega > 0$. Indeed, noting    $d_H^2(f, f_0) \leq \frac{1}{\Omega} \dLH^2(\theta, \theta_0)$ for any $\theta = (f,g,p) \in \Theta$,  it is easy to establish the following corollary of Theorem \ref{thm:consistency}.
\begin{corollary}
Suppose that all conditions of Theorem \ref{thm:consistency} are satisfied and $\Omega > 0$. Then the LEAP posterior is Hellinger consistent at $f_0$. In other words, for every $\epsilon > 0$, almost surely
$$
\lim_{N \to \infty} \Pi(\{ (f,g,p) \in \Theta : d_H(f, f_0) > \epsilon\} |Z_{1:N}) = 0.
$$
\end{corollary}

\subsection{Proof of Theorem \ref{thm:consistency}}
\label{sec:proof-consistency}

Our proof uses ideas from Bayesian posterior consistency for i.n.i.d. observations \cite[Theorem 6.41]{ghosal2017fundamentals}. 

The LEAP model posits that observations $Z_{1:N}=(X_{1:n}, Y_{1:m})$ are independent with $X_1, \ldots, X_n \iid P_f$ and $Y_1, \ldots, Y_m \iid P_{p f + (1-p) g}$. Following the notation in \cite[Section 6.7.1]{ghosal2017fundamentals}, this leads to the definition of the pseudo-metric $\dLH$ that we introduced earlier in \eqref{def:LH-dist}, which is seen to satisfy
$$
\begin{aligned}
\dLH[\eta_N]^2(\theta_1, \theta_2) &= \frac{1}{N} \left( \sum_{i=1}^n d_H^2(f_1, f_2) + \sum_{j=1}^m d_H^2(p_1 f_1 + (1-p_1) g_1, p_2 f_2+(1-p_2) g_2) \right)\\
&= d_{N,H}^2(\bP_{\theta_1}^N, \bP_{\theta_2}^N)
\end{aligned}
$$
where $\eta_N = n/N$ is the fraction of total observations $N=n+m$ that are current, $\bP_{\theta}^N = P_f^{\otimes n} \otimes P_{p f + (1-p) g}^{\otimes m}$ denotes the joint law of data $Z_{1:N} = (X_{1:n}, Y_{1:m})$ under our LEAP model with parameter $\theta = (f,g,p)$, and $d_{N,H}(\otimes_{i=1}^N P_i, \otimes_{i=1}^N Q_i) \doteq \sqrt{\frac{1}{N} \sum_{i=1}^N d_H^2(P_i,Q_i)}$ is the standard \emph{root average square Hellinger metric} used for posterior consistency \cite[Section 6.7.1]{ghosal2017fundamentals}.

As stated in the following lemma, Theorem 
\ref{thm:consistency} can be established by checking two standard conditions. 


\begin{lemma} Suppose Assumption \ref{a:data} holds and the following two conditions are satisfied by a prior $\Pi$ on $\Theta$:
\label{lem:consistency-conditions} 
\begin{enumerate}
    \item (KL support condition) For every $\epsilon > 0$ there is a set $B \subseteq \Theta$ with $\Pi(B) > 0$ such that
\begin{equation}
\label{eq:kl-support-equation}
\Omega \KL(f_0| f) + (1-\Omega) \KL(p_0 f_0 + (1-p_0) g_0| p f + (1-p) g) \leq \epsilon
\end{equation}
for every $\theta = (f,g,p) \in B$.
    \item (Sieve condition) There is an $\epsilon_0 > 0$ such that given any $\epsilon \in (0, \epsilon_0]$, there are constants $C, N_0 > 0$ such that for each $N \in \nat$ with $N \geq N_0$, there is a subset $\Theta_N \subseteq \Theta$, such that
\begin{enumerate}
    \item $\log \Nmc(\epsilon, \Theta_N, \dLH[\eta_N]) \leq 12N \epsilon^2$, \text{ and }
    \item $\Pi(\Theta \setminus \Theta_N) \leq e^{-N C}$.
\end{enumerate}

\end{enumerate}
Then, for every $\delta > 0$, the posterior $\Pi(\cdot|Z_{1:N})$ satisfies almost surely
\begin{equation}
\label{eq:consistency-condition}
\lim_{N \to \infty} \Pi(\{ \theta: \dLH[\eta_N](\theta, \theta_0) > \delta\} |Z_{1:N}) = 0.
\end{equation}

\end{lemma}
\begin{proof}
    First we consider a consequence of the KL support condition. Fix an $\epsilon > 0$ and $B \subseteq \Theta$ be such that $\Pi(B) > 0$ and \eqref{eq:kl-support-equation} holds. An appeal to Jenesen's inequality with the conditional probability measure $\Pi_B(\cdot) \doteq \frac{\Pi(\cdot \cap B)}{\Pi(B)}$ shows
\begin{align*}
    \frac{1}{N} \log \int \frac{L(\theta|Z_{1:N})}{L(\theta_0|Z_{1:N})} \Pi(d\theta)   \notag
    &\geq \frac{\log \Pi(B)}{N} + \frac{1}{N}\log \int \frac{L(\theta|Z_{1:N})}{L(\theta_0|Z_{1:N})} \Pi_B(d\theta)\notag \\ 
    &\geq \frac{\log \Pi(B)}{N} + \frac{1}{N}\int \log \frac{L(\theta|Z_{1:N})}{L(\theta_0|Z_{1:N})} \Pi_B(d\theta)\\
    &= \frac{\log \Pi(B)}{N} - \frac{\eta_N}{n} \sum_{i=1}^n \int  \log \frac{f_0(X_i)}{f(X_i)} \Pi_B(d\theta) \\
    &- \frac{1-\eta_N}{m} \sum_{j=1}^m \int \log \frac{p_0 f_0(Y_i) + (1-p_0) g_0(Y_i)}{p f(Y_i) + (1-p) g(Y_i)} \Pi_B(d\theta).
\end{align*}
Here we have used $\int \frac{L(\theta|Z_{1:N})}{L(\theta_0|Z_{1:N})} \Pi(d\theta) \geq \int_B \frac{L(\theta|Z_{1:N})}{L(\theta_0|Z_{1:N})} \Pi(d\theta)$ in the first line, Jensen's inequality in the second line, and $L(\theta|Z_{1:N}) \doteq \prod_{i=1}^n f(X_i) \prod_{j=1}^m (p f(Y_j) + (1-p) g(Y_j))$ and $\eta_N \doteq n/N$ for any $\theta = (f,g,p) \in \Theta$ in the last line.

Under Assumption \ref{a:data}, the law of large numbers and Fubini's theorem show
$$
\begin{aligned}
 &\liminf_{N \to \infty} \frac{1}{N} \log \int \frac{L(\theta|Z_{1:N})}{L(\theta_0|Z_{1:N})} \Pi(d\theta)\\
 &\geq -\int  \left\{\Omega \KL(f_0|f) + (1-\Omega) \KL(p_0 f_0 + (1-p_0) g_0| p f + (1-p) g))   \right\} \Pi_B(df,dg,dp)\\
 &\geq -\sup_{(f,g,p) \in B} \left\{\Omega \KL(f_0|f) + (1-\Omega) \KL(p_0 f_0 + (1-p_0) g_0| p f + (1-p) g)\right\} \geq -\epsilon
\end{aligned}
$$
almost surely. Since $\epsilon > 0$ is arbitrary, this shows
\begin{equation}
    \label{eq:denom-liminf}
\liminf_{N \to \infty} \frac{1}{N} \log \int \frac{L(\theta|Z_{1:N})}{L(\theta_0|Z_{1:N})} \Pi(d\theta) \geq 0 \qquad \text{ almost surely.}
\end{equation}

Next, applying \cite[Proposition D.9]{ghosal2017fundamentals} along with the connection between $\dLH[\eta_N]$ and the root average square Hellinger metric discussed at the beginning of Section \ref{sec:proof-consistency}, for every $c, \epsilon > 0$, $N \in \nat$, and $\theta_1 \in \Theta$ such that $\dLH[\eta_N](\theta_0, \theta_1) > \epsilon$, there is a test $\tilde{\phi}_N = \tilde{\phi}_N(Z_{1:N}) \in [0,1]$ such that
$$
\bP_{\theta_0}^N(\tilde{\phi}_N) \leq ce^{-N \epsilon^2/8} \quad \text{ and } \quad \sup_{\substack{\theta' \in \Theta \\ \dLH[\eta_N](\theta',\theta_1) < \epsilon/5}} \bP_{\theta'}^N(1-\tilde{\phi}_N) \leq c^{-1}e^{- N \epsilon^2/8}.
$$ 

It suffices to show \eqref{eq:consistency-condition} given an arbitrary $\delta \in (0,10\epsilon_0]$. We will use the Sieve condition  with $\epsilon = \frac{\delta}{10}$ to obtain a $N
_0 \in \nat$, $C > 0$, and a decomposition $\Theta = \Theta_{N,1} \cup \Theta_{N,2}$ for each $N \geq N_0$. Writing 
$T_N = \{\theta: \dLH[\eta_N](\theta, \theta_0) > \delta\}$ as $T_{N,1} \cup T_{N,2}$ where $T_{N,1} = T_N \cap \Theta_{N}$ and $T_{N,2} = T_N \cap (\Theta \setminus \Theta_{N})$, it will be sufficient to show that $\Pi(T_{N,1}|Z_{1:N})$ and $\Pi(T_{N,2}|Z_{1:N})$ converge to zero almost surely as $N \to \infty$. By standard properties of covering numbers and the Sieve condition, for any subset $S \subseteq T_{N,1}  \subseteq \Theta_{N}$ it follows that $\log \Nmc(\delta/5, S, \dLH[\eta_N]) \leq \log \Nmc(\delta/10, \Theta_{N}, \dLH[\eta_N]) \leq \frac{3}{25} N \delta^2$ \cite[Appendix C]{ghosal2017fundamentals}.

First let us show that $\Pi(T_{N,1}|Z_{1:N})$ converges to zero as $N \to \infty$. The previously displayed inequalities allow us to apply \cite[Theorems D.5]{ghosal2017fundamentals} (we apply the theorem with  $P=\bP_{\theta_0}^N$, $\cQ = \{ \bP_{\theta}^N : \theta \in T_{N,1}\}$, $\xi = 1/5, c = 1/2, K=N/8$, $d=e$, $d(\bP_{\theta}^N, \bP_{\theta'}^N) = \dLH[\eta_N](\theta, \theta')$, $\epsilon_0 = 0$, $\epsilon = \delta$, $N(\epsilon) = \Nmc(\epsilon/5, \cQ, d)$, $j=1$ noting that $\log N(\delta) \leq \frac{3}{25}N \delta^2$) to conclude that for every $N \geq N_0$, there is a test $\phi_N = \phi_N(Z_{1:N}) \in [0,1]$ such that 
$$
\bP_{\theta_0}^N(\phi_N) \leq e^{-(1/8-3/25) N\delta^2} 
\quad \text{ and } \quad
\sup_{\theta \in T_{N,1}} \bP_{\theta}^N(1-\phi_N) \leq 2 e^{- \frac{N \epsilon^2}{8}}
$$
if $N_0$ is chosen sufficiently large to satisfy $e^{-(1/8-3/25) N_0 \delta^2} \leq 1/2$.
Now by Fubini's theorem we can note
$$
\bP_{\theta_0}^N\left[(1-\phi_N)\int_{T_{N,1}} \frac{L(\theta|Z_{1:N})}{L(\theta_0|Z_{1:N})}\Pi(d\theta)\right] = \int_{T_{N,1}} \bP_{\theta}^N(1-\phi_N) \Pi(d\theta) \leq 2 e^{- \frac{N \delta^2}{8}}
$$
since $\bP_{\theta_0}^N[(1-\phi_N(Z_{1:n}))\frac{L(\theta|Z_{1:N})}{L(\theta_0|Z_{1:N})}] = \bP_{\theta}^N[1-\phi_N(Z_{1:n})]$. 

Since the right hand side in the previous display continues to be summable in $N$ even when multiplied by $e^{\frac{N \delta^2}{16}}$, the Borel Cantelli theorem shows that $e^{\frac{N \delta^2}{16}} (1-\phi_N)\int_{T_{N,1}} \frac{L(\theta|Z_{1:N})}{L(\theta_0|Z_{1:N})}\Pi(d\theta)$ converges to zero almost surely as $N \to \infty$. Similarly, since $\bP_{\theta_0}^N(\phi_N)$ is summable in $N$, this also shows the almost sure convergence of $\phi_N$ to zero as $N \to \infty$. 

Finally, using \eqref{eq:denom-liminf} we note that $\int \frac{L(\theta|Z_{1:N})}{L(\theta_0|Z_{1:N})} \Pi(d\theta) \geq e^{-N \frac{\delta^2}{16}}$ for every $N$ bigger than a (random) $N_0 \in \nat$. This allows us to conclude that for $N \geq N_0$
$$
\begin{aligned}
\Pi(T_{N,1}|Z_{1:N}) &= \frac{\int_{T_{N,1}} \frac{L(\theta|Z_{1:N})}{L(\theta_0|Z_{1:N})} \Pi(d\theta)}{\int \frac{L(\theta|Z_{1:N})}{L(\theta_0|Z_{1:N})} \Pi(d\theta)}\\
 &\leq \phi_N + \frac{(1-\phi_N)\int_{T_{N,1}} \frac{L(\theta|Z_{1:N})}{L(\theta_0|Z_{1:N})} \Pi(d\theta)}{\int \frac{L(\theta|Z_{1:N})}{L(\theta_0|Z_{1:N})} \Pi(d\theta)}\\
 &\leq \phi_N + e^{\frac{N \delta^2}{16}} (1-\phi_N)\int_{T_{N,1}} \frac{L(\theta|Z_{1:N})}{L(\theta_0|Z_{1:N})} \Pi(d\theta),
\end{aligned}
$$
showing that $\limsup_{N \to \infty} \Pi(T_{N,1}|Z_{1:N}) = 0$ almost surely. 

Establishing that $\Pi(T_{N,2}|Z_{1:N})$ follows a very similar proof where $\phi_N = 0$. Indeed, we need only note that
$$
\bP_{\theta_0}^N\left[\int_{T_{N,2}} \frac{L(\theta|Z_{1:N})}{L(\theta_0|Z_{1:N})}\Pi(d\theta)\right] = \int_{T_{N,2}} \bP_{\theta_0}^N\left[\frac{L(\theta|Z_{1:N})}{L(\theta_0|Z_{1:N})}\right] \Pi(d\theta) = \Pi(T_{N,2}) \leq e^{-CN},
$$
and appealing to the Borel-Cantelli lemma as before shows that $e^{CN/2} \int_{T_{N,2}} \frac{L(\theta|Z_{1:N})}{L(\theta_0|Z_{1:N})}\Pi(d\theta)$ converges to zero almost surely. Using \eqref{eq:denom-liminf} as before, we can conclude that $\limsup_{N \to \infty} \Pi(T_{N,2}|Z_{1:N}) = 0$ almost surely. The result now follows by combining the bounds for $\Pi(T_{N,1}|Z_{1:N})$ and $\Pi(T_{N,2}|Z_{1:N})$.
\end{proof}

\begin{proof}[Proof of Theorem \ref{thm:consistency}]
    As we show in Lemma \ref{lem:verify-KL-support} and Lemma \ref{lem:verify-metric-entropy-condition} below, conditions of Lemma \ref{lem:consistency-conditions} will be satisfied under our assumptions. This shows that \eqref{eq:consistency-condition} holds with $\delta = \epsilon$ for every $\epsilon > 0$. 

    Since $\eta_N \to \Omega \in [0,1]$, for any $\delta > 0$ there is a $N_0 \in \nat$ such that $|\eta_N - \Omega| \leq \delta$ for each $N \geq N_0$. For any $\theta, \theta' \in \Theta$ and $N \geq N_0$, since $d^2_H(\cdot,\cdot) \leq 2$ (see proof of Lemma \ref{lem:pseudo-metric}) we have
    $$
    |\dLH^2(\theta, \theta') - \dLH[\eta_H]^2(\theta, \theta')| \leq |\Omega - \eta_N||d^2_H(f,f') - d^2_H(h,h')| \leq 2\delta
    $$
    where we write $h = p f  + (1-p) g$ and $h' = p'f' + (1-p') g'$ if $\theta = (f,g,p)$ and $\theta' = (f',g',p')$.
    
    Given any $\epsilon > 0$, the choice $\delta= \epsilon^2/4$ then shows that for each $N \geq N_0$:
    $$
    \Pi(\{\theta: \dLH^2(\theta, \theta_0) > \epsilon^2 \} | Z_{1:N}) \leq \Pi(\{\theta: \dLH[\eta_N]^2(\theta, \theta_0) > \epsilon^2/2 \} | Z_{1:N}).
    $$
    Since \eqref{eq:consistency-condition} holds, this shows that $\Pi(\{\theta: \dLH(\theta, \theta_0) > \epsilon\} | Z_{1:N})$ converges to zero almost surely as $N \to \infty$ for each $\epsilon > 0$.
\end{proof}

\subsubsection{Proof of the KL support condition}

\begin{lemma} If Assumptions \ref{a:leap-prior} and \ref{a:kl-supports} hold then the LEAP prior satisfies the KL support condition of Lemma \ref{lem:consistency-conditions}. In other words, for every $\epsilon > 0$, there is a set  $B \subseteq \Theta$  with $\Pi(B) > 0$ such that 
$$
\Omega \KL(f_0| f) + (1-\Omega) \KL(p_0 f_0 + (1-p_0) g_0| p f + (1-p) g) \leq \epsilon
$$
for every $\theta = (f,g,p) \in B$. 
\label{lem:verify-KL-support}
\end{lemma}
\begin{proof} By Assumption \ref{a:kl-supports}, there are measurable sets $B_1, B_2 \subseteq \DenX$ such that $\Pi_1(B_1) > 0$, $\Pi_2(B_2) > 0$ and $\KL(f_0|f) \leq \frac{\epsilon}{2}$ for every $f \in B_1$ and $\KL(g_0|g) \leq \frac{\epsilon}{2}$ for every $g \in B_2$. Further, for some suitably small $\delta > 0$ to be chosen later, we know that $\Pi_3$ assigns positive mass to $B_3 \doteq [p_0 - \delta, p_0 + \delta] \subseteq (0,1)$. Thus letting $B = B_1 \times B_2 \times B_3$, it follows from Assumption \ref{a:leap-prior} that $\Pi(B) = \Pi_1(B_1) \Pi_2(B_2)\Pi_3(B_3) > 0$. 
    
For any $\theta = (f,g,p) \in B$, we thus have $\KL(f_0|f) \leq \frac{\epsilon}{2}$, $\KL(g_0|g) \leq \frac{\epsilon}{2}$, and $|p - p_0| \leq \delta$. This also means that $\{f = 0\} \subseteq \{f_0=0\}$ and $\{g = 0\} \subseteq \{g_0 = 0\}$ by definition of the $\KL$ divergence. In particular defining
$$
h \doteq p_0 f_0 + (1-p_0) g_0, \quad k \doteq p f + (1-p) g, \quad \text{and } q \doteq p_0 f + (1-p_0) g
$$
and using $p,p_0 \in (0,1)$ shows $\{q = 0\} = \{k=0\} \subseteq \{h=0\}$ and
$$
\begin{aligned}
\KL(h|k) & = \int_{\{h > 0\}} h(x) \log \frac{h(x)}{k(x)} dx\\
&= \int_{\{h > 0\}} h(x) \log \frac{h(x)}{q(x)} dx + \int_{\{h > 0\}} h(x) \log \frac{q(x)}{k(x)} dx\\
&= \KL(h|q) + \int_{\{h > 0\}} h(x) \log \frac{q(x)}{k(x)} dx.
\end{aligned}
$$
Next the joint convexity of KL divergence \cite[Lemma 2.4b]{budhiraja2019analysis} shows
$$
\KL(h|q) \leq p_0 \KL(f_0|f) + (1-p_0) \KL(g_0|g) \leq \frac{\epsilon}{2}.
$$
Using the fact that $\log(1+x) \leq x$ for all $x > -1$, we can bound the second term as follows:
$$
\begin{aligned}
\int_{\{h > 0\}} h(x) \log \frac{q(x)}{k(x)} dx 
&= \int_{\{h > 0\}} h(x) \log \left(1 + \frac{p_0 - p}{k(x)} (f(x) - g(x)) \right) dx\\
&\leq (p_0-p)\int_{\{h > 0\}} \frac{h(x)}{k(x)} (f(x) - g(x)) dx\\
&= \frac{(p_0-p)}{p}\int_{\{h > 0\}}  \frac{h(x)}{k(x)} (k(x) - g(x)) dx\\
&= \frac{(p_0-p)}{p}\int_{\{h > 0\}}  h(x) \left(1 - \frac{g(x)}{k(x)}\right) dx\\
&\leq \frac{|p_0-p|}{p}\int_{\{h > 0\}}  h(x) \left|1 - \frac{g(x)}{k(x)}\right| dx\\
&\leq \frac{|p_0-p|}{p}\int_{\{h > 0\}}  h(x) \left(1 + \frac{g(x)}{k(x)}\right) dx \leq\frac{\delta}{p} (1 + \frac{1}{1-p}).
 \end{aligned}
$$
Since $p_0 \in (0,1)$, a choice $0 < \delta \leq \min(p_0, 1-p_0)/2$ ensures that $p \in  [\frac{p_0}{2}, 1 - \frac{p_0}{2}]$ for every $p \in B_3$, which means that the last expression is bounded above by $\frac{2\delta}{p_0} (1 + \frac{2}{p_0})$. This means that for a suitably small choice of $\delta > 0$, we will have $\int_{\{h > 0\}} h(x) \log \frac{q(x)}{k(x)} dx \leq \frac{\epsilon}{2}$. 

Combining the previous bounds we now see that 
$$
\KL(p_0 f_0 + (1-p_0) g_0|p f + (1-p) g) = \KL(h|k) \leq \epsilon
$$
for every $(f,g,p) \in B$. The result now follows since we already know that $\KL(f_0|f) \leq \epsilon$ since $f \in B_1$.
\end{proof}

\subsubsection{Proof of the Sieve condition}

Under Assumptions \ref{a:leap-prior} and \ref{a:prior-sieve} we will now show that our LEAP prior satisfies the Sieve condition in Lemma \ref{lem:consistency-conditions}. 
Namely, for any $\epsilon > 0$ we will show the existence of constants $C, N_0 > 0$, such that for each $N \in \nat$ with $N \geq N_0$ there is a subset $\Theta_{N} \subseteq \Theta$ such that: 
\begin{enumerate}
    \item $\sup_{\Omega \in [0,1]} \log \Nmc(\epsilon, \Theta_{N}, \dLH) \leq 12N \epsilon^2$, \text{ and }
    \item $\Pi(\Theta \setminus \Theta_{N} ) \leq e^{- CN}$.
\end{enumerate}

Recall from \eqref{def:LH-dist} that $\dLH$ is defined in terms of the Hellinger distance $d_H$ on densities $\DenX$.  However, owing to the close relation between $d_H$ and the $L^1$ distance $\|f-g\|_1 \doteq \int |f(x)-g(x)| dx$, it is often convenient to work with the latter. For example, a standard connection \cite[Eq.~8]{gibbs2002choosing} between the total-variation distance $d_{\textrm{TV}}(f,g) \doteq \frac{1}{2}\|f-g\|_1$ and the Hellinger distance $d_H$ shows
$$
\frac{1}{2} \|f-g\|_1 \leq d_{H}(f,g) \leq  \sqrt{\|f-g\|_1} \qquad \forall f,g \in \DenX.
$$
This allows us to reduce the task of finding $\Theta_{N}$ satisfying our Sieve condition above to standard results in the literature \cite[Chapter 7]{ghosal2017fundamentals} where such a condition is established for standard priors $\widetilde{\Pi}$ on the space of densities $\DenX$ with respect to the $L^1$ distance. 

\begin{lemma} Given subsets $A, B \subseteq \DenX$, for each $\epsilon > 0$, the set $\tilde{\Theta} = A \times B \times [0,1] \subseteq \Theta$ satisfies
$$
\sup_{\Omega \in [0,1]} \log \Nmc(\epsilon, \tilde{\Theta}, \dLH) \leq \log \Nmc(\epsilon^2/2, A, \|\cdot\|_1) + \log \Nmc(\epsilon^2/2, B, \|\cdot\|_1) + \log \lceil4/\epsilon^2\rceil.
$$
\label{lem:prod-space-entropy}
\end{lemma}
\begin{proof}
With $\delta \doteq \frac{\epsilon^2}{2}$, let $A_1 = \{f_1, \ldots, f_{K_1}\}$ and $B_1 = \{g_1, \ldots, g_{K_1}\}$ be $\delta$-nets\footnote{We call $T \subseteq S$ a $\delta$-net of a metric space $(S,d)$ if for each $s \in S$ there is a $t \in T$ such that $d(s,t) \leq \delta$.} for $A$ and $B$ under $L^1$ norm with minimal size, i.e.~$K_1 = \Nmc(\delta, A, \|\cdot\|_1)$ and $K_2 = \Nmc(\delta, B, \|\cdot\|_1)$. With $M \doteq \lceil\frac{4}{\epsilon^2}\rceil$, let $C_1 \doteq \{0, \frac{1}{M}, \ldots \frac{M-1}{M}\} \subseteq [0,1]$. Then  defining $\tilde{\Theta}_1 \doteq A_1 \times B_1 \times C_1$, we note: 
$$
\log |\tilde{\Theta}_1| = \log |A_1| + \log |B_1|  + \log |C_1| = \log K_1 + \log K_2  + \log M. 
$$
Thus it will be sufficient to show that $\tilde{\Theta}_1$ forms an $\epsilon$-net for $\tilde{\Theta}$ in the metric $\dLH$ for any $\Omega \in [0,1]$. This will follow if for any $\theta = (f,g,p) \in \tilde{\Theta}$ and $\Omega \in [0,1]$, we can produce $\theta' = (f',g',p') \in \tilde{\Theta}_1$ such that 
\begin{align}
\dLH^2(\theta, \theta') &\doteq \Omega d_{H}^2(f,f') + (1-\Omega) d_{H}^2(p f + (1-p)g, p'f'  + (1-p')g') \nonumber\\
&\leq \max(\|f-f'\|_1, \|p f + (1-p)g - p'f' - (1-p')g'\|_1) \leq \epsilon^2. \label{eq:dLH-L1-bound}
\end{align}

Fix any $\theta = (f,g,p) \in \tilde{\Theta}$. We now show that \eqref{eq:dLH-L1-bound} holds for some $(f',g',p') \in \tilde{\Theta}_1$. Indeed, since $\tilde{\Theta}_1 = A_1 \times B_1 \times C_1$, there is a $(f', g', p') \in \tilde{\Theta}_1$ such that $\|f-f'\|_1, \|g-g'\|_1 \leq \delta$ and $|p'-p| \leq \frac{1}{M}$. Since $\|\cdot\|_1$ is a norm satisfying $\|a h + bg\|_1\leq |a|\|h\|_1 + |b|\|g\|_1$ for any constants $a, b \in \R$ and functions $h,g: \cX \to \R$, we have:
$$
\begin{aligned}
&\|p f + (1-p)g - p'f' - (1-p')g'\|_1\\
&\leq \|p f + (1-p)g - p'f - (1-p')g\|_1 + \|p' f + (1-p')g - p'f' - (1-p')g'\|_1 \\
&\leq |p-p'|\|f-g\|_1 + p'\|f-f'\|_1 + (1-p')\|g-g'\|_1 \leq \frac{2}{M} +  \delta \leq \epsilon^2
\end{aligned}
$$
where the last line uses $\|f-g\|_1 \leq \|f\|_1 + \|g\|_1 = 2$ for densities $f, g \in \DenX$ and the definitions of $M$ and $\delta$. With this \eqref{eq:dLH-L1-bound} is seen to hold since we already know that $\|f-f'\|_1 \leq \frac{\epsilon^2}{2}$.
    
\end{proof}

The above result allows us to show that our Sieve condition will be satisfied.

\begin{lemma}
    Under Assumptions \ref{a:leap-prior} and \ref{a:prior-sieve}, the LEAP prior satisfies the Sieve condition in Lemma \ref{lem:consistency-conditions}.
    \label{lem:verify-metric-entropy-condition}
\end{lemma}
\begin{proof}
    In order to satisfy the Sieve condition in Lemma \ref{lem:consistency-conditions}, given $\delta \in (0,1]$ take $\epsilon = \delta^2/2$ in Assumption \ref{a:prior-sieve} to find $C', N'_0 > 0$ such that for every $N \in \nat$ with $N \geq N'_0$, there are sets $A_N, B_N \subseteq \DenX$ such that
    $$
    \max\{\log \Nmc(\epsilon, A_N, \|\cdot\|_1),\log \Nmc(\epsilon, B_N, \|\cdot\|_1)\} \leq 3 N\epsilon^2 \leq 3 N\epsilon = \frac{3}{2} N\delta^2
    $$
    and
    $$
    \Pi_1(\DenX \setminus A_N) \leq e^{-NC'} \quad \text{  and } \quad  \Pi_2(\DenX \setminus B_N) \leq e^{-NC'}.
    $$
    By Lemma \ref{lem:prod-space-entropy} the set $\Theta_N = A_N \times B_N \times [0,1]$ satisfies
    $$
    \sup_{\Omega \in [0,1]} \log \Nmc(\delta, \Theta_N, \dLH) \leq 3 N\delta^2 + \log \lceil 4/\delta^2 \rceil \leq 4 N \delta^2
    $$
    whenever $N \geq N_0 \doteq \max(N_0', \frac{\log \lceil 4/\delta^2.
    \rceil}{\delta^2})$.  Note further that 
    $$
    \Pi(\Theta \setminus \Theta_N) \leq \Pi_1(\DenX \setminus A_N) + \Pi_2(\DenX \setminus B_N) \leq 2e^{-NC'} \leq e^{-N C'/2}
    $$
    as long as $N_0 C' \geq 2$. The Sieve condition in Lemma \ref{lem:consistency-conditions} is thus seen to be satisfied with our notation $\epsilon_0 = 1$ and $\delta = \epsilon$.
\end{proof}

\subsubsection{Proof of auxiliary lemmas}
\label{sec:proof-aux-lemmas}

\begin{proof}[Proof of Lemma \ref{lem:pseudo-metric}]
    That $\dLH$ is symmetric, takes a values in $[0,\sqrt{2}]$, and satisfies $\dLH(\theta, \theta) = 0$ follows from the corresponding properties \cite{gibbs2002choosing} of the Hellinger metric $d_H$ on $\DenX$. 

    To see that $\dLH$ satisfies the triangle inequality, we start from the connection between the Hellinger metric  $d_H(f,g) = \|\sqrt{f}-\sqrt{g}\|_2$ and the norm
    $\|h\|_2 = \sqrt{\int h^2(x) dx}$ on the space $L^2(\cX)$ of square integrable functions on $\cX$, which is a Hilbert space. In particular the product $\Hmc = L^2(\cX) \times L^2(\cX)$ is also an Hilbert space with the component-wise inner-product:
    $$
    \langle (h,k), (l,m)\rangle \doteq \int h(x) l(x) dx +  \int k(x) m(x) dx
    $$
    the induced norm
    $$
    \|(h,k)\| \doteq \sqrt{\langle (h,k), (h,k) \rangle} = \sqrt{\|h\|_2^2 + \|k\|_2^2}.
    $$
    Importantly, we can use this discussion to write:
    $$
    \begin{aligned}
    \dLH(\theta_1, \theta_2) &\doteq \sqrt{\Omega d_H^2(f_1, f_2) + (1-\Omega) d_H^2\left(p_1 f_1 + (1-p_1) g_1, p_2 f_2 + (1-p_2) g_2\right)}\\
    &=\sqrt{\Omega \|\sqrt{f_1} - \sqrt{f_2}\|_2^2 + (1-\Omega) \left\|\sqrt{p_1 f_1 + (1-p_1) g_1} - \sqrt{p_2 f_2 + (1-p_2) g_2} \right\|^2_2}\\
    &= \|H_{\Omega} \theta_1 -  H_{\Omega} \theta_2\|
    \end{aligned}
    $$
    where $H_{\Omega} : \Theta \to \Hmc$ is given by
    $$
    H_{\Omega} (f,g,p)  \doteq \left(\sqrt{\Omega f}, \sqrt{(1-\Omega) \{p f + (1-p)g\}}\right).
    $$
    Given $\theta_1, \theta_2, \theta_3 \in \Theta$ we thus have the triangle inequality:
    $$
    \begin{aligned}
    \dLH(\theta_1, \theta_3) &= \|H_{\Omega} \theta_1 -  H_{\Omega} \theta_3\| = \|H_{\Omega} \theta_1 -  H_{\Omega} \theta_2 + H_{\Omega} \theta_2 - H_{\Omega} \theta_3  \|\\
    &\leq \|H_{\Omega} \theta_1 -  H_{\Omega} \theta_2\| + \|H_{\Omega} \theta_2 -  H_{\Omega} \theta_3\| = \dLH(\theta_1, \theta_2) + \dLH(\theta_2, \theta_3).
    \end{aligned}
    $$

    Finally, while $\|\cdot\|$ is a norm on $\Hmc$, $\dLH$ is not a metric since $H_{\Omega}$ is not a one-to-one mapping (see the discussion around \eqref{eq:non-identifiability}).  
\end{proof}

\section{Dirichlet process Gaussian mixture prior}
\label{sec:dp-mixture-of-Gaussians}



Here we note that Assumptions \ref{a:kl-supports} and \ref{a:prior-sieve} will be satisfied when $\Pi_1$ and $\Pi_2$ are Dirichlet process (DP) mixture of Gaussians \cite[Chapter 5]{ghosal2017fundamentals} and $f_0$,  $g_0$, along with the base measures used for the DP priors, are suitably regular. While results like  \cite[Theorems~7.3 \& 7.15]{ghosal2017fundamentals} provide a starting point to verify these assumptions for general  (non-Gaussian) mixture models, directly specializing \cite[Theorem 7.15]{ghosal2017fundamentals} to a Gaussian location-scale kernel seems to demand restrictive assumptions on the base measure. Instead, by modifying the results in \cite{Canale2017} according to their Remark 2, we establish these assumptions for a DP location-scale mixture of the multivariate Gaussian kernel with covariance restricted to a scalar multiple of the identity matrix. 

In the following, we assume that $\cX = \R^d$ and consider a DP Gaussian mixture prior $\widetilde{\Pi}$ on $\DenX[\R^d]$ corresponding to a draw $f=f_F \sim \widetilde{\Pi}$ given by 
\begin{equation}
f_F(x) \doteq \int \phi_d(x; \mu, \sigma^2 I_d) F(d\mu, d\sigma^2) = \sum_{h=1}^{\infty} W_h \phi_d(x; \mu_h, \sigma^2_h I_d)
\label{eq:stick-breaking-dp-mixture}
\end{equation}
based on a draw $F \sim \DP(\alpha, \nu \times \xi )$, the latter representing a Dirichlet process with concentration parameter $\alpha > 0$ and a base (product) probability measure $\nu \times \xi$ on $\R^d \times (0,\infty)$. Here $\phi_d(x; \mu, \Sigma) \doteq \frac{1}{\sqrt{(2\pi)^d \det \Sigma}} e^{ -\frac{1}{2} (x-\mu)^t \Sigma^{-1} (x-\mu)}$ denotes a general multivariate Gaussian density with mean $\mu \in \R^d$ and a positive definite covariance $\Sigma \in \R^{d \times d}$, and $W_h = V_h \prod_{j < h} (1-V_j)$ with $V_1, V_2, \ldots \stackrel{iid}{\sim} \textrm{Beta}(1, \alpha)$ denote the  weights in the stick-breaking representation $F=\sum_{h=1}^\infty W_h \delta_{(\mu_h, \sigma^2_h)}$ of the Dirichlet process (e.g.~\cite[Theorem~4.12]{ghosal2017fundamentals}) with independent draws $\mu_1,\mu_2, \ldots \iid \nu$ and $\sigma^2_1, \sigma^2_2, \ldots \iid \xi$.

\begin{remark}
Standard choices of the base measures $\nu$ and $\xi$ conjugate to the normal location-scale family considered above  are given by $\nu(d\mu) = \phi_d(\mu; m_0, \Sigma_0) d\mu$ and  $\xi(d\sigma^2) = \frac{b^{a}}{\Gamma(a)}(\sigma^2)^{-a-1} e^{-\frac{b}{\sigma^2}} d\sigma^2$, where the latter is the  Inverse-Gamma density with shape $a > 0$ and scale $b > 0$ (i.e.~$\sigma^2 \sim \xi$ is distributionally equivalent to $(\sigma^2)^{-1} \sim \textrm{Gamma}(a,b)$). As noted in \cite[Example~7.14]{ghosal2017fundamentals}, as long as $a > 1$, the pair $(\nu,\xi)$ will satisfy our  Assumption \ref{a:dp-base-measure} below.
\end{remark}

\begin{assumption}
   The base probability measures $\nu$ on $\R^d$ and $\xi$ on $(0,\infty)$ satisfy the following regularity conditions:
    \begin{enumerate}
    \item $\nu$ and $\xi$ have full support, i.e. $\textrm{supp}(\nu) = \R^d$ and $\textrm{supp}(\xi) = (0,\infty)$,
    \item $\nu$ has exponentially decaying tails, i.e.~there are constants $C, t_0 > 0$ such that $\nu(\R^d \setminus [-t,t]^d) \leq e^{-C t}$ for every $t \geq t_0$, and
    \item $\xi$ has a finite moment $K \doteq \int \sigma^2 \xi(d\sigma^2) < \infty$ and exponentially decreasing lower tail, i.e.~there are constants $c, h_0 > 0$ such that $\xi(\sigma < h ) \leq e^{-c/h^2}$ for all $h \in (0,h_0)$. 
\end{enumerate}
    \label{a:dp-base-measure}
\end{assumption}

Now we can use \cite[Theorem~5]{Wu2008} (along with \cite[Theorem 4.15]{ghosal2017fundamentals}), which states that $\tilde{f}_0 \in \KL(\widetilde{\Pi})$ whenever Assumption \ref{a:density-regularity} below is satisfied. Thus Assumption \ref{a:kl-supports} will be satisfied when  Assumption \ref{a:density-regularity} is satisfied by $\tilde{f}_0 \in \{f_0, g_0\}$ and $\widetilde{\Pi} \in \{\Pi_1, \Pi_2\}$ are chosen to be DP Gaussian mixture priors with base measures satisfying the first item in Assumption \ref{a:dp-base-measure}. 

\begin{assumption} $\tilde{f_0} \in \DenX[\R^d]$ is a continuous and bounded density satisfying:
    \begin{enumerate}
    \item $\int \tilde{f}_0(x) \log \tilde{f}_0(x) dx < \infty$,
    \item $-\int \tilde{f}_0(x) \log \inf_{\|y\| < \delta} \tilde{f}_0(x-y) dx < \infty$ for some $\delta > 0$, and
    \item $\int \|x\|^{2+\eta} \tilde{f}_0(x) dx < \infty$ for some $\eta > 0$.
    \end{enumerate}
    \label{a:density-regularity}
\end{assumption}

Next we turn our attention towards Assumption \ref{a:prior-sieve}. The key results to this end are the following two lemmas, which for a certain sieve, bound its $L^1$ metric entropy (log of the $L_1$ covering number) and the prior mass of its complement. The proofs of these lemmas, based on the application of ideas from \cite{Canale2017,shen2013adaptive} to our setting, are provided in Section \ref{sec:aux-lemmas-dp-mixture-prior}.

\begin{lemma}
Given $H \in \mathbb{N}$ and $M, s, \varepsilon, a >0$, define $\Fms \subseteq \DenX[\R^d]$ as
$$
\begin{aligned}
\Fms &\doteq \left\{ f_F : F = \sum_{h=1}^\infty w_h \delta_{(\mu_h, \sigma_h^2)} \text { with } (w_h)_{h=1}^\infty \in \Delta_{\infty} \text{ such that } \sum_{h>H} w_h \leq \varepsilon, and  \right.\\
&\qquad \qquad \left. \left\|\mu_h\right\| \leq a \text{ and } s \leq \sigma_h \leq s(1+\varepsilon / \sqrt{d})^M \text { for each } h \in\{1, \ldots, H\}\vphantom{F = \sum_{h=1}^H w_h \delta_{(\mu_h, \sigma_h^2)}}\right\}.
\end{aligned}
$$
We then have the $L^1$-metric entropy bound:
$$
\begin{aligned}
\log \Nmc\left(6\varepsilon, \Fms,\|\cdot\|_1\right) \leq  H \log \frac{5}{\varepsilon} + d H 
    \log \frac{3a}{s\varepsilon} +  H \log M
\end{aligned}
$$
whenever $\varepsilon \in (0, 1)$ and $a \geq s > 0$.
\label{lem:dp-mixture-sieve-entropy}
\end{lemma}

\begin{lemma}
    Suppose Assumption \ref{a:dp-base-measure} holds. Then for every  $\varepsilon, \delta \in (0,1)$, there are constants $C, N_0 > 0$ such that 
    \begin{align*}
    \widetilde{\Pi}(\DenX[\R^d] \setminus \Fms_{n}) \leq e^{- C n} \qquad \text{ for each } n \geq N_0
    \end{align*}
    where $\Fms_{n}$ is the sieve defined in Lemma \ref{lem:dp-mixture-sieve-entropy} with $s^{-2} = a = M = n$, $H = \left\lfloor \delta n / \log n\right\rfloor$, and $\varepsilon = \varepsilon$.
    \label{lem:dp-mixture-sieve-prior-mass}
\end{lemma}

Given any $\epsilon > 0$, taking $\varepsilon = \epsilon/6$ and $\delta = \epsilon^2/d$, consider the sieve $\Fms_n$ defined in Lemma \ref{lem:dp-mixture-sieve-prior-mass}. By Lemma \ref{lem:dp-mixture-sieve-entropy}, we have the $L^1$-metric entropy bound
$$
\begin{aligned}
\log \Nmc\left(\epsilon, \Fms_n,\|\cdot\|_1\right) &\leq \frac{\delta n}{\log n} \log \frac{30}{\epsilon} + \frac{d\delta n}{\log n} 
    \left(\frac{\log n}{2} + \log \frac{18}{\epsilon}\right) +  \frac{\delta n}{\log n} \log n\\
    &\leq 2 n \epsilon^2
\end{aligned}
$$
along with the prior mass bound from Lemma \ref{lem:dp-mixture-sieve-prior-mass}
$$
\widetilde{\Pi}(\DenX[\R^d] \setminus \Fms_{n}) \leq e^{- C n}
$$
whenever $n \geq N_0$ for suitable constants $N_0, C > 0$ that are independent of $n$. This verifies Assumption \ref{a:prior-sieve} when $\Pi_1$ and $\Pi_2$ are taken to be DP Gaussian mixture priors like \eqref{eq:stick-breaking-dp-mixture} above, with base measures satisfying Assumption \ref{a:dp-base-measure}.



\subsection{Proof of sieve entropy and prior mass lemmas}
\label{sec:aux-lemmas-dp-mixture-prior}

We start with the following standard lemma bounding the $L^1$-norm difference between two Gaussian kernels.

\begin{lemma}
    For any $\mu_1, \mu_2 \in \R^d$ and $\sigma_1, \sigma_2 \in (0,\infty)$, we have
$$
\begin{aligned}
\|\phi_d(\cdot; \mu_1, \sigma_1^2 I_d) - \phi_d(\cdot; \mu_2, \sigma_2^2 I_d) \|_1 \leq \frac{\|\mu_1-\mu_2\|}{\max(\sigma_1, \sigma_2)} +\frac{2\sqrt{d}|\sigma_1-\sigma_2|}{\min(\sigma_1, \sigma_2)}.
\end{aligned}
$$
    \label{lem:gaussian-kernel-lipschitz}
\end{lemma}
\begin{proof}
    Combining Pinsker's inequality and the Kullback--Leibler divergence between two multivariate Gaussians, one can note (see e.g. \cite[Proposition~2.1]{devroye2018total}):
        \begin{equation}
        \begin{aligned}
        &\|\phi_d(\cdot; \mu_1, \Sigma_1) - \phi_d(\cdot; \mu_2, \Sigma_2) \|_1 = 2\TV(\phi_d(\cdot; \mu_1, \Sigma_1), \phi_d(\cdot; \mu_2, \Sigma_2))\\
        &\leq \sqrt{(\mu_1 - \mu_2)^t \Sigma_2^{-1} (\mu_1 - \mu_2) + \textrm{tr}(\Sigma_2^{-1} \Sigma_1 - I_d) - \log \det (\Sigma_1 \Sigma_2^{-1})}.
        \end{aligned} \label{eq:tv-gaussians}
    \end{equation}

        Without loss of generality, suppose that $\sigma_1 \geq \sigma_2$.
        Using the triangle inequality we can write
        $$
        \begin{aligned}
        &\|\phi_d(\cdot; \mu_1, \sigma_1^2 I_d) - \phi_d(\cdot; \mu_2, \sigma_2^2 I_d) \|_1 \\
        &\leq \|\phi_d(\cdot; \mu_1, \sigma_1^2 I_d) - \phi_d(\cdot; \mu_2, \sigma_1^2 I_d) \|_1 + \|\phi_d(\cdot; \mu_2, \sigma_1^2 I_d) - \phi_d(\cdot; \mu_2, \sigma_2^2 I_d) \|_1.
        \end{aligned}
        $$
        Using \eqref{eq:tv-gaussians} to bound the first term shows
        $$
       \|\phi_d(\cdot; \mu_1, \sigma_1^2 I_d) - \phi_d(\cdot; \mu_2, \sigma_1^2 I_d) \|_1 \leq \frac{\|\mu_1 - \mu_2\|}{\sigma_1}.
        $$
        Using \eqref{eq:tv-gaussians} to bound the second term shows
        $$
        \|\phi_d(\cdot; \mu_2, \sigma_1^2 I_d) - \phi_d(\cdot; \mu_2, \sigma_2^2 I_d) \|_1 \leq \sqrt{d\left(\frac{\sigma_1^2}{\sigma_2^2} - 1 - \log\left(\frac{\sigma_1^2}{\sigma_2^2}\right)\right)}.
        $$
        Now we will use the identity $\log(1+x) \geq \frac{x}{1+x}$ for $x > -1$ to simplify further. Taking $x = \frac{\sigma_1^2}{\sigma_2^2} - 1$ shows
        $$
        \frac{\sigma_1^2}{\sigma_2^2} - 1 - \log\left(\frac{\sigma_1^2}{\sigma_2^2}\right) = x - \log(1+x) \leq x - \frac{x}{1+x} = \frac{x^2}{1+x}.
        $$
        This shows that
        $$
        \begin{aligned}
        \|\phi_d(\cdot; \mu_2, \sigma_1^2 I_d) - \phi_d(\cdot; \mu_2, \sigma_2^2 I_d) \|_1 &\leq \sqrt{d} \frac{|\frac{\sigma_1^2}{\sigma_2^2} - 1|}{\sqrt{\frac{\sigma_1^2}{\sigma_2^2}}} = \sqrt{d} \left|\frac{\sigma_1}{\sigma_2} - \frac{\sigma_2}{\sigma_1}\right|\\
        &=  \frac{\sqrt{d}(\sigma_1 - \sigma_2)(\sigma_1 + \sigma_2)}{\sigma_1\sigma_2} \leq \frac{2\sqrt{d}|\sigma_1 - \sigma_2|}{\sigma_2}.
        \end{aligned}
        $$
\end{proof}

Now we are ready to control the $L^1$ entropy of the sieve from Lemma \ref{lem:dp-mixture-sieve-entropy}.

\begin{proof}[Proof of Lemma \ref{lem:dp-mixture-sieve-entropy}]
    Note first that given any $F = \sum_{h=1}^\infty w_h \delta_{(\mu_h, \sigma_h^2)}$ and $F' = \sum_{h=1}^\infty w_h' \delta_{(\mu_h', \sigma_h'^2)}$ the triangle inequality and $\|f\|_1  = 1$ for $f \in \DenX$ yield
    $$
    \begin{aligned}
    &\|f_F - f_{F'}\|_1 = \left\|\sum_{h=1}^\infty w_h \phi_d(\cdot; \mu_h, \sigma_h^2 I_d) - \sum_{h=1}^\infty w_h' \phi_d(\cdot; \mu_h', \sigma_h'^2 I_d) \right\|_1 \\
    &\leq \sum_{h=1}^H w_h\|\phi_d(\cdot; \mu_h, \sigma_h^2 I_d) - \phi_d(\cdot; \mu_h', \sigma_h'^2 I_d)\|_1 + \sum_{h = 1}^H |w_h - w_h'| + \sum_{h > H} \{w_h + w_h'\}\\
    &\leq \max_{h \in \{1, \ldots, H\}} \left(\frac{\|\mu_h - \mu'_h\|}{\max(\sigma_h, \sigma'_h)} + \frac{2\sqrt{d}|\sigma_h - \sigma'_h|}{\min(\sigma_h, \sigma'_h)} \right) + \sum_{h = 1}^H |w_h - w_h'| + \sum_{h > H} \{w_h + w_h'\},
    \end{aligned}
    $$
    where we have used Lemma \ref{lem:gaussian-kernel-lipschitz} in the last line.

    Let $\varepsilon > 0$ be as given. We now construct a $(6\varepsilon)$-net $\Fmc \subseteq \DenX[\R^d]$ of $\Fms$ in the $L^1$-norm.  For this we define
    \begin{enumerate}
        \item $\Amc \subseteq \R^d$, an $s\varepsilon$-net of the set $\{\mu \in \R^d : \|\mu\| \leq a\}$ of the smallest size. Then $|\Amc| \leq (\frac{3a}{s\varepsilon})^d$ (e.g.~\cite[Proposition~C.2]{ghosal2017fundamentals}).
        \item $\Wmc \subseteq \Delta_{H}$ be a largest $\varepsilon$-dispersed subset of the $H$-dimensional unit simplex $\Delta_{H} = \{(w_1, \ldots, w_H) : w_i \geq 0, \sum_{i=1}^H w_i = 1 \}$ in the $\ell^1$-norm. Then $|\Wmc| \leq \left(\frac{5}{\varepsilon}\right)^{H-1}$ (e.g.~\cite[Proposition~C.1]{ghosal2017fundamentals}).
        \item $\Smc = \{\sigma_1, \ldots, \sigma_M\}$ with $\sigma_m = s(1+\varepsilon / \sqrt{d})^{m-1}$ for each $m \in \{1, \ldots, M\}$.
    \end{enumerate}
    We can use this to define
    $$
    \begin{aligned}
    \Fmc &\doteq \left\{ f_{\tilde{F}} : \tilde{F} = \sum_{h=1}^H \tilde{w}_h \delta_{(\tilde{\mu}_h, \tilde{\sigma}_h^2)} \text { with } (\tilde{w}_h)_{h=1}^H \in \Wmc, \text{ and}   \right.\\
&\qquad \qquad \left. \tilde{\mu}_h \in \Amc, \tilde{\sigma}_h \in \Smc \text{ for each }  h \in\{1, \ldots, H\} \right\}.
    \end{aligned}
    $$
    The proof will be complete once we show that $\Fmc$ is a $(6\varepsilon)$-net of $\Fms$ in the $L^1$-norm since:
    $$
    \begin{aligned}
    |\Fmc| &\leq |\Wmc| \cdot |\Amc|^H \cdot |\Smc|^H \leq \left(\frac{5}{\varepsilon}\right)^{H-1} \left(\frac{3a}{s\varepsilon}\right)^{dH} M^H\\
    &\leq \exp \left\{H \log \frac{5}{\varepsilon} + d H 
    \log \frac{3a}{s\varepsilon} +  H \log M\right\}.
    \end{aligned}
    $$
    To this end, pick an arbitrary $f_F \in \Fms$ with $F = \sum_{h=1}^{\infty} w_h \delta_{(\mu_h, \sigma_h^2)}$ satisfying
    \begin{enumerate}
        \item $\|\mu_h\| \leq a$ for each $h \in \{1, \ldots, H\}$,
        \item $s \leq \sigma_h \leq s(1+\varepsilon / \sqrt{d})^M$ for each $h \in \{1, \ldots, H\}$, and
        \item $\sum_{h > H} w_h \leq \varepsilon$.
    \end{enumerate} 
    For each $h \in \{1, \ldots, H\}$, we can find $\tilde{\mu}_h \in \Amc$ and $\tilde{\sigma}_h^2 \in \Smc$ such that
    $$
    \|\mu_h - \tilde{\mu}_h\| \leq s \varepsilon, \quad \text{and } 1 \leq \frac{\sigma_h}{\tilde{\sigma}_h} \leq  1 + \frac{\varepsilon}{\sqrt{d}}.
    $$
    Further since $\bar{w}_h \doteq w_h / \sum_{j=1}^H w_j$ for $h \in \{1, \ldots, H\}$ defines a probability vector $\bar{w} = (\bar{w}_1, \ldots, \bar{w}_H) \in \Delta_H$, we can find a $\tilde{w} = (\tilde{w}_1, \ldots, \tilde{w}_H) \in \Wmc$ such that $\|\bar{w} - \tilde{w}\|_1 \leq \varepsilon$. This means that:
    $$
    \begin{aligned}
    \sum_{h=1}^H |w_h - \tilde{w}_h| &\leq \sum_{h=1}^H \left|w_h - \bar{w}_h\right| + \sum_{h=1}^H \left|\bar{w}_h - \tilde{w}_h\right|\\
    &= \sum_{h=1}^H w_h \left|1 - \frac{1}{\sum_{j=1}^H w_j}\right| + \sum_{h=1}^H \left|\bar{w}_h - \tilde{w}_h\right|\leq 2 \varepsilon.
    \end{aligned}
    $$
    Thus defining $\tilde{F} \doteq \sum_{h=1}^\infty \tilde{w}_h \delta_{(\tilde{\mu}_h, \tilde{\sigma}_h^2)}$ with $\tilde{w}_h = 0$ for $h > H$ and noting $f_{\tilde{F}} \in \Fmc$, we can use the above bounds along with our initial inequality to note that
    $$    \begin{aligned}
    \|f_F - f_{\tilde{F}}\|_1 &\leq \max_{h \in \{1, \ldots, H\}} \left(\frac{\|\mu_h - \tilde{\mu}_h\|}{\max(\sigma_h, \tilde{\sigma}_h)} + \frac{2\sqrt{d}|\sigma_h - \tilde{\sigma}_h|}{\min(\sigma_h, \tilde{\sigma}_h)} \right) + \sum_{h = 1}^H |w_h - \tilde{w}_h| + \sum_{h > H} \{w_h + \tilde{w}_h\}\\
    &\leq \frac{s \varepsilon}{s} + 2\sqrt{d} \left((1 + \frac{\varepsilon}{\sqrt{d}}) - 1\right) + 2 \varepsilon + \varepsilon = 6 \varepsilon.
    \end{aligned}
    $$
\end{proof}

\begin{proof}[Proof of Lemma \ref{lem:dp-mixture-sieve-prior-mass}] For the sieve $\Fms$ defined in Lemma \ref{lem:dp-mixture-sieve-entropy}, note from \eqref{eq:stick-breaking-dp-mixture} that 
    \begin{equation}
\begin{aligned}
&\widetilde{\Pi}\left(\DenX[\R^d] \setminus \Fms\right)\\
&\leq \operatorname{Pr}\left\{\sum_{h>H} W_h>\varepsilon\right\}+ \sum_{h=1}^H \left\{ \Pr\left(\sigma_h <s\right)+\Pr\left(\sigma_h >s(1+\varepsilon / \sqrt{d})^{M}\right) +\Pr\left(\|\mu_h\|>a\right)\right\},
\end{aligned} \nonumber
\end{equation}
where $W_h = V_h \prod_{j < h} (1-V_j)$ with $V_1, V_2, \ldots \stackrel{iid}{\sim} \textrm{Beta}(1, \alpha)$ denote the  weights in the stick-breaking representation $F=\sum_{h=1}^\infty W_h \delta_{(\mu_h, \sigma^2_h)}$ of the Dirichlet process, with independent draws $\mu_1,\mu_2, \ldots \iid \nu$ and $\sigma^2_1, \sigma^2_2, \ldots \iid \xi$.

    There is a standard trick (e.g.~\cite[Proposition 2]{shen2013adaptive} or \cite[Theorem 7.15]{ghosal2017fundamentals}) to bound the first term, which shows:
    \begin{align*}
    \Pr\left(\sum_{h > H} W_h \geq \varepsilon\right) \leq \left(\frac{e \alpha \log \varepsilon^{-1}}{H}\right)^H.
    \end{align*}

    The remaining terms can now be bounded using Assumption \ref{a:dp-base-measure}. Namely, as long as $s < h_0$ and $a > t_0$, we have
    $$
    \Pr(\sigma_h <s) \leq e^{-c/s^2}, \quad \Pr\left(\sigma_h >s(1+\varepsilon / \sqrt{d})^{M}\right) \leq \frac{K}{s^2(1+\varepsilon / \sqrt{d})^{2M}}, \quad \Pr(\|\mu_h\|>a) \leq e^{-C a}.
    $$
    For the choices $s^{-2} = M = a = n$ and $H = \frac{\delta n}{\log n}$, this leads to the bound
    $$
    \begin{aligned}
    \widetilde{\Pi}\left(\DenX[\R^d] \setminus \Fms_{n}\right) &\leq \left(\frac{e \alpha \log \varepsilon^{-1}}{H}\right)^H + H e^{-c/s^2} + \frac{K H}{s^2(1+\varepsilon / \sqrt{d})^{2M}} + H e^{-C a}\\
    &\leq \left(\frac{e \alpha \log \varepsilon^{-1} \log n}{\delta n}\right)^{\frac{\delta n}{\log n}} + \frac{\delta n}{\log n} e^{-c n} + \frac{K \delta n^2}{\log n (1+\varepsilon / \sqrt{d})^{2n}} + \frac{\delta n}{\log n} e^{-C n}\\
    &= e^{ - \delta n + o(n)} + e^{-c n + o(n)} + e^{-2 n \log (1+\varepsilon / \sqrt{d}) + o(n)} + e^{-C n + o(n)},
    \end{aligned}
    $$
    which completes the proof.
\end{proof}

\section{Data analysis}
\subsection{Tipping point analysis}
When borrowing information for analyses submitted to regulators, it has become common to perform a tipping point analysis to understand the sensitivity of the study conclusions to the amount of borrowing \citep{Best2021Bayesian}. Using the LEAP, it is quite easy to perform this type of analysis by fixing a value for the proportion of exchangeability ($\gamma$) over a range of values. We perform such a tipping point analysis for $\gamma \in \{ 0.0, 0.1, 0.2, \ldots, 1.0 \}$. The results are presented in Figure~\ref{fig:tipping_point}, where we show the posterior mean and 95\% credible interval for the difference in survival probabilities at 12 months. The 95\% credible interval excludes $0$ at $\gamma \in \{0.9, 1.0\}$.
\begin{figure}
    \centering
    \includegraphics[width=0.9\linewidth]{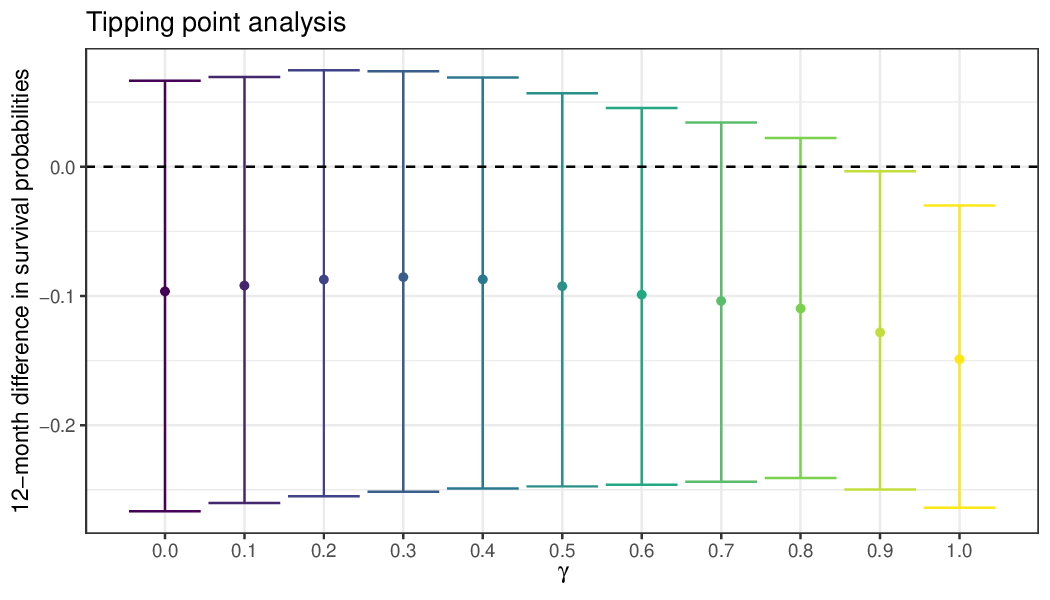}
    \caption{Tipping point analysis for difference in 12-month survival probabilities. The points depict the posterior mean and the bars depict the 95\% credible interval.}
    \label{fig:tipping_point}
\end{figure}

Interestingly, the lower tail of the CI is quite robust to the degree of borrowing while the upper tail shifts more dramatically particularly when $\gamma > 0.5$. We conjecture that this is because, as shown in Figure~\ref{fig:km_curve}, there is some separation between the historical control data in a neighborhood less than 12 months while congruence is observed in a neighborhood greater than 12 months. Indeed, changing the estimand to 24 months (data not shown) shows that the lower bound also changes.

\begin{longtable}{>{\raggedright\arraybackslash}p{4.4cm}
                  >{\raggedright\arraybackslash}p{6.6cm}
                  >{\raggedright\arraybackslash}p{3.5cm}}
\caption{Notation used in the paper. Sections~E--F use deliberately
         different notation for the consistency proof; those symbols are
         listed separately at the end.}
\label{tab:notation}\\
\toprule
\textbf{Symbol} & \textbf{Description} & \textbf{First used} \\
\midrule
\endfirsthead
\multicolumn{3}{c}{\tablename~\thetable{} (continued)}\\
\toprule
\textbf{Symbol} & \textbf{Description} & \textbf{First used} \\
\midrule
\endhead
\midrule
\multicolumn{3}{r}{\textit{Continued on next page}}\\
\endfoot
\bottomrule
\endlastfoot

\multicolumn{3}{l}{\textit{Data}}\\
\midrule
$n$ & Number of current-trial individuals & \S 2 \\
$n_0$ & Number of historical individuals & \S 2 \\
$y_i$ & $i$-th current-trial outcome, $i=1,\ldots,n$ & \S 2 \\
$y_{0j}$ & $j$-th historical outcome, $j=1,\ldots,n_0$ & \S 2 \\
$\bm{y},\,\bm{y}_0$ & Vectors of current and historical outcomes & \S 2 \\

\addlinespace[4pt]
\multicolumn{3}{l}{\textit{Latent variables}}\\
\midrule
$\epsilon_{0j}\in\{0,1\}$ &
    Exchangeability indicator for historical individual $j$:
    $\epsilon_{0j}=1$ (exchangeable with current data),
    $\epsilon_{0j}=0$ (nonexchangeable);
    $\epsilon_{0j}\overset{\mathrm{ind}}{\sim}\mathrm{Bernoulli}(\gamma)$ &
    \S 2 \\
$z_i\in\{1,2,\ldots\}$ &
    Cluster assignment for current-data individual $i$;
    always an exchangeable cluster & \S 2 \\
$z_{0j}\in\{1,2,\ldots\}$ &
    Cluster assignment for historical individual $j$;
    indexes exchangeable clusters $k$ if $\epsilon_{0j}=1$,
    nonexchangeable clusters $l$ if $\epsilon_{0j}=0$ & \S 2 \\

\addlinespace[4pt]
\multicolumn{3}{l}{\textit{Mixture weight and prior on $\gamma$}}\\
\midrule
$\gamma\in[0,1]$ &
    Probability that a historical individual is exchangeable
    with the current data & \S 2 \\
$p_0\in[0,1]$ &
    Prior point-mass weight at $\gamma=1$;
    $\pi(\gamma)=p_0\,\delta_1(\gamma)+(1-p_0)\,
    \mathrm{Beta}(\gamma\mid c_1,c_0)$ & \S 2 \\
$c_1,\,c_0>0$ &
    Beta hyperparameters for the continuous component of $\pi(\gamma)$;
    $c_1$ governs tendency toward full exchangeability ($\gamma=1$),
    $c_0$ toward no exchangeability ($\gamma=0$);
    $c_1=c_0=1$ gives a uniform prior & \S 2 \\

\addlinespace[4pt]
\multicolumn{3}{l}{\textit{Exchangeable component (cluster index $k$)}}\\
\midrule
$f(y\mid\bm{\theta})$ &
    Kernel density; $\bm{\theta}$ collects cluster parameters & \S 2 \\
$\mu\sim\mathrm{DP}(\alpha,P_0)$ &
    Random mixing measure for the exchangeable component & \S 2 \\
$f_\mu(y)=\int f(y\mid\bm{\theta})\,\mu(d\bm{\theta})$ &
    Induced exchangeable mixture density & \S 2 \\
$\alpha>0$ &
    Concentration parameter of the exchangeable DP & \S 2 \\
$P_0$ &
    Base measure of the exchangeable DP;
    prior on each $\bm{\theta}_k$ & \S 2 \\
$v_k\overset{\mathrm{iid}}{\sim}\mathrm{Beta}(1,\alpha)$ &
    Stick-breaking variables for the exchangeable weights,
    $k=1,\ldots,K-1$ & \S 2 \\
$w_k=v_k\prod_{m<k}(1-v_m)$ &
    Weights for the exchangeable mixture;
    $\bm{w}\sim\mathrm{GEM}(\alpha)$ & \S 2 \\
$\bm{\theta}_k\overset{\mathrm{iid}}{\sim}P_0$ &
    Parameters for exchangeable cluster $k$ & \S 2 \\
$K$ &
    Number of occupied exchangeable clusters
    (random; updated each MCMC step) & Supp.~\S\S B--C \\
$a_\alpha,\,b_\alpha$ &
    Shape and rate of the Gamma hyperprior on $\alpha$ & Supp.~\S C \\

\addlinespace[4pt]
\multicolumn{3}{l}{\textit{Nonexchangeable component (cluster index $l$)}}\\
\midrule
$g(y\mid\bm{\lambda})$ &
    Kernel density; $\bm{\lambda}$ collects cluster parameters & \S 2 \\
$\nu\sim\mathrm{DP}(\eta,Q_0)$ &
    Random mixing measure for the nonexchangeable component & \S 2 \\
$g_\nu(y)=\int g(y\mid\bm{\lambda})\,\nu(d\bm{\lambda})$ &
    Induced nonexchangeable mixture density & \S 2 \\
$\eta>0$ &
    Concentration parameter of the nonexchangeable DP & \S 2 \\
$Q_0$ &
    Base measure of the nonexchangeable DP;
    prior on each $\bm{\lambda}_l$ & \S 2 \\
$\psi_l\overset{\mathrm{iid}}{\sim}\mathrm{Beta}(1,\eta)$ &
    Stick-breaking variables for the nonexchangeable weights,
    $l=1,\ldots,L-1$ & \S 2 \\
$\rho_l=\psi_l\prod_{m<l}(1-\psi_m)$ &
    Weights for the nonexchangeable mixture;
    $\bm{\rho}\sim\mathrm{GEM}(\eta)$ & \S 2 \\
$\bm{\lambda}_l\overset{\mathrm{iid}}{\sim}Q_0$ &
    Parameters for nonexchangeable cluster $l$ & \S 2 \\
$L$ &
    Number of occupied nonexchangeable clusters
    (random; updated each MCMC step) & Supp.~\S\S B--C \\
$a_\eta,\,b_\eta$ &
    Shape and rate of the Gamma hyperprior on $\eta$ & Supp.~\S C \\

\addlinespace[4pt]
\multicolumn{3}{l}{\textit{Class-level counts}}\\
\midrule
$N_0=\sum_{j=1}^{n_0}\epsilon_{0j}$ &
    Number of exchangeable historical individuals & Supp.~\S\S B--C \\
$M_0=\sum_{j=1}^{n_0}(1-\epsilon_{0j})$ &
    Number of nonexchangeable historical individuals;
    $N_0+M_0=n_0$ & Supp.~\S\S B--C \\

\addlinespace[4pt]
\multicolumn{3}{l}{\textit{Cluster-level counts}}\\
\midrule
$n_k=\sum_{i=1}^{n}I(z_i=k)$ &
    Current-data individuals in exchangeable cluster $k$ &
    Supp.~\S B \\
$N_{0k}=\sum_{j=1}^{n_0}I(\epsilon_{0j}=1,\,z_{0j}=k)$ &
    Exchangeable historical individuals in cluster $k$ &
    Supp.~\S B \\
$\tilde{N}_k=n_k+N_{0k}$ &
    Total individuals (current + exchangeable historical)
    in exchangeable cluster $k$ & Supp.~\S C \\
$M_{0l}=\sum_{j=1}^{n_0}I(\epsilon_{0j}=0,\,z_{0j}=l)$ &
    Nonexchangeable historical individuals in cluster $l$ &
    Supp.~\S B \\

\addlinespace[4pt]
\multicolumn{3}{l}{\textit{Cluster datasets}}\\
\midrule
$D^{\mathrm{exch}}_k$ &
    $\{y_i:z_i=k\}\cup\{y_{0j}:\epsilon_{0j}=1,\,z_{0j}=k\}$;
    data for exchangeable cluster $k$ & Supp.~\S C \\
$D^{\mathrm{unexch}}_l$ &
    $\{y_{0j}:\epsilon_{0j}=0,\,z_{0j}=l\}$;
    data for nonexchangeable cluster $l$ & Supp.~\S C \\

\addlinespace[4pt]
\multicolumn{3}{l}{\textit{ANOVA DDP instantiation (\S 2.7 and Supp.~\S C.8)}}\\
\midrule
$y_i=\min\{\log t_i,\log c_i\}$ &
    Log observed time (event or censoring) & \S 2.7 \\
$a_i\in\{0,1\}$ &
    Treatment arm indicator ($0=$ control, $1=$ treatment) & \S 2.7 \\
$\bm{\theta}_k=(\bm{\beta}_k',\tau_k)'$ &
    Exchangeable cluster parameters: arm-specific log-normal
    intercepts $\bm{\beta}_k$ and precision $\tau_k$ & \S 2.7 \\
$\bm{\mu}_0,\,\bm{\Sigma}_0$ &
    Mean and covariance of the Normal base measure $P_0$
    for $\bm{\beta}_k$; elicited from the log-normal AFT MLE & Supp.~\S C.8 \\
$\delta_\alpha,\,\kappa_\alpha$ &
    Shape and rate of the Gamma base measure $P_0$ for $\tau_k$ &
    Supp.~\S C.8 \\

\addlinespace[6pt]
\multicolumn{3}{l}{\textit{Sections~E--F only: posterior consistency proof
    (deliberately different notation)}}\\
\midrule
$\theta=(f,g,p)\in\Theta$ &
    Abstract LEAP parameter: current density $f$,
    nonexchangeable density $g$, mixture proportion $p$.
    Here $f,g$ are \emph{densities}, not the kernels
    $f(\cdot\mid\bm{\theta})$, $g(\cdot\mid\bm{\lambda})$
    of Sections~A--D & Supp.~\S E \\
$\Theta=\mathscr{D}(\mathcal{X})^2\times[0,1]$ &
    Parameter space for the consistency proof & Supp.~\S E \\
$d_H(f,g)$ & Hellinger distance between densities & Supp.~\S E \\
$d_{\Omega,H}(\theta_1,\theta_2)$ &
    Weighted pseudo-metric on $\Theta$; $\Omega\in[0,1]$ is
    the limiting fraction of current observations & Supp.~\S E \\
$\iota_N\doteq n/N,\;N=n+n_0$ &
    Fraction of total observations that are current;
    $\iota_N\to\Omega$ as $N\to\infty$ & Supp.~\S E \\
$\Pi=\Pi_1\otimes\Pi_2\otimes\Pi_3$ &
    Product prior on $\Theta$ & Supp.~\S E \\
$\theta_0=(f_0,g_0,p_0)$ &
    True data-generating parameter & Supp.~\S E \\

\addlinespace[6pt]
\multicolumn{3}{l}{\textit{Disambiguation notes}}\\
\midrule
\multicolumn{3}{p{14.5cm}}{%
\textbf{(1) Index letter determines component throughout
Sections~A--D.}
The letter $k$ always indexes exchangeable clusters and the letter
$l$ always indexes nonexchangeable clusters.
Count symbols follow the same rule: $n_k$, $N_{0k}$, $\tilde{N}_k$
are always exchangeable-component quantities;
$M_{0l}$ is always a nonexchangeable-component quantity.
No $F$/$G$ subscripts on count quantities are needed.
}\\[3pt]
\multicolumn{3}{p{14.5cm}}{%
\textbf{(2) $c_1$ and $c_0$ mirror $\epsilon_{0j}$.}
The Beta hyperparameters $c_1$ (exchangeable) and $c_0$ (nonexchangeable)
are subscripted by the values of $\epsilon_{0j}$ they govern,
making the prior $\pi(\gamma)\propto\gamma^{c_1-1}(1-\gamma)^{c_0-1}$
immediately readable.
}\\[3pt]
\multicolumn{3}{p{14.5cm}}{%
\textbf{(3) $P_0$ and $Q_0$ are base measures, not CDFs.}
$P_0$ is the base measure of the exchangeable DP (the prior on
each $\bm{\theta}_k$) and $Q_0$ is the base measure of the
nonexchangeable DP (the prior on each $\bm{\lambda}_l$).
Neither is the abstract CDF $F$ or $G$ from Proposition~1 /
Supplement~\S A.
}\\[3pt]
\multicolumn{3}{p{14.5cm}}{%
\textbf{(4) Section~E uses $\iota_N$ for the current-data fraction.}
The symbol $\iota_N=n/(n+n_0)$ in Sections~E--F plays the role
that the ratio $n/N$ plays in the proof of posterior consistency.
It is distinct from all parameters in Sections~A--D.
}\\
\end{longtable}

\bibliography{03_refs}